\documentclass[amsthm]{elsart}
\usepackage{yjsco}

\usepackage{amssymb, amsbsy, amsfonts}
\usepackage{amsmath, amsthm}
\usepackage{latexsym,float,color}
\usepackage{makeidx}         
\usepackage{multicol}

\usepackage{etex, xy}
\xyoption{all}

\numberwithin{equation}{section}

\newcommand{\PPowers}[3]{{#3}^{#1}_{#2}}

\newcommand{\ssum}[2]{\text{\small$\displaystyle\sum_{#1}^{#2}$}}

\newcommand{\ev}{{\rm ev}}
\newcommand{\Shift}{{\mathcal S}}
\newcommand{\seqP}[1]{{S}(#1)}
\newcommand{\expr}{\operatorname{expr}}
\newcommand{\mmod}{\operatorname{mod}}

\newcommand{\notion}[1]{\emph{#1}}

\DeclareMathOperator*{\tsum}{{\textstyle\sum}}
\DeclareMathOperator*{\tprod}{{\textstyle\prod}}

\let\set\mathbb
\def\vect#1{\mathbf{#1}}
\def\lcm{\operatorname{lcm}}

\def\seqK{{\cal S}(\KK)}

\def\ord{\operatorname{ord}}

\def\lm{\operatorname{lm}}

\def\lc{\operatorname{lc}}

\def\rE{$R$}
\def\piE{$\Pi$}

\def\rpisiSE{$R\Pi\Sigma^*$}
\def\pisiSE{$\Pi\Sigma^*$}
\def\sigmaSE{$\Sigma^*$}

\def\KK{\set K}
\def\NN{\set N}
\def\ZZ{\set Z}
\def\GG{\set G}
\def\LL{\set L}
\def\HH{\set H}
\def\FF{\set F}
\def\EE{\set E}
\def\QQ{\set Q}
\def\AA{\set A}
\def\SA{\set S}

\newcommand{\fct}[3]{#1\colon #2 \to #3}
\newcommand{\dfield}[2]{({#1},{#2})}
\newcommand{\const}[2]{{\rm const}{(#1,#2)}}

\newcommand{\lr}[1]{\langle #1\rangle}
\newcommand{\ltr}[1]{[#1,\tfrac{1}{#1}]}

\newdimen\listablecorrection
\newtheorem{theorem}{Theorem}[section]
\newtheorem{proofstep}[theorem]{Proof}
\newtheorem{proposition}[theorem]{Proposition}
\newtheorem{corollary}[theorem]{Corollary}
\newtheorem{lemma}[theorem]{Lemma}
\newtheorem{remark}[theorem]{Remark}
\newtheorem{definition}[theorem]{Definition}
\newtheorem{example}[theorem]{Example}

\makeindex

\begin{document}

\begin{frontmatter}

\title{Summation Theory II: Characterizations of $\boldsymbol{R\Pi\Sigma^*}$-extensions and algorithmic aspects}

\thanks{Supported by the Austrian Science Fund (FWF) grant SFB F50 (F5009-N15).}

\author{Carsten Schneider}
\address{Research Institute for Symbolic Computation (RISC)\\ 
Johannes Kepler University Linz\\
Altenbergerstra{\ss}e 69, 4040 Linz, Austria}
\ead{Carsten.Schneider@risc.jku.at}

\begin{abstract}
Recently, \rpisiSE-extensions have been introduced which extend Karr's \pisiSE-fields substantially:
one can represent expressions not only in terms of transcendental sums and products, but one can work also with products
over primitive roots of unity. Since one can solve the parameterized telescoping problem in such rings, covering as special cases the summation paradigms of telescoping and creative telescoping, one obtains a rather flexible toolbox for symbolic summation. This article is the continuation of this work. Inspired by Singer's Galois theory of difference equations we will work out several alternative characterizations of \rpisiSE-extensions: adjoining naively sums and products leads to an \rpisiSE-extension iff the obtained difference ring is simple iff the ring can be embedded into the ring of sequences iff the ring can be given by the interlacing of \pisiSE-extensions. From the viewpoint of applications this leads to a fully automatic machinery to represent indefinite nested sums and products in such \rpisiSE-rings. In addition, we work out how the parameterized telescoping paradigm can be used to prove algebraic independence of indefinite nested sums. Furthermore, one obtains an alternative reduction tactic to solve the parameterized telescoping problem in basic \rpisiSE-extensions exploiting the interlacing property.
\end{abstract}

\begin{keyword}
difference ring extensions\sep roots of unity\sep indefinite nested sums and products\sep simple difference rings\sep difference ideals \sep constants \sep interlacing of difference rings \sep embedding into the difference ring of sequences\sep Galois theory
\end{keyword}

\end{frontmatter}

\section{Introduction}\label{Sec:Introduction}

The foundation of the difference field approach was laid
in~\cite{Karr:81,Karr:85}. Karr's summation algorithm enables one to solve the parameterized telescoping problem within a given rational function field $\FF$ in which indefinite nested sums and products are represented by variables and the shift-behaviour of the objects is modelled by a field automorphism $\fct{\sigma}{\FF}{\FF}$. As observed in~\cite{Schneider:00} Karr's algorithm covers also the creative telescoping paradigm of definite summation~\cite{Zeilberger:91}.
Various improvements and extensions of such a difference field $\dfield{\FF}{\sigma}$, also called a \pisiSE-field, have been incorporated into a strong summation theory within the last 15 years. These significant refinements enable one to find optimal product and sum representations by minimizing the nesting depth~\cite{Schneider:05f,Schneider:08c,Schneider:10b}, by minimizing the number of sums arising within the summand~\cite{Schneider:04a,Schneider:15}, or minimizing the degree of the sums and products arising in the summands and multiplicands~\cite{Abramov:75,Paule:95,Schneider:05c,Schneider:07d,Petkov:10}.

A central drawback of this field approach is the lack of treatment of the alternating sign $(-1)^n$ which is one of the central building blocks in many summation problems. This object, and more generally $\alpha^n$ with a primitive root of unity $\alpha$, cannot be treated in a field or an integral domain: zero-divisors such as the factors of $(1-(-1)^n)(1+(-1)^n)=0$ are introduced.  
In order to overcome this vulnerability, a new difference ring theory has been elaborated~\cite{Schneider:16a,DR2}. The so-called \rpisiSE-extensions have been introduced in which indefinite nested sums and products together with $\alpha^n$  can be represented and in which the parameterized telescoping problem can be solved by efficient algorithms. 

The aim of this article is twofold. First, we will provide new insight in the theory of \rpisiSE-extensions. In particular we will gain important characterizations and structural theorems for \rpisiSE-extensions. Second, we will emphasize the algorithmic nature of these results in order to push forward the difference ring approach to symbolic summation.\\
Here we will restrict our attention to the so-called basic \rpisiSE-extensions~\cite{Schneider:16a,DR2}. More precisely, we start with a difference field $\dfield{\FF}{\sigma}$ (e.g., any of Karr's \pisiSE-fields) with the set of constants $\KK=\const{\FF}{\sigma}=\{c\in\FF|\sigma(c)=c\}$.
Then we extend this field to a difference ring $\dfield{\EE}{\sigma}$ by a tower of ring extensions of the following form (see Prop.~\ref{Prop:AltDefOfBasic} below):

\smallskip

\begin{itemize}
\item algebraic ring extensions ($A$-extensions) which introduce objects of the type $\alpha^n$ with $\alpha\in\KK$ being a root of unity;
\item product extensions ($P$-extensions) which introduce nested products in terms of Laurent polynomials where the multiplicands are invertible and the multiplicands are free of the defining generators of $A$-extensions;
\item sum extensions ($S$-extensions) which introduce nested sums in terms of polynomials.
\end{itemize}

\smallskip

\noindent A tower of such extensions will be called a basic $APS$-extension. Finally, we impose that during the construction the constants remain unchanged, i.e., $\const{\EE}{\sigma}=\{c\in\EE|\sigma(c)=c\}=\KK$. In this case, a basic $APS$-extension will be called a basic \rpisiSE-extension. 

We will supplement the difference ring theory of~\cite{Schneider:16a} with the following alternative characterizations of a basic \rpisiSE-extension under the assumption that $\dfield{\FF}{\sigma}$ is constant-stable (see Definition~\ref{Def:ConstantStable}). More precisely, for a given tower of basic $APS$-extensions $\dfield{\EE}{\sigma}$ over the ground field $\dfield{\FF}{\sigma}$ the following statements are equivalent:

\smallskip

\begin{enumerate}
 \item[(1)] It forms a basic \rpisiSE-extension, i.e., $\const{\EE}{\sigma}=\const{\FF}{\sigma}$.
 \item[(2)] $\dfield{\EE}{\sigma}$ is simple, i.e., any ideal of $\EE$ which is closed under $\sigma$ is either $\EE$ or $\{0\}$.
 \item[(3)] $\dfield{\EE}{\sigma}$ can be composed by the interlacing of \pisiSE-extensions.
 \item[(4)] $\dfield{\EE}{\sigma}$ can be embedded into the ring of sequences provided that $\dfield{\FF}{\sigma}$ can be embedded into the ring of sequences.
 \end{enumerate}

\smallskip

We emphasize that similar statements have been derived in the setting of Picard-Vessiot rings in~\cite{Singer:97}; see also~\cite{Singer:08}. However, much stronger assumptions are imposed in this setting, like assuming that the constant field $\KK$ is algebraically closed or that the algebraic closure of the ground field $\FF$ can be embedded into the ring of sequences. Both avenues, the Picard-Vessiot approach~\cite{Singer:97} and the basic \rpisiSE-approach, cover the class of d'Alembertian solutions~\cite{Abramov:94,Abramov:96}, the former from a very general and abstract viewpoint and the latter from  a very algorithmic viewpoint. In other respects, they complement each other non-trivially: Picard-Vessiot rings can represent all the solutions of linear homogeneous recurrence relations, which is not possible by \rpisiSE-extensions, and vice versa, basic \rpisiSE-extensions can represent, e.g., nested products, which cannot be formulated in Picard-Vessiot rings.

The worked out machinery for the embedding of basic \rpisiSE-extensions into the ring of sequences 
is algorithmic. It has been implemented within the summation package \texttt{Sigma}~\cite{Schneider:07a,Schneider:13a} and provides a solution to the following problem: given an expression in terms of nested sums and products, find an alternative expressions such that the set of arising sums and products (up to products with roots of unity) is algebraically independent among each other. This toolbox produces very compact representations of nested sum expressions and plays a central role for large scale calculations in particle physics; see~\cite{Schneider:16b} and references therein.
Moreover, following the ideas of~\cite{Schneider:10c} we will show that the parameterized telescoping problem enables one to show the algebraic independence of sums whose summands are expressible within an \rpisiSE-extension. Finally, exploiting the interlacing property of an \rpisiSE-extension we will present a new reduction technique that solves the parameterized telescoping problem in an \rpisiSE-extension or in its total ring of fractions.

The outline of the article is as follows. In Section~\ref{Sec:APS} we will introduce basic $APS$-extensions and \rpisiSE-extensions. In addition, we summarize the crucial (algorithmic) properties of \rpisiSE-extensions introduced in~\cite{Schneider:16a}, and we present simple criteria to verify if an $A$-extension is an \rE-extensions. 
The above equivalences between (1)--(4) are shown in Sections~\ref{Sec:SimpleRings}--\ref{Sec:Embedding}, respectively. Finally, the applications mentioned in the previous paragraph will be worked out in Sections~\ref{Sec:Telescoping} and~\ref{Sec:Application}. A conclusion is given in Section~\ref{Sec:Conclusion}.

\section{Definitions and basic properties of
$APS$- and \rpisiSE-extensions}\label{Sec:APS}

In this article we will use a ring $\AA$ to represent the given summation objects and a difference ring automorphism $\fct{\sigma}{\AA}{\AA}$ to model the shift behaviour of the summation objects by the appropriate action on the ring elements. In this regard, the following conventions and definitions will be used throughout this article.\\
$\bullet$ \textit{Basics:} $\ZZ$ and $\NN$ denote the set of integers and non-negative integers, respectively.\\ 
$\bullet$ \textit{Rings:} By a ring we always mean a commutative ring with unity. In addition, all rings (resp.\ fields) contain the rational numbers as a subring (resp.\ subfield). $\AA^*$ is the set of units which is a multiplicative group. $Q(\AA)$ denotes the total ring of fractions. If $\AA$ is integral, $Q(\AA)$ is the field of fractions of $\AA$. For a ring element $f\in\AA[z]$ ($z$ is transcendental or algebraic over $\AA$) and $a\in\AA[z]$ we write $f(a)=f|_{z\to a}$.\\
$\bullet$ \textit{Difference rings:} 
A difference ring (resp.\ field) $\dfield{\AA}{\sigma}$ is a ring (resp.\ field) equipped with a ring (resp.\ field) automorphism $\fct{\sigma}{\AA}{\AA}$.  In this regard, the set of constants
$$\KK=\const{\AA}{\sigma}=\{c\in\AA|\sigma(c)=c\}$$
plays a central role. Note that $\KK$ forms a subring of $\AA$. In particular, it is a subfield of $\AA$ if $\AA$ is a field. Moreover, observe that $\QQ$ is always contained as a subring (resp.\ subfield) in $\KK$ (since we assume that $\QQ$ is contained in $\AA$).\\
$\bullet$ \textit{Extensions:}
For the construction of difference rings (resp.\ fields) we need the concept of difference ring (resp.\ field) extensions. A difference ring (resp.\ field) $\dfield{\EE}{\sigma'}$ is a difference ring (resp.\ field) extension of $\dfield{\AA}{\sigma}$ if $\AA$ is a subring (resp.\ subfield) of $\EE$ and $\sigma'|_{\AA}=\sigma$, i.e., $\sigma'(f)=\sigma(f)$ for all $f\in\AA$. Since $\sigma$ and $\sigma'$ agree on $\AA$, we do not distinguish between them anymore. Note that $\dfield{\AA}{\sigma}$ is a difference ring (resp.\ field) extension of $\dfield{\KK}{\sigma}$.\\
$\bullet$ \textit{Ideals:} Let $\dfield{\AA}{\sigma}$ be a difference ring. An ideal in $\AA$ generated by the elements $f_1,\dots,f_r\in\AA$ is denoted by
$\langle f_1,\dots,f_r\rangle$;
we set $\langle\,\rangle=\{0\}$. An ideal $I$ in $\AA$ is maximal if there are no other ideals contained between $I$ and $\AA$.
$I$ is a difference ideal in $\dfield{\AA}{\sigma}$ if for any $f\in I$ we have that $\sigma(f)\in\AA$.
$I$ is called reflexive if also $\sigma^{-1}(f)\in I$ holds for all $f\in I$. Note that for a given reflexive difference ideal $I$ we obtain a ring automorphism $\fct{\sigma'}{\AA/I}{\AA/I}$ by defining $\sigma'(a+I)=\sigma(a)+I$ for all $a\in\AA$. A difference ring $\dfield{\AA}{\sigma}$ is simple if there is no difference ideal $I$ in $\AA$, except $I=\{0\}$ and $I=\AA$.\\ 
$\bullet$ \textit{Embeddings:} A difference ring homomorphism $\fct{\tau}{\AA_1}{\AA_2}$ between two difference
rings $\dfield{\AA_1}{\sigma_1}$ and $\dfield{\AA_2}{\sigma_2}$ is a
ring homomorphism such that $\tau(\sigma_1(f))=\sigma_2(\tau(f))$ holds for all
$f\in\AA_1$. If $\tau$ is injective, we call $\tau$ a difference ring monomorphism or a difference ring embedding. In this case, $\dfield{\AA_2}{\sigma}$ is a difference ring extension of $\dfield{\tau(\AA_1)}{\sigma}$ where $\dfield{\AA_1}{\sigma}$ and $\dfield{\tau(\AA_1)}{\sigma}$ are the same up to the renaming of the objects with $\tau$. If $\tau$ is bijective, $\tau$ is also called a difference ring isomorphism and we say that the difference rings $\dfield{\AA_1}{\sigma}$ and $\dfield{\AA_2}{\sigma}$ are isomorphic; we also write $\dfield{\AA_1}{\sigma}\simeq\dfield{\AA_2}{\sigma}$.

\smallskip

\noindent For further details on these notions and properties we refer to~\cite{Cohn:65,Levin:08}.

\medskip

In Subsection~\ref{Sec:APSExt} we will present $APS$-extensions to rephrase indefinite nested sums and products in difference rings. In order to obtain a precise description, we will introduce the subclass of \rpisiSE-extensions in Subsection~\ref{Sec:RPSExt}, and we will restrict this class further to basic \rpisiSE-extensions in Subsection~\ref{Sec:BasicExt}. In this latter class there are algorithms available~\cite{Schneider:16a} which can be used to construct such a tower of extensions automatically. We remark that this class covers all the types of nested sums and products, like~\cite{Bluemlein:99,Vermaseren:99,ABS:11,Moch:02,ABS:13,ABRS:14},
that we have encountered in practical problem solving so far. Some further basic properties of $A$-extensions to be used later will be elaborated in Subsection~\ref{Sec:SimpleTests}.

\subsection{$APS$-extensions}\label{Sec:APSExt}

We start with $S$-extensions and $P$-extensions that enable one to represent iterated sums and products in a naive way.
Let $\dfield{\AA}{\sigma}$ be a difference ring (in which sums and products have already been defined by previous extensions). Now let $\beta\in\AA$ and take the ring of polynomials $\AA[t]$ (i.e., $t$ is transcendental over $\AA$). Then there is a unique difference ring extension $\dfield{\AA[t]}{\sigma}$ of $\dfield{\AA}{\sigma}$ with $\sigma(t)=t+\beta$. Such an extension is called a \notion{sum-extension} (in short \notion{$S$-extension}); the generator $t$ is also called an \notion{$S$-monomial}.\\  
Similarly, take a unit $\alpha\in\AA^*$ and take the ring of Laurent polynomials $\AA\ltr{t}$ (i.e., $t$ is transcendental over $\AA$). Then there is a unique difference ring extension $\dfield{\AA\ltr{t}}{\sigma}$ of $\dfield{\AA}{\sigma}$ with $\sigma(t)=\alpha\,t$. Such an extension is called a \notion{product-extension} (in short \notion{$P$-extension}); the generator $t$ is also called a \notion{$P$-monomial}.\\
Of special interest is the case when $\AA$ is a field and $\AA(t)$ is a rational function field (i.e., $t$ is transcendental over $\AA$). Let $\alpha\in\AA^*$ and $\beta\in\AA$. Then there is a unique difference field extension $\dfield{\AA(t)}{\sigma}$ of $\dfield{\AA}{\sigma}$ with $\sigma(t)=\alpha\,t+\beta$. If $\alpha=1$, this extension is called an \notion{$S$-field extension} and $t$ is called an \notion{$S$-monomial}. Similarly, if $\beta=0$, this extension is called a \notion{$P$-field extension} and $t$ is called a \notion{$P$-monomial}. Note that for $\alpha=1$ we get the chain of extensions
$\dfield{\AA}{\sigma}\leq\dfield{\AA[t]}{\sigma}\leq\dfield{\AA(t)}{\sigma}$
and if $\beta=0$ we get the chain of extensions
$\dfield{\AA}{\sigma}\leq\dfield{\AA\ltr{t}}{\sigma}\leq\dfield{\AA(t)}{\sigma}.$

Finally, we will introduce $A$-extensions to model algebraic objects like $\alpha^k$ for a root of unity $\alpha$. More precisely, let $\lambda\in\NN$ with $\lambda>1$, let $\dfield{\AA}{\sigma}$ be a difference ring and let $\alpha\in\AA^*$ be a $\lambda$th root of unity, i.e., $\alpha^{\lambda}=1$. Now take the difference ring extension
$\dfield{\AA[z]}{\sigma}$ of $\dfield{\AA}{\sigma}$ with $z$ being transcendental over $\AA$ and $\sigma(z)=\alpha\,z$ (again this construction is unique).
Next,
take the ideal $I:=\lr{z^{\lambda}-1}$ and consider the ring $\EE=\AA[z]/I$.
Since $I$ is closed under $\sigma$ and $\sigma^{-1}$, i.e., $I$ is a reflexive difference ideal, it follows that $\fct{\sigma}{\EE}{\EE}$ with
$\sigma(f+I)=\sigma(f)+I$
is a ring automorphism. In other words, $\dfield{\EE}{\sigma}$ is a
difference ring. Moreover, there is the natural embedding of $\AA$ into $\EE$
with
$a\mapsto a+I$.
By identifying $a$ with $a+I$, $\dfield{\EE}{\sigma}$ forms a difference ring extension of $\dfield{\AA}{\sigma}$. Finally, by setting $y:=z+I$ we get the difference ring extension $\dfield{\AA[y]}{\sigma}$ of $\dfield{\AA}{\sigma}$ subject to the relation $y^{\lambda}=1$. 
Note that this ring contains zero-divisors such as the factors of 
$(y-1)(1+y+\dots+y^{\lambda-1})=0.$
By construction we have that $\min\{i>0|y^i=1\}=\lambda$. Moreover, for any polynomial $f(z)\in\AA[z]\setminus\{0\}$ with $f(y)=0$ we have that $\deg(f)\geq\lambda$. 
This extension is called an \notion{algebraic extension} (in short \notion{$A$-extension}) \notion{of order $\lambda$}; the generator $y$ is also called an \notion{$A$-monomial} and we define $\ord(y)=\lambda$.
Note that $y$ with the relations $y^{\lambda}=1$ and $\sigma(y)=\alpha\,y$ imitates  $\alpha^k$ with the relations $(\alpha^k)^{\lambda}=1$ and $\alpha^{k+1}=\alpha\,\alpha^k$, respectively. 

\begin{example}\label{Exp:MainDRNaiveDef}
We illustrate with our brute force tactic how the product-sum expressions  
\begin{equation}\label{Equ:EExpr}
\begin{split}
k,\quad(-1)^k,\quad\,2^k,\quad\binom{n}{k}=\prod_{i=1}^k\frac{n-i+1}{i},\quad E_1(k)=\sum_{i=1}^k\frac{(-1)^i}{i},\\
E_2(k)=\sum_{i=1}^k \frac{(-1)^i}{i (1+i)},\quad
E_3(k)=\sum_{j=1}^k \frac{\displaystyle(-1)^j}{j} 
\sum_{i=1}^j \frac{(-1)^i}{i (1+i)}
\end{split}
\end{equation}
and the shift-operator $\Shift_k$ acting on these objects can be formulated in such extensions.
\begin{enumerate}
 \item Start with the rational function field $\KK=\QQ(n)$ and the automorphism $\fct{\sigma}{\KK}{\KK}$ with $\sigma=\text{id}_{\KK}$. Now take the $S$-extension $\dfield{\AA_1}{\sigma}$ of $\dfield{\QQ(n)}{\sigma}$ with $\AA_1=\QQ(n)(x)$ and $\sigma(x)=x+1$. There the rational expressions in $n$ and $k$ are represented in $\AA_1$ where $x$ takes over the role of $k$, and the shift operator $\Shift_k$ is rephrased by $\sigma$.

 \item Taking the $A$-extension $\dfield{\AA_2}{\sigma}$ of $\dfield{\AA_1}{\sigma}$ with $\AA_2=\AA_1[y]$ and $\sigma(y)=-y$ of order $2$ we can model $(-1)^k$ by $y$.
 
\item Constructing the $P$-extension $\dfield{\AA_3}{\sigma}$ of $\dfield{\AA_2}{\sigma}$ with $\AA_3=\AA_2\ltr{p_1}$ and $\sigma(p_1)=2\,p_1$ we model in addition polynomial expressions in $2^k$ and $2^{-k}$ with $\Shift_k 2^k=2\,2^k$ (and $\Shift_k \frac1{2^k}=\frac12\,\frac1{2^k}$) by rephrasing $2^k$ with $p_1$ (and $\frac{1}{2^k}$ with $\frac1{p_1}$).

\item Introducing the $P$-extension $\dfield{\AA_4}{\sigma}$ of $\dfield{\AA_3}{\sigma}$ with $\AA_4=\AA_3\ltr{p_2}$ and $\sigma(p_2)=\frac{n-x}{x+1}\,p_2$ we are in the position to model polynomial expressions in $\binom{n}{k}$ and $\binom{n}{k}^{-1}$ with $\Shift_k\binom{n}{k}=\frac{n-k}{k+1}\binom{n}{k}$ by rephrasing $\binom{n}{k}$ with $p_2$ and $\binom{n}{k}^{-1}$ with $\frac1{p_2}$.

\item Further, taking the $S$-extension $\dfield{\AA_5}{\sigma}$ of $\dfield{\AA_4}{\sigma}$ with $\AA_5=\AA_4[t_1]$ and $\sigma(t_1)=t_1+\frac{-y}{x+1}$ we can represent polynomial expressions in the sum $E_1(k)$ with $\Shift_k E_1(k)=E_1(k+1)=E_1(k)+\frac{-(-1)^k}{k+1}$ by rephrasing $E_1(k)$ with $t_1$.

\item In addition, building the $S$-extension $\dfield{\AA_6}{\sigma}$ of $\dfield{\AA_5}{\sigma}$ with $\AA_6=\AA_5[t_2]$ and $\sigma(t_2)=t_2+\frac{-y}{(x+1)(x+2)}$ we can represent polynomial expressions in the sum $E_2(k)$ with $\Shift_k E_2(k)=E_2(k)+\frac{-(-1)^k}{(k+1)(k+2)}$ by rephrasing $E_2(k)$ with $t_2$.

\item Finally, introducing the $S$-extension $\dfield{\AA_7}{\sigma}$ of $\dfield{\AA_6}{\sigma}$ with $\AA_7=\AA_6[t_3]$ and $\sigma(t_3)=t_3+\frac{1-(x+1)(x+2) y t_ 2}{(x+1)^2 (x+2)}$ one can represent polynomial expressions in the sum $E_3(k)$ with $\Shift_k  E_3(k)=E_3(k)+\frac{1-(k+1)(k+2) (-1)^k E_2(k)}{(k+1)^2 (k+2)}$ by rephrasing $E_3(k)$ with $t_3$.
\end{enumerate}
\end{example}

For convenience we introduce the following notations.
Let $\dfield{\EE}{\sigma}$ be a difference ring extension of $\dfield{\AA}{\sigma}$ with $t\in\EE$. Then $\AA\lr{t}$ denotes
the polynomial ring $\AA[t]$ (we assume that $t$ is transcendental over $\AA$) if $\dfield{\AA[t]}{\sigma}$ is an
$S$-extension of $\dfield{\AA}{\sigma}$. 
$\AA\lr{t}$ denotes the ring of Laurent polynomials $\AA\ltr{t}$ (i.e., $t$ is transcendental over $\AA$) if
$\dfield{\AA\ltr{t}}{\sigma}$ is a $P$-extension of $\dfield{\AA}{\sigma}$.
Finally, $\AA\lr{t}$ is the ring $\AA[t]$ with $t\notin\AA$ subject to the relation $t^{\lambda}=1$ if
$\dfield{\AA[t]}{\sigma}$ is an $A$-extension of $\dfield{\AA}{\sigma}$ of
order ${\lambda}$. We call a difference ring extension $\dfield{\AA\lr{t}}{\sigma}$ of $\dfield{\AA}{\sigma}$  
\begin{itemize}
\item an \notion{$AP$-extension} (and $t$ is an \notion{$AP$-monomial}) if it is an $A$- or a $P$-extension;
\item an \notion{$AS$-extension} (and $t$ is an \notion{$AS$-monomial}), if it is an $A$- or an $S$-extension; 
\item a \notion{$PS$-extension} (and $t$ is a \notion{$PS$-monomial}), if it is a $P$- or an $S$-extension; 
\item an \notion{$APS$-extension} (and $t$ is an \notion{$APS$-monomial}) if it is an $A$-, $P$- or $S$-extension.
\end{itemize}
$\dfield{\AA\lr{t_1}\dots\lr{t_e}}{\sigma}$ is called a \notion{(nested) $APS$-extension} (resp.\ \notion{(nested) $A$-, $P$-, $S$-, $AP-$, $AS$-, $PS$-extension}) of $\dfield{\AA}{\sigma}$ if it is built by a tower of such extensions. For a field $\AA$ we say that $\dfield{\AA(t)}{\sigma}$ is a \notion{$PS$-field extension} of $\dfield{\AA}{\sigma}$ if it is a $P$-field or an $S$-field extension. In addition, $\dfield{\AA(t_1)\dots(t_e)}{\sigma}$ is a \notion{(nested) $P$-/$S$-/$PS$-field extension} of $\dfield{\AA}{\sigma}$ if it is a tower of such extensions.

Let $\dfield{\GG}{\sigma}$ be a difference field and let $\dfield{\GG\lr{t_1}\dots\lr{t_e}}{\sigma}$ be a $PS$-extension of $\dfield{\GG}{\sigma}$. Then field of fractions $\FF=Q(\GG\lr{t_1}\dots\lr{t_e})$ forms a rational function field with the variables $t_i$, i.e., $\FF=\GG(t_1)\dots(t_e)$. Then there is exactly one automorphism $\fct{\sigma'}{\FF}{\FF}$ with $\sigma'|_{\GG\lr{t_1}\dots\lr{t_e}}=\sigma$ which is defined by $\sigma'(\frac{p}{q})=\frac{\sigma(p)}{\sigma(q)}$ with $p,q\in\GG[t_1,\dots,t_e]$. 
As above we do not distinguish anymore between $\sigma$ and $\sigma'$.
We call this difference field extension $\dfield{\FF}{\sigma}$ of $\dfield{\GG}{\sigma}$ also a \notion{polynomial $PS$-field extension}. Similarly, we call it a polynomial $P$-/$S$-field extension if it is built only by $P$-/$S$-monomials.

\smallskip

\noindent\textit{Remark.} Any polynomial $PS$-field extension is also a $PS$-field extension, but the reverse statement does not hold in general: in the summands and multiplicands of polynomial $PS$-monomials, the $PS$-monomials introduced earlier can occur only as (Laurent) polynomial expressions and not as rational function expressions.

\begin{example}\label{Exp:PSField}
We introduce the rational difference field and the $q$--mixed case.
\begin{enumerate}
\item Take the difference field $\dfield{\KK}{\sigma}$ with $\sigma(c)=c$ for all $c\in\KK$. Now consider the $S$-extension $\dfield{\KK[x]}{\sigma}$ of $\dfield{\KK}{\sigma}$ with $\sigma(x)=x+1$. Then there is exactly one difference ring extension $\dfield{\KK(x)}{\sigma}$ of $\dfield{\KK[x]}{\sigma}$ with $\sigma(x)=x+1$. In other words, $\dfield{\KK(x)}{\sigma}$ is an $S$-field extension of $\dfield{\KK}{\sigma}$.
For $\KK=\QQ(n)$ this construction yields the difference field $\dfield{\AA_1}{\sigma}$ in Example~\ref{Exp:MainDRNaiveDef}. $\dfield{\KK(x)}{\sigma}$ is also called the rational difference field.
\item Take the difference field $\dfield{\KK}{\sigma}$ with $\sigma(c)=c$ for all $c\in\KK$ where $\KK=\KK'(q_1,\dots,q_v)$ is a rational function field. Let $\dfield{\EE}{\sigma}$ with $\EE=\KK[x]\ltr{x_1}\dots\ltr{x_v}$ be the $PS$-extension of $\dfield{\KK}{\sigma}$ with $\sigma(x)=x+1$ and $\sigma(x_i)=q_i\,x_i$ for $1\leq i\leq v$. If we take the field of fractions $\FF=Q(\EE)=\KK(x)(x_1)\dots(x_v)$,
the difference field $\dfield{\FF}{\sigma}$ is a $PS$-field extension of $\dfield{\KK}{\sigma}$. $\dfield{\FF}{\sigma}$ is also called the $q$--mixed difference field; compare~\cite{Bauer:99}.
\end{enumerate}
\end{example}


\subsection{\rpisiSE-extensions}\label{Sec:RPSExt}

So far we adjoined nested sums and products in a difference ring (resp.\ difference field) without taking care if algebraic relations exist between them. Later we will work out that such relations can be excluded if and only if the ring of constants remains unchanged during the construction of $APS$-extension; see Theorem~\ref{Thm:EquivSimpleConst} below. 

\begin{definition}\label{Def:NestedExt}
An $A$-extension/$P$-extension/$S$-extension $\dfield{\AA\lr{t}}{\sigma}$ of $\dfield{\AA}{\sigma}$ is called an \notion{\rE-extension/\piE-extension/\sigmaSE-extension} if $\const{\AA\lr{t}}{\sigma}=\const{\AA}{\sigma}$. 
Similarly, we call such an extension an \notion{\rE\piE-/\rE\sigmaSE-/\pisiSE-/\rpisiSE-extension} if it is an $AP$-/$AS$-/$PS$/-$APS$-ex\-tension with $\const{\AA\lr{t}}{\sigma}=\const{\AA}{\sigma}$. Depending on the type of extension, $t$ is called an \notion{\rE-/\piE-/\sigmaSE-/\rE\piE-/\rE\sigmaSE-/\pisiSE-/\rpisiSE-monomial}, respectively.\\
An $APS$-extension $\dfield{\AA\lr{t_1}\dots\lr{t_e}}{\sigma}$ of $\dfield{\AA}{\sigma}$ is called a \notion{(nested) \rE\pisiSE-extension (resp.\ (nested) \rE-, \piE-, \sigmaSE-, \rE\piE, \rE\sigmaSE-, \pisiSE- extension)} if it is a tower of such extensions.\\
We will consider also mixed cases like $RPS$-extensions built by $R$- and $PS$-extensions.
\end{definition}

\noindent We will rely heavily on the following property of \rpisiSE-extensions~\cite[Thm~2.12]{Schneider:16a} generalizing the \pisiSE-field results given in~\cite{Karr:85,Schneider:01}.

\begin{theorem}\label{Thm:RPSCharacterization}
Let $\dfield{\AA}{\sigma}$ be a difference ring. Then
the following holds.
\begin{enumerate}
\item Let $\dfield{\AA[t]}{\sigma}$ be an $S$-extension 
of $\dfield{\AA}{\sigma}$ with $\sigma(t)=t+\beta$ where $\beta\in\AA$ such that $\const{\AA}{\sigma}$ is a field. Then this is a \sigmaSE-extension (i.e.,
$\const{\AA[t]}{\sigma}=\const{\AA}{\sigma}$) iff there does not exist a
$g\in\AA$ with $\sigma(g)=g+\beta$.

\item Let $\dfield{\AA\ltr{t}}{\sigma}$ be a $P$-extension of
$\dfield{\AA}{\sigma}$ with $\sigma(t)=\alpha\,t$ where $\alpha\in\AA^*$. Then
this is a
\piE-extension (i.e., $\const{\AA\ltr{t}}{\sigma}=\const{\AA}{\sigma}$) iff
there are no $g\in\AA\setminus\{0\}$ and $m\in\ZZ\setminus\{0\}$
with
$\sigma(g)=\alpha^m\,g$. 

\item Let $\dfield{\AA[t]}{\sigma}$ be an $A$-extension of $\dfield{\AA}{\sigma}$ of order $\lambda>1$ with $\sigma(t)=\alpha\,t$
where $\alpha\in\AA^*$. Then this is an \rE-extension (i.e.,
$\const{\AA[t]}{\sigma}=\const{\AA}{\sigma}$) iff there are no
$g\in\AA\setminus\{0\}$ and $m\in\{1,\dots,\lambda-1\}$ with
$\sigma(g)=\alpha^m\,g$. 
\end{enumerate}
\end{theorem}

\noindent Finally, we introduce \pisiSE-field extensions; compare~\cite{Karr:81,Karr:85,Schneider:16a}.

\begin{definition}
A (polynomial) $PS$-field/$P$-field/$S$-field extension $\dfield{\FF(t_1)\dots(t_e)}{\sigma}$ of a difference field $\dfield{\FF}{\sigma}$ is called a \notion{(polynomial) \pisiSE-field/\piE-field/\sigmaSE-field extension} if $\const{\FF(t_1)\dots(t_e)}{\sigma}=\const{\FF}{\sigma}$. In particular, if $\const{\FF}{\sigma}=\FF$, $\dfield{\FF(t_1)\dots(t_e)}{\sigma}$ is called a \notion{(polynomial) \pisiSE-field over $\FF$}.
\end{definition}

\noindent In Section~\ref{Sec:Telescoping} below we will utilize the following result. If $\dfield{\FF\lr{t_1}\dots\lr{t_e}}{\sigma}$ is a \pisiSE-extension  of a difference field $\dfield{\FF}{\sigma}$, then the difference field of fractions $\dfield{\FF(t_1)\dots(t_e)}{\sigma}$ forms a polynomial \pisiSE-field extension. This property follows by Corollary~\ref{Cor:LiftToField}.

\begin{corollary}\label{Cor:LiftToField}
Let $\dfield{\FF\lr{t}}{\sigma}$ be a $PS$-extension of a difference field $\dfield{\FF}{\sigma}$ with $\sigma(t)=\alpha\,t+\beta$ and let $\dfield{\FF(t)}{\sigma}$ be the $PS$-field extension of $\dfield{\FF}{\sigma}$ with $\sigma(t)=\alpha\,t+\beta$. Then $\const{\AA\lr{t}}{\sigma}=\const{\AA}{\sigma}$ iff $\const{\AA(t)}{\sigma}=\const{\AA}{\sigma}$. 
\end{corollary}
\begin{proof}
The implication ``$\Leftarrow$'' is obvious. Now let $t$ be an $S$-monomial with $\sigma(t)=t+\beta$, and suppose that $\const{\FF[t]}{\sigma}=\const{\FF}{\sigma}$, which is a field. By Theorem~\ref{Thm:RPSCharacterization} it follows that there does not exist a $g\in\FF$ with $\sigma(g)=g+\beta$. This implies that $\const{\FF(t)}{\sigma}=\const{\FF}{\sigma}$ by~\cite[Thm.~3.11]{Schneider:16a} or~\cite{Karr:85}.
Similarly, let $t$ be a $P$-monomial with $\sigma(t)=\alpha\,t$ where $\alpha\in\FF^*$ and assume that $\const{\FF\ltr{t}}{\sigma}=\const{\FF}{\sigma}$. By Theorem~\ref{Thm:RPSCharacterization} it follows that there does not exist a $g\in\FF^*$ and $m\in\set\ZZ\setminus\{0\}$ with $\sigma(g)=\alpha^m\,g$. This implies that $\const{\FF(t)}{\sigma}=\const{\FF}{\sigma}$ by~\cite[Thm.~3.18]{Schneider:16a} or~\cite{Karr:85}.
\end{proof}

\subsection{Basic $APS$- and \rpisiSE-extensions}\label{Sec:BasicExt}

Extending the techniques developed in~\cite{Karr:81} we turned Theorem~\ref{Thm:RPSCharacterization} in~\cite{Schneider:16a} to a constructive version for various subclasses of \rpisiSE-extensions. Here we will focus on the subclass of basic \rpisiSE-extensions (which we called simple and single-rooted \rpisiSE-extensions in~\cite{Schneider:16a}). In order to introduce this type of extensions, we need the following notation.

Let $\dfield{\EE}{\sigma}$ with $\EE=\AA\lr{t_1}\dots\lr{t_e}$ be an $APS$-extension  (or \rpisiSE-extension) of $\dfield{\AA}{\sigma}$ as given in Definition~\ref{Def:NestedExt}. Let $G$ be a subgroup of $\AA^*$. Then we define the set
$$\PPowers{\EE}{\AA}{G}=\{u\,t_1^{l_1}\dots t_{e}^{l_{e}}|\begin{array}[t]{l}\,u\in G\text{ and $l_i\in\ZZ$ for $1\leq i\leq e$}\\[-0.1cm]
\text{where $l_i=0$ if $t_i$ is an $AS$-monomial (or \rE\sigmaSE-monomial)}\}
\end{array}.$$
which forms a subgroup of the multiplicative group $\EE^*$; for a more general group see~\cite{Schneider:16a}.

\begin{definition}\label{Def:BasicAPSExt}
Let $\dfield{\AA}{\sigma}$ be a difference ring and $G$ be a subgroup of $\AA^*$.
An $APS$-extension (or \rpisiSE-extension) $\dfield{\AA\lr{t_1}\dots\lr{t_e}}{\sigma}$ of $\dfield{\AA}{\sigma}$ is called
\notion{$G$-basic} if for $1\leq i\leq e$ the following holds:
\begin{enumerate}
 \item if $t_i$ is an $A$-monomial (or \rE-monomial) then $\frac{\sigma(t_i)}{t_i}\in\const{\AA}{\sigma}^*$;
 \item if $t_i$ is a $P$-monomial (or \piE-monomial) then $\frac{\sigma(t_i)}{t_i}\in\PPowers{\AA\lr{t_1}\dots\lr{t_{i-1}}}{\AA}{G}$.
\end{enumerate}
An $\AA^*$-basic extension is also called a \notion{basic extension}. The $t_i$ are called \notion{basic/$G$-basic $APS$-/\rpisiSE monomials}.
\end{definition}

\noindent Everywhere (except in Section~\ref{Sec:Embedding}) we will consider the most general case $G=\AA^*$.  Given such a tower of basic $APS$-extensions we emphasize the following property: for any \piE-monomial $t_i$ the multiplicand $\frac{\sigma(t_i)}{t_i}$ may depend on earlier introduced $P$-monomials but is free of $AS$-monomials. 
So far we have introduced basic $APS$-monomials (most of them being basic \rpisiSE-monomials) where the $P$-monomials are not nested. Here is an example for nested basic $P$-monomials.

\begin{example}\label{Exp:NestedP}
Take the \pisiSE-field $\dfield{\KK(x)}{\sigma}$ over $\KK$ with $\sigma(x)=x+1$. Then we can construct 
the basic $P$-extension $\dfield{\KK(x)\ltr{p_1}\ltr{p_2}\ltr{p_3}}{\sigma}$ of $\dfield{\KK(x)}{\sigma}$ with $\sigma(p_1)=(x+1)p_1$, $\sigma(p_2)=(x+1)p_1\,p_2$ and $\sigma(p_3)=(x+1)p_1\,p_2\,p_3$.   
\end{example}

\noindent Note that the order of the generators in the tower of a general $APS$-extension $\EE=\AA\lr{t_1}\dots\lr{t_e}$ is essential. Take, e.g., an S-monomial $t_i$ with $1\leq i\leq e$ where $\sigma(t_i)-t_i$ depends on $t_j$ with $1\leq j\leq e$. Then $t_j$ must be introduced before $t_i$, i.e., $j<i$. Similarly, one has to take into account the correct order for $P$- and $A$-monomials. But if $\dfield{\EE}{\sigma}$ is a basic $APS$-extension (or a basic \rpisiSE-extension) of $\dfield{\AA}{\sigma}$, we can exploit the property mentioned already above: for any $P$-extension $t_i$ the multiplicand $\frac{\sigma(t_i)}{t_i}$ may depend on the previously defined $P$-monomials but not on the previously defined $R$- or $S$-monomials. Hence one can reorder a given tower of such extensions where first all $A$-monomials are adjoined, then all $P$-monomials are adjoined, and afterwards all $S$-monomials are adjoined. In other words, we may assume that
\begin{equation}\label{Equ:APSOrdered}
\EE=\AA[y_1,\dots,y_l][p_1,p_1^{-1}]\dots[p_r,p_r^{-1}][s_1]\dots[s_v]
\end{equation}
holds where the $y_i$ are $A$-monomials (or \rE-monomials), the $p_i$ are $P$-monomials (or \piE-monomials)  and the $s_i$ are $S$-monomials (or \sigmaSE-monomials). Note that within the block of $P$-monomials and the block of $S$-monomials we are only allowed to use rearrangements that take care of the recursive nature of the $P$-monomials and $S$-monomials.  Besides this representation, we will also use rearrangements of the form
\begin{equation}\label{Equ:PASOrdered}
\EE=\AA[p_1,p_1^{-1}]\dots[p_r,p_r^{-1}][y_1,\dots,y_l][s_1]\dots[s_v].
\end{equation}
A reordering to this form is possible since the $p_i$ are free of $A$-monomials (or of \rE-monomials). 

\medskip

\noindent In the following we will elaborate further properties used in the rest of the paper.

\begin{lemma}\label{Lemma:SInverseElements}
(1) Let $\dfield{\EE}{\sigma}$ be an $RPS$-extension of a difference field $\dfield{\FF}{\sigma}$. Then $\EE$ is reduced, i.e., there are no non-zero nilpotent elements in $\EE$.\\
(2) Let $\dfield{\EE}{\sigma}$ be an $S$-extension of a reduced difference ring $\dfield{\AA}{\sigma}$. Then $\EE^*=\AA^*$.\\
(3)
Let $\dfield{\EE}{\sigma}$ be a $PS$-extension of a difference field $\dfield{\FF}{\sigma}$. Then 
$$\EE^*=\{h\,t_1^{m_1}\dots t_e^{m_e}|\,h\in\FF^*, \text{$m_i\in\ZZ$ for $1\leq i\leq e$ with $m_i=0$ if $t_i$ is an $S$-monomial}\}.$$
\end{lemma}
\begin{proof}
Let $\dfield{\EE}{\sigma}$ be given as in~\eqref{Equ:APSOrdered} (with $\AA=\FF$). Set $\HH=\FF[y_1,\dots,y_l]$ where the $y_i$ are \rE-monomials. Note that $\EE$ is built by a tower of $PS$-monomials on top of $\HH$. By~\cite[Corollary~4.3]{Schneider:16a}, $\HH$ is reduced. Moreover, one can verify the following fact: if a ring $\AA$ is reduced, the ring of polynomials $\AA[t]$ or the ring of Laurent polynomials $\AA\ltr{t}$ is reduced (for a straightforward proof see ~\cite[Corollary~4.3]{Schneider:16a}: there we showed this result for \pisiSE-monomials -- however, we never used the property that the constants remain unchanged and thus the proof is valid also for $PS$-monomials). Thus by iterative application of this result, it follows that $\EE$ is reduced and (1) is proven.\\ 
Now we use the fact that $\AA[t]^*=\AA^*$ and $\AA\ltr{t}^*=\{w\,t^m|w\in\AA^*, m\in\ZZ\}$ if the coefficient ring $\AA$ is reduced; for details and more general statements see~\cite[Theorem~1]{Karpilovsky:83} (and~\cite{Neher:09}). Hence by iterative application of these properties in combination with part~(1) we conclude that parts~(2) and~(3) follow.
\end{proof}

\noindent Next, we present an alternative characterization of basic \rpisiSE-extensions (compare Section~\ref{Sec:Introduction}) under the assumption that the ground ring $\dfield{\AA}{\sigma}$ is a field.

\begin{proposition}\label{Prop:AltDefOfBasic}
Let $\dfield{\FF\lr{t_1}\dots\lr{t_e}}{\sigma}$ be an $APS$-extension of a difference field $\dfield{\FF}{\sigma}$. Then this extension is basic iff for $1\leq i\leq e$ the following holds:
\begin{enumerate}
 \item if $t_i$ is an $A$-monomial then $\frac{\sigma(t_i)}{t_i}\in\const{\FF}{\sigma}^*$;
 \item if $t_i$ is a $P$-monomial then $\alpha=\frac{\sigma(t_i)}{t_i}\in\EE^*$ where $\alpha$ is free of $R$-monomials.
\end{enumerate}
 \end{proposition}
\begin{proof}
A basic $APS$-extension fulfils properties (1) and (2) by definition.
Now suppose that we are given an $APS$-extension with the properties~(1) and~(2). Observe that the $APS$-extension can be reordered to the form~\eqref{Equ:APSOrdered} with $\FF=\AA$ by~\cite[Lemma~4.13]{Schneider:16a} (this reordering property is stated for basic \rpisiSE-extensions; however, the proof holds also for $APS$-extensions since it does not use the property that the constants remain unchanged). Since all $P$-monomials are free of $A$-monomials, we can reorder this extension further to the form~\eqref{Equ:PASOrdered}. Set $\EE=\FF\lr{p_1}\dots\lr{p_r}$. 
By part~(3) of Lemma~\ref{Lemma:SInverseElements} we have $\EE^*=\PPowers{\EE}{\FF}{(\FF^*)}$. Hence $\frac{\sigma(p_i)}{p_i}\in\EE^*$ for $1\leq i\leq r$. Consequently the $APS$-extension is basic.
\end{proof}

\noindent In this article we will use heavily the following result from~\cite[Thm.~2.22]{Schneider:16a}.

\begin{theorem}\label{Thm:ShiftQuotient}
Let $\dfield{\EE}{\sigma}$ be a basic \rpisiSE-extension of a difference field $\dfield{\FF}{\sigma}$ of the form~\eqref{Equ:APSOrdered}. Let 
$a=u\,y_1^{n_1}\dots y_l^{n_l}p_1^{z_1}\dots p_r^{z_r}$
with $u\in\FF^*$, $n_i\in\NN$ and $z_i\in\ZZ$, and let
$f\in\EE\setminus\{0\}$. If $\sigma(f)=a\,f$, then there are a $w\in\FF^*$, $\xi_i\in\NN$ and $\pi_i\in\ZZ$ with
\begin{equation}\label{Equ:ProductSol}
f=w\,y_1^{\xi_1}\dots y_l^{\xi_l} p_1^{\pi_1}\dots p_r^{\pi_r}\in\EE^*.
\end{equation}
\end{theorem}

\noindent Finally, we summarize the algorithmic toolbox from~\cite{Schneider:16a} in the following remark.

\begin{remark}\label{Remark:ConstructiveVersion}
We suppose that we are given a ground field $\dfield{\FF}{\sigma}$ with certain algorithmic properties stated in~\cite[Thm.~2.23]{Schneider:16a}; see also Section~\ref{Sec:Telescoping} below. For instance, we can take any \pisiSE-field extension $\dfield{\FF}{\sigma}$ of a difference field $\dfield{\GG}{\sigma}$ where $\dfield{\GG}{\sigma}$ is the constant field~\cite{Karr:81} (i.e., $\dfield{\FF}{\sigma}$ is a \pisiSE-field over $\GG$). For typical examples see the difference fields given in Example~\ref{Exp:qRat} below. Further,
$\dfield{\GG}{\sigma}$ can be, e.g., the difference field of unspecified sequences~\cite{Schneider:06d} or of radical expressions~\cite{Schneider:07f}. Note that in all three cases the underlying constant field must also meet certain algorithmic properties~\cite{Karr:81,Schneider:06d}. This is the case if we choose, e.g., a rational function field over an algebraic number field~\cite{Schneider:05c}.\\
Now let $\dfield{\EE}{\sigma}$ with $\EE=\FF\lr{t_1}\dots\lr{t_e}$ be a basic $APS$-extension of $\dfield{\FF}{\sigma}$ as specified above. Then Theorem~\ref{Thm:RPSCharacterization} provides a strategy to verify if the given $APS$-monomials are also \rpisiSE-monomials. Namely, suppose that we have checked that $\dfield{\AA}{\sigma}$ with $\AA=\FF\lr{t_1}\dots\lr{t_{k-1}}$ is an \rpisiSE-extension of $\dfield{\FF}{\sigma}$. Then given this setting the algorithms from~\cite{Schneider:16a} can be used in order to solve the following problems:
\begin{enumerate}
 \item given $\beta\in\AA$, decide constructively if there exists a $g\in\AA$ with $\sigma(g)=g+\beta$;
 \item given $\alpha\in\PPowers{\AA}{\FF}{(\FF^*)}$ with $\alpha$ not being a root of unity, decide constructively if there exists a $g\in\AA\setminus\{0\}$ and an $m\in\ZZ\setminus\{0\}$ such that $\sigma(g)=\alpha^m\,g$ holds;
 \item given $\alpha\in\KK^*$ with $\ord(\alpha)=\lambda\in\NN$ where $\lambda>1$, decide constructively if there exists a $g\in\AA\setminus\{0\}$ and an $m\in\NN$ with $0<m<\lambda$ such that $\sigma(g)=\alpha^m\,g$ holds.
\end{enumerate}
Hence we can treat the next \rpisiSE-monomial $t_k$ as follows. If $t_k$ is an $S$-monomial, we can check if $t_k$ is a \sigmaSE-monomial by testing if there exists a $g\in\AA$ with $\sigma(g)=g+\beta$ with $\beta=\sigma(t_k)-t_k$. If $t_k$ is a $P$-monomial (resp.\ $A$-monomial), we can check if $t_k$ is a \piE-monomial (resp.\ \rE-monomial) by testing if there exists a $g\in\AA\setminus\{0\}$ and an $m\in\ZZ\setminus\{0\}$ (resp.\ $0<m<\ord(t_k)$) with $\sigma(g)=\alpha^m\,g$ where $\alpha=\sigma(t_k)/t_k$.
\end{remark}

\begin{example}\label{Exp:qRat}
We go back to the ($q$--)rational difference field from Example~\ref{Exp:PSField}.
\begin{enumerate}
 \item Consider the rational difference ring $\dfield{\KK(x)}{\sigma}$. Since there is no $g\in\KK$ with $\sigma(g)=g+1$, $x$ is a \sigmaSE-monomial. 
  In particular, $\dfield{\KK(x)}{\sigma}$ is a \pisiSE-field over $\KK$.
  \item Let $\KK=\KK'(q_1,\dots,q_v)$ and  consider the $PS$-field extension $\dfield{\EE}{\sigma}$ of $\dfield{\KK(x)}{\sigma}$ from Example~\ref{Exp:PSField} with $\EE=\KK(x)\ltr{x_1}\dots\ltr{x_v}$ where $\sigma(x_i)=q_i\,x_i$ for $1\leq i\leq v$. It is not too difficult to see that there does not exist a $g\in\KK(x)^*$ and $(0,\dots,0)\neq(m_1,\dots,m_v)\in\ZZ^v$ with $\sigma(g)=q_1^{m_1}\dots q_v^{m_v}\,g$. Hence by~\cite[Thm.~9.1]{Schneider:10c} (or by the iterative application of Theorem~\ref{Thm:RPSCharacterization} combined with Corollary~\ref{Cor:LiftToField}) it follows that $\const{\EE}{\sigma}=\const{\KK(x)}{\sigma}=\KK$. Thus $\dfield{\EE}{\sigma}$ is a \pisiSE-field over $\KK$.
\end{enumerate}
\end{example}

\begin{example}\label{Exp:ConstructRPiSiExt}
We analyse the basic $APS$-monomials in Example~\ref{Exp:MainDRNaiveDef} using the algorithms from~\cite{Schneider:16a} that are implemented within the package \texttt{Sigma}.
\begin{enumerate}
\item As observed in Example~\ref{Exp:qRat} the difference field $\dfield{\KK(x)}{\sigma}$ with $\sigma(x)=x+1$ and $\const{\KK}{\sigma}=\KK=\QQ(n)$ is a \pisiSE-field over $\KK$.

 \item There are no $g\in\KK(x)^*$ and $m\in\{1\}$ with $\sigma(g)=(-1)^m\,g$. Hence the $A$-extension $\dfield{\AA_2}{\sigma}$ of $\dfield{\AA_1}{\sigma}$ is an $R$-extension.
 
\item There are no $g\in\AA_2\setminus\{0\}$ and $m\in\ZZ\setminus\{0\}$ with $\sigma(g)=2^m\,g$. Thus $\dfield{\AA_3}{\sigma}$ is a \piE-extension of $\dfield{\AA_2}{\sigma}$.

\item There are no $g\in\AA_3\setminus\{0\}$ and $m\in\ZZ\setminus\{0\}$ with $\sigma(g)=\big(\frac{n-x}{x+1}\big)^m\,g$. Thus $\dfield{\AA_4}{\sigma}$ is a \piE-extension of $\dfield{\AA_3}{\sigma}$.

\item There is no $g\in\AA_4$ with $\sigma(g)=g+\frac{-y}{x+1}$. Thus $\dfield{\AA_5}{\sigma}$ is a \sigmaSE-extension of $\dfield{\AA_4}{\sigma}$.

\item However, we find $g=\frac{2 (x+1) t_1+x-y+1}{x+1}\in\AA_5$ with $\sigma(g)=g+\frac{-y}{(x+1)(x+2)}$. Thus $\dfield{\AA_6}{\sigma}$ is not a \sigmaSE-extension of $\dfield{\AA_5}{\sigma}$. 
\end{enumerate}
To sum up, the difference ring $\dfield{\AA_5}{\sigma}$ is an \rpisiSE-extension of $\dfield{\QQ(n)(x)}{\sigma}$ with $\AA_5=\QQ(n)(x)(y)\ltr{p_1}\ltr{p_2}[t_1]$, i.e., $\const{\AA_5}{\sigma}=\const{\QQ(n)(x)}{\sigma}=\QQ(n)$. But $t_2$ is not a \sigmaSE-monomial, e.g., we find the constant 
\begin{equation}\label{Equ:ConstantForSigma}
c=t_2-g=t_2-\frac{2 (x+1) t_1+x-y+1}{x+1}\in\const{\AA_6}{\sigma}\setminus\AA_5.
\end{equation}
\end{example}

\begin{example}\label{Exp:3AMonomials}
Take the \pisiSE-field $\dfield{\KK(x)}{\sigma}$ over $\KK$ with $\sigma(x)=x+1$ and the $A$-extension $\dfield{\KK(x)[y_1][y_2][y_3]}{\sigma}$ of $\dfield{\KK(x)}{\sigma}$ with $y_1^2=1$ and $\sigma(y_1)=-y_1$, with 
$\sigma(y_2)=(-1)^{2/3}\,y_2$ and $y_2^3=1$, and with $\sigma(y_3)=(-1)^{1/6}y_3$ and $y_3^{12}=1$. We apply Theorem~\ref{Thm:RPSCharacterization} together with our algorithmic machinery~\cite{Schneider:16a}. Since there does not exist a $g\in\KK(x)\setminus\{0\}$ with $\sigma(g)=-g$, $y_1$ is an \rE-monomial. Moreover, since there are no $g\in\KK(x)[y_1]\setminus\{0\}$ and $m\in\{1,2\}$ with $\sigma(g)=(-1)^{2\,m/3} g$, $y_2$ is also an \rE-monomial. But we find $m=10$, $g=y_1\,y_2$ with $\sigma(g)=(-1)^{m/6}\,g$. Thus $y_3$ is not an \rE-monomial. In particular, we get 
\begin{equation}\label{Equ:ConstantForR}
c=y_1 y_2 y_3^2\in\const{\KK(x)[y_1][y_2][y_3]}{\sigma}\setminus\KK(x)[y_1][y_2].\end{equation}
For an alternative and direct check we refer to Example~\ref{Exp:SimpleCheck1}.
\end{example}

\begin{example}\label{Exp:3PiMonomials}
Take the difference ring $\dfield{\AA_3}{\sigma}$ with $\AA_3=\QQ(n)(x)[y]\ltr{p_1}$ given in Example~\ref{Exp:MainDRNaiveDef}. Now take the $P$-extension $\dfield{\AA_3\ltr{p_2}\ltr{p_3}}{\sigma}$ of $\dfield{\AA_3}{\sigma}$ with $\alpha_2:=\frac{\sigma(p_2)}{p_2}=\frac{(n
        -x
        )^2}{(x+1)^2}$
and with $\alpha_3:=\frac{\sigma(p_3)}{p_3}=-\frac{4 (1
        +n
        -x
        )}{x+1}$.
Using part 2 of Theorem~\ref{Thm:RPSCharacterization} we can verify that $p_2$ is a \piE-monomial. However,
we find $m=2$ and $g=
        -\frac{p_1^4 p_2}{(1
        +n
        -x
        )^2}$
such that $\sigma(g)=\alpha_3^m\,g$ holds. Therefore $p_3$ is not a \piE-monomial and we get
\begin{equation}\label{Equ:ConstantForPi}
c=\frac{-p_1^4 p_2}{(1
        +n
        -x
        )^2 p_3^2}\in\const{\AA_3\ltr{p_2}\ltr{p_3}}{\sigma}\setminus\AA_3\ltr{p_2}.
\end{equation}
\end{example}

\begin{example}\label{Exp:NestedP2}
Take the basic $P$-extension $\dfield{\KK(x)\ltr{p_1}\ltr{p_2}\ltr{p_3}}{\sigma}$ of $\dfield{\KK(x)}{\sigma}$ from Example~\ref{Exp:NestedP}. We check that there is no $g\in\KK(x)^*$ and $m\in\ZZ\setminus\{0\}$ with $\sigma(g)=(x+1)^m\,g$. Thus $p_1$ is a \piE-monomial. Next, we show by our algorithms that there is no $g\in\KK(x)\ltr{p_1}\setminus\{0\}$ and $m\in\ZZ\setminus\{0\}$ with $\sigma(g)=((x+1)p_1)^m\,g$. Hence $p_2$ is a \piE-monomial. Finally, we verify that there is no $g\in\KK(x)\ltr{p_1}\ltr{p_2}\setminus\{0\}$ and $m\in\ZZ\setminus\{0\}$ such that $\sigma(g)=((x+1)p_1\,p_2)^m\,g$ holds. Consequently, $p_3$ is a \piE-monomial. In total, $\dfield{\KK(x)\ltr{p_1}\ltr{p_2}\ltr{p_3}}{\sigma}$ is a \piE-extension of $\dfield{\KK(x)}{\sigma}$. 
\end{example}

\subsection{Simple tests for basic \rE-extensions}\label{Sec:SimpleTests}

In the following we will carve out further properties of \rE-extensions. In particular, we will show in Proposition~\ref{Prop:CharactSeveralRExt} that the verification of an \rE-extension is immediate if certain properties hold for the underlying difference field $\dfield{\FF}{\sigma}$. In addition, we will show in Lemma~\ref{Lemma:RSingleIsoMultiple} that a ring obtained by several \rE-extensions is isomorphic to a ring with just one properly chosen \rE-extension; this construction will be relevant in Section~\ref{Sec:Copies}.

\medskip

\noindent We start with the following definition that will be crucial throughout this article.
\begin{definition}\label{Def:ConstantStable}
A difference ring (resp.\ field) is called \notion{constant-stable} if $\const{\FF}{\sigma^k}=\const{\FF}{\sigma}$ for all $k\in\NN$.
\end{definition}

\begin{remark}\label{Remark:ConstantStable}
We note that the class of difference fields introduced in Remark~\ref{Remark:ConstructiveVersion} are all constant-stable; compare~\cite[Subsec.~2.3.3]{Schneider:16a}. In particular, all difference rings in our examples (except Example~\ref{Exp:NotConstantStable} below) are built by a tower of basic $APS$- or \rpisiSE- extensions (or mixed versions) over a constant-stable difference field.  
\end{remark}

\noindent We obtain the following simple characterization (compare~\cite[Prop.~1]{DR2}) under the assumption that the ground field $\dfield{\FF}{\sigma}$ is constant-stable.

\begin{proposition}\label{Prop:RExt}
Let $\dfield{\FF}{\sigma}$ be a constant-stable difference field. Let $\dfield{\FF[t]}{\sigma}$ be a basic $A$-extension of $\dfield{\FF}{\sigma}$ of order $\lambda$ with $\alpha=\frac{\sigma(t)}{t}\in\const{\FF}{\sigma}^*$. Then this is an \rE-extension iff $\alpha$ is a primitive root of unity. 
\end{proposition}
\begin{proof}
``$\Rightarrow$'': Suppose that $\alpha$ is not a primitive root of unity of order $\lambda$. Then we get $\alpha^s=1$ for some $s\in\NN$ with $0<s<\lambda$. Hence $\sigma(t^s)=\alpha^s\,t^s=t^s$, i.e., $t^s\in\const{\FF[t]}{\sigma}$. Since $t^{\lambda}=1$ is the defining relation, $t^s\notin\FF$. Thus $t$ is not an \rE-monomial.\\
``$\Leftarrow$'':
Assume there is a $g=\sum_{i=0}^{\lambda-1}g_i\,t^i\in\FF[t]\setminus\FF$ with $\sigma(g)=g$. Thus we can take an $m$ with $0<m<\lambda$ and $g_m\in\FF^*$. Comparing the $m$-th coefficient in $\sigma(g)=g$ gives $\sigma(g_m)=\alpha^{-m}\,g_m$.
Since $\alpha$ is a constant, $\sigma^{\lambda}(g_m)=(\alpha^{-m})^{\lambda}\,g_m=g_m$. Hence $g_m\in\const{\FF}{\sigma}^*$ using the property that  $\dfield{\FF}{\sigma}$ is constant-stable. Consequently, $g_m=\sigma(g_m)=\alpha^{-m}\,g_m$ and therefore $\alpha^m=1$. Thus $\alpha$ is not a primitive root of unity of order $\lambda$.
\end{proof}

\noindent Note that this characterization does not hold in general if $\dfield{\FF}{\sigma}$ is not constant-stable.

\begin{example}\label{Exp:NotConstantStable}
Take the rational function field $\QQ(x)$ and define the field automorphism subject to the relation $\sigma(x)=-x$. Since $\sigma(x^2)=x^2$, we get $\const{\QQ(x)}{\sigma}=\QQ(x^2)$. Moreover, since $\sigma^2(x)=x$, we have $\const{\QQ(x)}{\sigma^2}=\QQ(x)$. Hence $\dfield{\QQ(x)}{\sigma}$ is not constant-stable. Now take the $A$-extension $\dfield{\QQ(x)[y]}{\sigma}$ with $\sigma(y)=-y$. Then $\sigma(\frac{y}{x})=\frac{y}{x}$, i.e., $\const{\QQ(x)[y]}{\sigma}\supsetneq\const{\QQ(x)}{\sigma}$. Hence $y$ is not an \rE-monomial. 
\end{example}

\noindent Next, we show that  several \rE-extensions can be merged into one \rE-extension by means of a difference ring isomorphism.

\begin{lemma}\label{Lemma:RSingleIsoMultiple}
Let $\dfield{\FF}{\sigma}$ be a constant-stable difference field.
Let $\dfield{\EE}{\sigma}$ with $\EE=\FF[t_1]\dots[t_e]$ be a basic $A$-extension of $\dfield{\FF}{\sigma}$ of orders $\lambda_1,\dots,\lambda_e$, respectively. Further, suppose that the $\alpha_i=\frac{\sigma(t_i)}{t_i}\in\const{\FF}{\sigma}^*$ with $1\leq i\leq e$ are primitive and that $\gcd(\lambda_i,\lambda_j)=1$ holds for pairwise different $i,j$. Then the following holds.
\begin{enumerate}
\item There is an \rE-extension $\dfield{\FF[t]}{\sigma}$ of
$\dfield{\FF}{\sigma}$ of order $\lambda:=\prod_{i=1}^e\lambda_i$ with $\frac{\sigma(t)}{t}=\prod_{i=1}^e\alpha_i$. 
\item There is a difference ring isomorphism $\fct{\tau}{\EE}{\FF[t]}$ with
$\tau|_{\FF}=\text{id}_{\FF}$ such that for $1\leq i\leq e$ we have $\tau(t_i)=t^{r_i}$ for explicitly computable $r_i\in\NN$. In particular, the inverse of $\tau$ is computable.
\item $\dfield{\EE}{\sigma}$ is an \rE-extension of $\dfield{\FF}{\sigma}$.
\end{enumerate}
\end{lemma}


\begin{proof}
(1) Note that $\alpha=\alpha_1\dots\alpha_e$ is a primitive $\lambda$th root of unity. Hence
by Proposition~\ref{Prop:RExt} $\dfield{\FF[t]}{\sigma}$ is an \rE-extension of $\dfield{\FF}{\sigma}$.\\ (2) 
Since $\frac{\lambda}{\lambda_i}$ and $\lambda_i$ are relatively prime, we can compute $u_i,v_i\in\ZZ$ such that 
$u_i\,\frac{\lambda}{\lambda_i}+v_i\,\lambda_i=1$ and $0\leq u_i<\lambda_i$ holds.
Now set $r_i:=u_i\,\frac{\lambda}{\lambda_i}\in\NN$ for $1\leq i\leq e$. 
Observe that $\lambda_i$ and $u_i$ are relatively prime (since otherwise the left hand side of $u_i\,\frac{\lambda}{\lambda_i}+v_i\,\lambda_i=1$ is divisible by a nontrivial integer factor). 
Clearly, there is a ring homomorphism $\fct{\tau}{\EE}{\FF[t]}$ defined by
$\tau|_{\FF}=\text{id}_{\FF}$ and $\tau(t_i)=t^{r_i}$ for $1\leq i\leq e$.
Further note that $\alpha_j^{r_i}=\alpha_j^{u_i\lambda/\lambda_i}=1$ if $j\neq i$. Hence 
$\alpha^{r_i}=\alpha_i^{u_i\frac{\lambda}{\lambda_i}}=\alpha_i^{1-v_i\lambda_i}=\alpha_i$,
and thus
$\tau(\sigma(t_i))=\tau(\alpha_i\,t_i)=\alpha_i\,t^{r_i}=\alpha^{r_i}t^{r_i}
=\sigma(t^{r_i})=\sigma(\tau(t_i)).$
Therefore $\sigma(\tau(t_1^{a_1}\dots t_e^{a_e}))=\tau(\sigma(t_1^{a_1}\dots t_e^{a_e}))$ for any $a_i\in\NN$ and hence 
$\tau$ is a difference ring homomorphism by linearity. Suppose that $\tau$ is not injective. 
Then there is an $h\in\EE\setminus\{0\}$ with $\tau(h)=0$. But this implies that there are at least two different monomials in $h$, say $h_1=t_1^{b_1}\dots t_e^{b_e}$ and $h_2=t_1^{c_1}\dots t_e^{c_e}$ with $(b_1,\dots,b_e)\neq(c_1,\dots,c_e)$ and $0\leq b_i,c_i<\lambda_i$ such that $\tau(h_1)=\tau(h_2)$ holds. Multiplying $\tau(t_1^{-c_1}\dots t_e^{-c_e})$ on both sides of $\tau(h_1)=\tau(h_2)$ yields $\tau(t_1^{b_1-c_1}\dots t_e^{b_e-c_e})=1$. Hence we can assume that there is an $\vect{0}\neq(a_1,\dots,a_e)\in\NN^e$ with $0\leq a_i<\lambda_i$ for $1\leq i\leq e$ and $\tau(t_1^{a_1}\dots t_e^{a_e})=1$. By definition of $\tau$ we get $t^{u_1\frac{\lambda}{\lambda_1}a_1+\dots+u_e\frac{\lambda}{\lambda_e}a_e}=1$, and hence $\lambda\mid(u_1\frac{\lambda}{\lambda_1}a_1+\dots+u_e\frac{\lambda}{\lambda_e}a_e)=:y$. Let $j$ with $1\leq j\leq e$ be arbitrary but fixed. Then $\lambda_j\mid y$. Further, $\lambda_j\mid u_i\frac{\lambda}{\lambda_i}a_i$ for $1\leq i\leq e$ with $i\neq j$. Hence $\lambda_j$ must divide also $u_j\frac{\lambda}{\lambda_j}a_j$. However, $\lambda_j$ is relatively prime with $u_j$ and $\frac{\lambda}{\lambda_j}$. Hence $\lambda_j$ divides $a_j$. Since $a_j<\lambda_j$, we conclude that $a_j=0$. Thus $a_j=0$ for all $j$, a contradiction. Consequently $\tau$ is injective. We show surjectivity by inverting $\tau$ explicitly. Let $a\in\NN$ with $0\leq a <\lambda$. Now define $d:=\gcd({u_1}\frac{\lambda}{\lambda_1},\dots,{u_e}\frac{\lambda}{\lambda_e})$.
Hence by the iterative application of the extended Euclidean algorithm we can compute $(\tilde{a}_1,\dots,\tilde{a}_e)\in\NN^e$ with $d=\tilde{a}_1\,{u}_1\,\frac{\lambda}{\lambda_1}+\dots+\tilde{a}_e\,{u}_e\,\frac{\lambda}{\lambda_e}$. Thus $\tau(t_1^{\tilde{a_1}}\dots t_e^{\tilde{a_e}})=t^d$. Next, observe that $d$ and $\lambda$ are relatively prime: Let $p$ be a prime number with $p\mid\lambda=\lambda_1\dots\lambda_e$. Then there is a $j$ ($1\leq j\leq e$) with $p\mid\lambda_j$. Since the $\lambda_u$ with $1\leq u\leq e$ are relatively prime, $p\nmid\frac{\lambda}{\lambda_j}$. Further,
since $\lambda_j$ and $u_j$ are relatively prime, $p\nmid u_j$. Thus $p\nmid u_j\frac{\lambda}{\lambda_j}$, hence $p\nmid d$ and thus $\gcd(d,\lambda)=1$. Thus we can compute $\mu,\nu\in\ZZ$ with $\mu\,d+\nu\,\lambda=1$ and get $t^{a\,\mu\,d}=t^{a-a\,\nu\,\lambda}=t^a$ which implies $\tau(t_1^{\mu\,\tilde{a_1}}\dots t_e^{\mu\,\tilde{a_1}})=t^{a\,\mu\,d}=t^a$. By linear extension we obtain the inverse of $\tau$.\\
(3) Suppose that $\dfield{\EE}{\sigma}$ is not an \rE-extension of $\dfield{\FF}{\sigma}$. Then there is an $f\in\EE\setminus\FF$ with $\sigma(f)=f$. Consequently, $\tau(f)=\tau(\sigma(f))=\sigma(\tau(f))$, i.e., $\tau(f)\in\const{\FF[t]}{\sigma}$. Since $\dfield{\FF[t]}{\sigma}$ is an \rE-extension of $\dfield{\FF}{\sigma}$, $\tau(f)\in\const{\FF}{\sigma}$. In particular, $\tau(f)\in\FF$. Since $\tau|_{\FF}=\text{id}_{\FF}$, $f=\tau(f)\in\FF$, a contradiction.
\end{proof}

\noindent Finally, we obtain a simple test that allows one to determine whether several basic $A$-extensions form an \rE-extension.

\begin{proposition}\label{Prop:CharactSeveralRExt}
Let $\dfield{\FF}{\sigma}$ be a constant-stable difference field. Let $\dfield{\EE}{\sigma}$ with $\EE=\FF[t_1]\dots[t_e]$ be an $A$-extension of $\dfield{\FF}{\sigma}$ over $\FF$ of orders $\lambda_1,\dots,\lambda_e$ respectively. Then this is an \rE-extension iff for $1\leq i\leq e$, the $\frac{\sigma(t_i)}{t_i}$ are primitive $\lambda_i$-th roots of unity and $\gcd(\lambda_i,\lambda_j)=1$ for pairwise distinct $i,j$.
\end{proposition}

\begin{proof}
``$\Leftarrow$'' follows by Lemma~\ref{Lemma:RSingleIsoMultiple}. 
``$\Rightarrow$'': Suppose that $\dfield{\EE}{\sigma}$ is an \rE-extension of $\dfield{\FF}{\sigma}$. By reordering any \rE-monomial can be moved to the front. Thus by Proposition~\ref{Prop:RExt} all $\sigma(t_l)/t_l$ with $0\leq l\leq e$ are primitive roots of unity. Now suppose that there are $i,j$ with $i<j$ such that $\mu=\gcd(\lambda_i,\lambda_j)\neq1$ holds. Define $r_i:=\lambda_i/\mu\in\NN$ and $r_j=\lambda_j/\mu\in\NN$. Let $\alpha_i=\sigma(t_i)/t_i$ and $\alpha_j=\sigma(t_j)/t_j$. Then $(\alpha_i^{r_i})^{\mu}=1=(\alpha_j^{r_j})^{\mu}$. Summarizing, we are given $r_i$ and $r_j$ with $0<r_i<\lambda_i$ and $0<r_j<\lambda_j$ such that
$\alpha_i^{r_i}$ and $\alpha_j^{r_j}$ are $\mu$th roots of unity. 
Since $\alpha_i$, $\alpha_j$ are primitive roots of unity of orders $\lambda_i,\lambda_j$, the powers of each of them comprise all $\mu$th roots of unity. We can therefore choose $0<r'_i<\lambda_i$ and $0<r'_j<\lambda_j$ such that 
$\alpha_i^{r'_i}=\alpha_j^{r'_j}$ and such that $\alpha_i^{r'_i}$ is a primitive $\mu$th root of unity. Now take $w:=t_j^{r'_j}/t_i^{r'_i}$. Then $\sigma(w)=w$, i.e., $w\in\const{\EE}{\sigma}=\const{\FF}{\sigma}$, and $t_j^{r'_j}=w\,t_i^{r'_i}$. Since $r'_j<\lambda_j$,  $t_j^{\lambda_j}=1$ is not the defining relation of $t_j$; a contradiction. \end{proof}

\begin{example}\label{Exp:SimpleCheck1}
Consider the $A$-extension $\dfield{\KK(x)[y_1][y_2][y_3]}{\sigma}$ of $\dfield{\KK(x)}{\sigma}$ from Example~\ref{Exp:3AMonomials}. Here we have $\sigma(y_i)=\alpha_i\,y_i$ and $y_i^{\lambda_i}=1$ for $i=1,2,3$ with $\lambda_1=2$, $\lambda_2=3$, $\lambda_3=12$ and $\alpha_1=-1$, $\alpha_2=(-1)^{2/3}$, $\alpha_3=(-1)^{1/6}$. Note that $\dfield{\KK(x)}{\sigma}$ is constant-stable, that the $\alpha_1$ and $\alpha_2$ are primitive $\lambda_1$-th and $\lambda_2$-th roots of unity, and that $\gcd(\lambda_1,\lambda_2)=1$. Hence $\dfield{\KK(x)[y_1][y_2]}{\sigma}$ is an \rE-extension of $\dfield{\KK(x)}{\sigma}$ by Proposition~\ref{Prop:CharactSeveralRExt}. Also $\alpha_3$ is a $\lambda_3$-th root of unity, but $\gcd(\lambda_2,\lambda_3)\neq1$ (or $\gcd(\lambda_1,\lambda_3)\neq1$). Thus $y_3$ is not an \rE-monomial by Proposition~\ref{Prop:CharactSeveralRExt}.
\end{example}

\section{Simple difference rings and the constant field}\label{Sec:SimpleRings}

Let $\dfield{\EE}{\sigma}$ be a basic $APS$-extension of a difference field $\dfield{\FF}{\sigma}$ which is constant-stable (see Definition~\ref{Def:ConstantStable}).
In this setting we will show in Subsections~\ref{SubSec:RPSIsSimple} and~\ref{SubSec:SimpleIsRPS}
that the property $\const{\EE}{\sigma}=\const{\FF}{\sigma}$ is equivalent to the property that \notion{$\dfield{\EE}{\sigma}$ is simple}. This means that there is no difference ideal in $\EE$, except the trivial ones: $I=\{0\}$ and $I=\EE$. In particular, if $\dfield{\EE}{\sigma}$ is not simple, we will be able to compute explicitly a non-trivial difference ideal $I$ in $\dfield{\EE}{\sigma}$ in Subsection~\ref{SubSec:NonTrivialI}. Further, we will address the problem of computing a maximal reflexive difference ideal in Subsection~\ref{SubSec:SimpleForSExt}.

\subsection{Basic \rpisiSE-extensions over fields are simple difference rings}\label{SubSec:RPSIsSimple}

\noindent First, we will show in Theorem~\ref{Thm:RPiSiIsSimple} that $\dfield{\EE}{\sigma}$ is simple if $\const{\EE}{\sigma}=\const{\FF}{\sigma}$ holds. Here we rely heavily on a specific term order that can be introduced as follows.

\smallskip

Suppose that we are given a basic $APS$-extension $\dfield{\EE}{\sigma}$ of a difference ring $\dfield{\AA}{\sigma}$ in the form~\eqref{Equ:PASOrdered} where $\KK:=\const{\AA}{\sigma}$ is a field. In addition, we take care of the recursive nature of the $P$-mono\-mials and $S$-monomials. Namely, w.l.o.g.\ we can refine the order among the $P$-extensions and among the $S$-extensions by their \notion{nesting depth}. This means that we can construct the difference ring by a tower of $APS$-extensions
\begin{equation}\label{Equ:DepthOrderedE}
\begin{split}
\EE=&\AA[p_{1,1},p_{1,1}^{-1},\dots,p_{1,\nu_1},p_{1,\nu_1}^{-1}][p_{2,1},p_{2,1}^{-1},\dots,p_{2,\nu_2},p_{2,\nu_2}^{-1}]\dots[p_{\delta,1},p_{\delta,1}^{-1},\dots,p_{\delta,\nu_{\delta}},p_{\delta,\nu_{\delta}}^{-1}]\\
&\hspace*{3cm}[y_1,\dots,y_l]
[s_{1,1},\dots,s_{1,n_1}][s_{2,1},\dots,s_{2,n_2}]\dots[s_{d,1},\dots,s_{d,n_d}]
\end{split}
\end{equation}
with the following properties:
\begin{description}
\item[(1)] for $1\leq i\leq\delta$ and $1\leq j\leq\nu_i$, the $p_{i,j}$ are basic $P$-monomials
with 
\begin{align*}
\tfrac{\sigma(p_{i,j})}{p_{i,j}}\in\PPowers{\AA[p_{1,1},p_{1,1}^{-1},\dots,p_{1,\nu_1},p_{1,\nu_1}^{-1}]\dots[p_{i-1,1},p_{i-1,1}^{-1}\dots,p_{i-1,\nu_{i-1}},p_{i-1,\nu_{i-1}}^{-1}]}{\AA}{(\AA^*)}
\end{align*}
where $\frac{\sigma(p_{i,j})}{p_{i,j}}$ depends at least on one of the $p_{i-1,1},\dots,p_{i-1,\nu_{i-1}}$;
\item[(2)] for $1\leq i\leq l$ the $y_i$ are $A$-monomials of order $\lambda_i$ with $\frac{\sigma(y_i)}{y_i}\in\KK^*$,
\item[(3)] for $1\leq i\leq d$ and $1\leq j\leq n_i$ the $s_{i,j}$ are $S$-monomials with 
\begin{equation*}
\sigma(s_{i,j})-s_{i,j}\in\AA[p_{1,1},\dots,p_{\delta,\nu_{\delta}}^{-1}][y_1,\dots,y_l][s_{1,1},\dots,s_{1,n_1}]\dots[s_{i-1,1},\dots,s_{i-1,n_{i-1}}]
\end{equation*}
where $\sigma(s_{i,j})-s_{i,j}$ depends at least on one of the $s_{i-1,1},\dots,s_{i-1,n_{i-1}}$.
\end{description}
Note that the $P$-monomials in $[p_{1,1},p_{1,1}^{-1},\dots,p_{1,\nu_1},p_{1,\nu_1}^{-1}]$ have nesting depth $1$ since their multiplicands $\frac{\sigma(p_{1,i})}{p_{1,i}}$ are from $\FF$, the $P$-monomials in $[p_{2,1},p_{2,1}^{-1},\dots,p_{2,\nu_2},p_{2,\nu_2}^{-1}]$ have nesting depth $2$, since their multiplicands $\frac{\sigma(p_{2,i})}{p_{2,i}}$ depend non-trivially on $P$-monomials with nesting depth $1$, a.s.o. Similarly, we rearrange the $S$-monomials w.r.t.\ their nesting depth.\\
In the following we say that a basic $APS$-extension (or basic \rpisiSE-extension) $\dfield{\EE}{\sigma}$ of $\dfield{\AA}{\sigma}$ is in \notion{depth-order}, if it is build as given in equation~\eqref{Equ:DepthOrderedE} with the properties (1)--(3).

Consider the set of all monomials 
\begin{equation}\label{Equ:DefineM}
M:=\{p_{1,1}^{\pi_{1,1}}\dots p_{\delta,\nu_{\delta}}^{\pi_{\delta,\nu_{\delta}}}y_1^{\xi_1}\dots y_l^{\xi_l}s_{1,1}^{\sigma_{1,1}}\dots s_{d,n_d}^{\sigma_{d,n_d}}|0\leq\xi_i<\lambda_i,\pi_{i,j}\in\ZZ,\sigma_{i,j}\in\NN\}.
\end{equation}
Then such a depth-ordered basic $APS$-extension introduces naturally the following lexicographical term order $<$ of $M$ where the underlying variable order is induced by the order of the generators given in~\eqref{Equ:DepthOrderedE}. This means that
\begin{equation}\label{Equ:varOrdering}
\begin{split}
p_{1,1}&<p_{1,1}^{-1}<\dots<p_{1,\nu_1}<p_{1,\nu_1}^{-1}<\dots <p_{\delta,1}<p_{\delta,1}^{-1}<\dots
<p_{\delta,\nu_{\delta}}<p_{\delta,\nu_{\delta}}^{-1}<\\
y_1&<y_2<\dots<y_l<s_{1,1}<\dots<s_{1,n_1}<\dots s_{d,1}<s_{d,2}<\dots<s_{d,n_d}.
\end{split}
\end{equation}
In a nutshell, $P$-monomials get the lowest priority where among the different monomials we have $p_{i,j}<p_{i,j}^{-1}$ for all $i,j$. More priority receive $A$-monomials, and the highest priority is granted to $S$ monomials. Finally, the nesting depth is encoded within the term order: the less the $P$- or $S$-monomials depend on the previous elements, the smaller the variable is considered in the term order. Take the lexicographical order of $M$ (see~\eqref{Equ:DefineM}) induced by the depth-order representation~\eqref{Equ:DepthOrderedE}, i.e, given by the variable order~\eqref{Equ:varOrdering}. Then we have for instance
$$p_{1,1}^{1}\,p_{2,1}<p_{1,1}^{10}\,p_{2,1}<p_{1,1}^{-1}\,p_{2,1}<p_{2,1}^2<y_1<p_{2,1}\,y_2\,\,s_{1,1}<s_{d,1}<s_{1,1}^3\,s_{d,1}.$$
Obviously, $<$ forms a total order on $M$.
The leading monomial of $f\in\EE\setminus\{0\}$, i.e., the \notion{largest monomial} in $f$ w.r.t.\ $<$ is denoted by $\lm(f)\in M$. In addition, the \notion{leading coefficient} of $f\in\EE$, i.e., the coefficient of $\lm(f)$ in $f$ is denoted by $\lc(f)\in\AA\setminus\{0\}$. 

\noindent We emphasize that the descending chain condition holds: for any ideal $I$ and any sequence $\langle f_i\rangle_{i\geq0}$ with non-zero entries from $I$ there does not exist an infinite chain $\lm(f_1)>\lm(f_2)>\lm(f_3)\dots.$
This fact is ensured by $p_{i,j}<p_{i,j}^{-1}$ for all $i,j$, i.e., eventually we obtain an $n\in\NN$ with $f_n=1$ which is the smallest element in this term order. In other words, any ideal has an element $f$ such that $\lm(f)$ is minimal.

\smallskip

\noindent The introduced term order is not admissible (for admissible term orders see~\cite{Pauer:97}). Nevertheless, we can exploit the following crucial property.

\begin{lemma}\label{Lemma:LMStable} 
Let $\dfield{\EE}{\sigma}$ be a basic $APS$-extension of $\dfield{\AA}{\sigma}$ which is depth-ordered and let $<$ be its lexicographical order on the set of monomials. 
Suppose that $\frac{\sigma(p)}{p}\in\AA^*$ for all $P$-monomials.
Then for any $f\in\EE\setminus\{0\}$ we have that
$\lm(\sigma(f))=\lm(f)$.
\end{lemma}
\begin{proof}
Let $\EE$ be given as in~\eqref{Equ:DepthOrderedE} with $\delta=1$. We will show the statement by induction on $d$. If $d=0$, there are no $S$-monomials in $\EE$. Let $f\in\EE\setminus\{0\}$ and take any monomial $m$ of $f$. This means that $m=p_{1,1}^{\pi_{1,1}}\dots p_{1,\nu_{1}}^{\pi_{1,\nu_1}}y_1^{\xi_1}\dots y_l^{\xi_l}$ for some $\pi_{i,j}\in\ZZ$ and $\xi_{i,j}\in\NN$. Then $\sigma(m)=g\,m$ for some $g\in\AA^*$. Hence the monomials in $f$ are precisely the monomials in $\sigma(f)$ which implies that $\lm(\sigma(f))=\lm(f)$ holds.\\
Now let $\EE$ as given in~\eqref{Equ:DepthOrderedE} with $d>1$ and suppose that the lemma holds for $d-1$. For a given $f\in\EE\setminus\{0\}$ we will show that $\lm(\sigma(f))=\lm(f)$ holds. Let 
$$\HH=\AA[p_{1,1},p_{1,1}^{-1},\dots,p_{1,\nu_1},p_{1,\nu_1}^{-1}][y_1,\dots,y_l][s_{1,1},\dots,s_{1,n_1}]\dots[s_{d-1,1},\dots,s_{d-1,n_{d-1}}].$$ 
Then $\dfield{\EE}{\sigma}$ with $\EE=\HH[s_{d,1}\dots,s_{d,n_d}]$ is a \sigmaSE-extension of $\dfield{\HH}{\sigma}$ with $\beta_i:=\sigma(s_{d,i})-s_{d,i}\in\HH$ for $1\leq i\leq n_d$. Now let $<_1$ be the lexicographical order of the basic $S$-extension $\dfield{\EE}{\sigma}$ of $\dfield{\HH}{\sigma}$. In particular, we denote by $M_1=[s_{d,1}\dots,s_{d,n_d}]$ the set of monomials and $\lm_1(f)\in M_1$ denotes the leading monomial w.r.t.\ $<_1$. 
Then we can write
\begin{equation}\label{Equ:WritefMonomial}
f=h\,\lm_1(f)+w
\end{equation}
with $h\in\HH\setminus\{0\}$ and $w\in\EE$ where all monomials of $w$ from $M_1$ are smaller than $\lm_1(f)$ w.r.t.\ $<_1$. Note that 
\begin{equation}\label{Equ:flmconnection}
\lm(f)=\lm(h)\,\lm_1(f)
\end{equation}
due to the variable order~\eqref{Equ:varOrdering}.
Observe further that for any $m=s_{d,1}^{l_1}\dots s_{d,n_d}^{l_{n_d}}\in M_1$ with $l_i\in\NN$ for $1\leq i\leq n_d$ we have that
$$\sigma(m)=(s_{d,1}+\beta_1)^{l_1}\dots(s_{d,n_d}+\beta_{n_d})^{l_{n_d}}=m+v$$
for some $v\in\EE$ where for each monomial in $v$ at least one exponent of the $s_{d,i}$ with $1\leq i\leq n_d$ is smaller than $l_i$. Thus all monomials of $v$ from $M_1$ are smaller than $m$ w.r.t.\ $<_1$. Together with~\eqref{Equ:WritefMonomial} we get
$$\sigma(f)=\sigma(h)\sigma(\lm_1(f))+\sigma(w)=\sigma(h)\,\lm_1(f)+\tilde{w}$$
for some $\tilde{w}\in\EE$ where all monomials of $\tilde{w}$ from $M_1$ are smaller than $\lm_1(f)$ w.r.t.\ $<_1$. By our induction assumption, we have that $\lm(\sigma(h))=\lm(h)$ and due to~\eqref{Equ:varOrdering} it follows that
$\sigma(f)=\lc(h)\lm(h)\,\lm_1(f)+w'$
for some $w'\in\EE$ where all terms of $w'$ from $M$ are smaller w.r.t\ $<$ than $\lm(h)\lm_1(f)$. Hence $\lm(\sigma(f))=\lm(h)\,\lm_1(f)$ and with~\eqref{Equ:flmconnection} it follows that $\lm(\sigma(f))=\lm(f)$.
\end{proof}

\noindent In order to obtain our first main result in Theorem~\ref{Thm:RPiSiIsSimple}, we need the following slight variation of Theorem~\ref{Thm:ShiftQuotient}.

\begin{corollary}\label{Cor:InversePart}
Let $\dfield{\EE}{\sigma}$ be a basic \rpisiSE-extension of $\dfield{\AA}{\sigma}$ with~\eqref{Equ:APSOrdered} where $\AA$ is an integral domain. Let   $f_1\in\AA\setminus\{0\}$ and 
$f_2=u\,y_1^{n_1}\dots y_l^{n_l}p_1^{z_1}\dots p_r^{z_r}$
with $u\in\AA\setminus\{0\}$, $n_i\in\NN$ for $1\leq i\leq l$ and $z_i\in\ZZ$ for $1\leq i\leq r$. Then for any $f\in\EE\setminus\{0\}$ with $f_1\,\sigma(f)=f_2\,f$, we have that $f=\phi\,w$
for some $w\in\AA\setminus\{0\}$ and $\phi\in\EE^*$.
\end{corollary}

\begin{proof}
Take the field of fractions $\HH$ of $\AA$ and extend $\sigma$ to the field automorphism $\fct{\sigma'}{\HH}{\HH}$ with $\sigma'(\frac{p}{q})=\frac{\sigma(p)}{\sigma(q)}$ for $p\in\AA$ and $q\in\AA\setminus\{0\}$. Then $\dfield{\HH}{\sigma}$ is a difference ring extension of $\dfield{\AA}{\sigma}$.
In particular, construct the basic $APS$-ring extension $\dfield{\tilde{\EE}}{\sigma}$ of $\dfield{\HH}{\sigma}$ 
with 
$\tilde{\EE}=\HH[y_1,\dots,y_l][p_{1},p_{1}^{-1},\dots,p_{r},p_{r}^{-1}]
[s_{1},\dots,s_{v}]$
where the generators have precisely the same shift behaviour as in the extension $\dfield{\EE}{\sigma}$ of $\dfield{\AA}{\sigma}$. Note that $\dfield{\tilde{\EE}}{\sigma}$ is a difference ring extension of $\dfield{\EE}{\sigma}$. In particular, as sets we get
\begin{equation}\label{Equ:SetEqual}
\EE\cap\HH=\AA\cap\HH=\AA.
\end{equation}
Now set $a:=\frac{f_2}{f_1}=\frac{u}{f_1}\,y_1^{n_1}\dots y_l^{n_l}p_1^{z_1}\dots p_r^{z_r}$ with $\frac{u}{f_1}\in\HH^*$, and let $f\in\EE\setminus\{0\}\subseteq\tilde{\EE}\setminus\{0\}$.
By Theorem~\ref{Thm:ShiftQuotient} (with $\EE$ and $\FF$ replaced by $\tilde{\EE}$ and $\HH$, respectively) we conclude that~\eqref{Equ:ProductSol} holds for some $w\in\HH^*$, $\xi_i\in\NN$ and $\pi_i\in\ZZ$. Set $\phi=y_1^{\xi_1}\dots y_l^{\xi_l} p_1^{\pi_1}\dots p_r^{\pi_r}\in\EE^*$.
Then $f=w\,\phi$. Since $f\in\EE$, $w\in\EE$ and thus $w\in(\EE\cap\HH)\setminus\{0\}$. Consequently, $w\in\AA\setminus\{0\}$ by~\eqref{Equ:SetEqual}. This proves the corollary.
\end{proof}

\begin{theorem}\label{Thm:RPiSiIsSimple}
Let $\dfield{\EE}{\sigma}$ be a basic \rpisiSE-extension of a difference field $\dfield{\FF}{\sigma}$. Then $\dfield{\EE}{\sigma}$ is a simple difference ring.
\end{theorem}
\begin{proof}
Let $\EE$ be given as in~\eqref{Equ:DepthOrderedE} with $\AA=\FF$ and let $I$ be a difference ideal in $\EE$ with $I\neq\{0\}$. We will show that $I=\EE$. 
Consider the basic $APS$-extension $\dfield{\EE}{\sigma}$ of $\dfield{\HH}{\sigma}$ with
$\HH=\FF[p_{1,1},p_{1,1}^{-1},\dots,p_{1,\nu_1},p_{1,\nu_1}^{-1}]\dots[p_{\delta-1,1},p_{\delta-1,1}^{-1},\dots,p_{\delta-1,
\nu_{\delta-1}},p_{\delta-1,\nu_{\delta-1}}^{-1}]$ and
\begin{equation}\label{Equ:SubOrderExt}
\EE=\HH[p_{\delta,1},p_{\delta,1}^{-1},\dots,p_{\delta,\nu_{\delta}},p_{\delta,\nu_{\delta}}^{-1}][y_1,\dots,y_l]
[s_{1,1},\dots,s_{1,n_1}]\dots[s_{d,1},\dots,s_{d,n_d}].
\end{equation}
Let $<_1$ be its lexicographical order induced by the variable order given by the order of generators in~\eqref{Equ:SubOrderExt} and let $\lm_1$ and $\lc_1$ be the corresponding leading monomial and coefficient functions. Now take $f\in I$ with $f\neq0$ where $\lm_1(f)$ is minimal w.r.t.\ $<_1$. This implies that $f\in\HH$ which can be seen as follows. We have that $\lm_1(f)=p_{\delta,1}^{\pi_{\delta,1}}\dots p_{\delta,\nu_{\delta}}^{\pi_{\delta,\nu_{\delta}}}y_1^{\xi_1}\dots y_l^{\xi_l}s_{1,1}^{\sigma_{1,1}}\dots s_{n,n_d}^{\sigma_{n,n_d}}$ for some $\pi_{i,j}\in\ZZ$, $\xi_i\in\NN$ and $\sigma_{i,j}\in\NN$. Take $\tilde{\alpha}=p_{\delta,1}^{\pi_{\delta,1}}\dots p_{\delta,\nu_{\delta}}^{\pi_{\delta,\nu_{\delta}}}y_1^{\xi_1}\dots y_l^{\xi_l}$ and define
$\alpha:=\frac{\sigma(\tilde{\alpha})}{\tilde{\alpha}}\in\HH^*$.
Further, take  
$\tilde{f}:=\lc_1(f)\in\HH\setminus\{0\}$ and define
$$h:=\tilde{f}\,\sigma(f)-\sigma(\tilde{f})\,\alpha\,f\in I.$$
By Lemma~\ref{Lemma:LMStable} it follows that $\lm_1(\tilde{f}\,\sigma(f))=\lm_1(\sigma(\tilde{f})\,\alpha\,f))$ and by construction we have that $\lc_1(\tilde{f}\,\sigma(f))=\lc_1(\sigma(\tilde{f})\,\alpha\,f))$. Thus the leading monomial $\lm_1(f)$ cancels in $h$ and therefore $\lm_1(h)<\lm_1(f)$. By the choice of $f$ it follows that $h=0$. Consequently, $\tilde{f}\,\sigma(f)=\sigma(\tilde{f})\,\alpha\,f$.
Note that $\HH$ is an integral domain. By Corollary~\ref{Cor:InversePart} (after reordering the monomials to the form~\eqref{Equ:APSOrdered}) it follows that $f=\phi\,w$
for some $\phi\in\EE^*$ and $w\in\HH\setminus\{0\}$. Hence
$w=\frac{1}{\phi}f\in I$. By the choice of $f$ we have that $\lm_1(f)$ is minimal among all nonzero elements of $I$. Consequently, $\lm_1(f)\leq\lm_1(w)$. Hence by the definition of our lexicographical order induced by the variable order given by the order of generators in~\eqref{Equ:SubOrderExt}, we conclude with $w\in\HH$ that\footnote{\label{Footnote:SimpleProof}If $\delta=1$ (i.e., $\frac{\sigma(p)}{p}\in\FF^*$ for all $P$-monomials $p$), we are done: in this case we have $f\in\HH=\FF$ with $0\neq f\in I$ which  implies $I=\EE$; further, Corollary~\ref{Cor:InversePart} is obsolete and one needs only Theorem~\ref{Thm:ShiftQuotient}.}  $f\in\HH$.\\
Now let $<$ be the lexicographical order with~\eqref{Equ:varOrdering} and let $\lm$ be the corresponding leading monomial function. Take $f'\in I\setminus\{0\}$ such that $\lm(f')$ is minimal. For any $f_1\in\HH$ and $f_2\in\EE\setminus\HH$ we have that $\lm(f_1)<\lm(f_2)$. Since we have showed that $f\in\HH$ for $f\in I$ with minimal $\lm_1(f)$, it follows that $f'\in\HH$. If $f'\in\FF$, we conclude that $I=\EE$ which proves the theorem. Now suppose that $f'\notin\FF$. We will show that this leads to a contradiction.
Let $k$ be maximal s.t.\ $f'$ depends on $p_{k,1},\dots,p_{k,\nu_k}$ (with positive or negative powers). Define $\HH'=\FF[p_{1,1},p_{1,1}^{-1},\dots,p_{1,\nu_1},p_{1,\nu_1}^{-1}]\dots[p_{k-1,1},p_{k-1,1}^{-1},\dots,p_{k-1,\nu_{k-1}},p_{k-1,\nu_{k-1}}^{-1}]$ and $\EE'=\HH'[p_{k,1},p_{k,1}^{-1},\dots,p_{k,\nu_{k}},p_{k,\nu_{k}}^{-1}]$, i.e., $f'\in\EE'\setminus\HH'$. Now we repeat the arguments from above. Let $<_2$ be the lexicographical order for the depth-ordered extension $\dfield{\EE'}{\sigma}$ of $\dfield{\HH'}{\sigma}$ and let $\lm_2$ and $\lc_2$ be the corresponding leading monomial and coefficient functions.  Now set $\tilde{f'}:=\lc_2(f')\in\HH'\setminus\{0\}$ and  $\alpha'=\frac{\sigma(\lm_2(f'))}{\lm_2(f')}\in(\HH')^*$, and define
$h':=\tilde{f'}\,\sigma(f')-\sigma(\tilde{f})\,\alpha'\,f'\in I.$
As above it follows with Lemma~\ref{Lemma:LMStable} that $\lm_2(h')<_2\lm_2(f')$. 
Note that for any $f_1,f_2\in\EE'\setminus\{0\}$ with $\lm_2(f_1)<_2\lm_2(f_2)$ we have that $\lm(f_1)<\lm(f_2)$. Thus $\lm(h')<\lm(f')$ which implies that $h'=0$ by the minimality of $f'$. Thus $\tilde{f'}\,\sigma(f')=\sigma(\tilde{f'})\,\alpha'\,f'$.
Since $\HH'$ is an integral domain, we use again Corollary~\ref{Cor:InversePart}. Thus  $f'=\phi'\,w'$
for some $w'\in\HH'\setminus\{0\}$, $\phi'\in(\EE')^*$. Hence $w'=\frac{1}{\phi'}f'\in I$. Since $w'\in\HH'$ and $f'\notin\EE\setminus\HH'$, $\lm(w')<\lm(f')$. This contradicts the minimality of $f'$.
\end{proof}

\begin{remark}
A related result can be deduced from~\cite{Singer:97} in the setting of Picard-Vessiot extensions. 
Suppose that we are given a difference ring extension $\dfield{\EE}{\sigma}$ of $\dfield{\FF}{\sigma}$ which contains precisely all solutions of a linear homogeneous difference equation with coefficients from $\FF$ and where $\const{\EE}{\sigma}=\const{\FF}{\sigma}$. Then by Cor.~1.24 of~\cite{Singer:97}, $\dfield{\EE}{\sigma}$ is simple (or equivalently, $\dfield{\EE}{\sigma}$ is a Piccard-Vessiot extension of $\dfield{\FF}{\sigma}$), iff there are no non-zero nilpotent elements and if the constant field of $\FF$ is algebraically closed.
Note that the first condition holds automatically for basic \rpisiSE-extensions by Lemma~\ref{Lemma:SInverseElements}; however, the second condition that $\const{\FF}{\sigma}$ is algebraically closed is not needed to apply Theorem~\ref{Thm:RPiSiIsSimple} for the class of basic \rpisiSE-extensions. Note further that Picard-Vessiot extensions and basic \rpisiSE-extensions cover the common class of basic \rpisiSE-extensions where the \piE-monomials are not nested, i.e., where their shift-quotients are from $\FF$. This subclass enables one to model nested sum expressions over hypergeometric products and more generally over simple products (see Section~\ref{Sec:Application}). In particular, it allows one to represent the so-called d'Alembertian solutions~\cite{Abramov:94,Abramov:96} of a linear recurrence, but also the so-called Liouvillian solutions~\cite{Singer:99} by exploiting ideas from~\cite{Reutenauer:12,Petkov:2013}. We notice that for this class the proof of Theorem~\ref{Thm:RPiSiIsSimple} simplifies significantly (see Footnote~\ref{Footnote:SimpleProof}).\\ 
In general, the results of both approaches complement each other: basic \rpisiSE-extensions, in particular nested \piE-monomials, cannot be described by a Picard-Vessiot extension in general, i.e., are not a solution of a homogeneous linear recurrence with coefficients from $\FF$. Conversely, solutions of linear recurrences being not d'Alembertian (or Liouvillian) cannot be expressed within the \rpisiSE-approach.
\end{remark}

\subsection{Simple and basic $APS$-extensions over fields are \rpisiSE-extensions}\label{SubSec:SimpleIsRPS}

Now we want to treat the other implication. Let $\dfield{\EE}{\sigma}$ be a basic $APS$-extension of a constant-stable difference field $\dfield{\FF}{\sigma}$. We will show in Theorem~\ref{Thm:EquivSimpleConst} that $\const{\EE}{\sigma}=\const{\FF}{\sigma}$ if $\dfield{\EE}{\sigma}$ is simple. 
More precisely, we will assume $\const{\FF}{\sigma}\neq\const{\EE}{\sigma}$ and will show that there is a non-trivial difference ideal $I$ in $\EE$. Here we will exploit

\begin{lemma}\label{Lemma:GetDIdeal}
Let $\dfield{\EE}{\sigma}$ be a difference ring with $h\in\EE$ where $\sigma(h)=u\,h$ for some $u\in\EE$. Then $I:=h\,\EE$ is a difference ideal. In particular, $I=\EE$ if $h\in\EE^*$, and $\{0\}\subsetneq I\subsetneq \EE$ if $h\neq0$ and $h\notin\EE^*$.
\end{lemma}
\begin{proof}
Obviously, $I$ is an ideal. If $h\in\EE^*$, $1\in I$ and therefore $I=\EE$. Otherwise, if $h\notin\EE^*$ and $h\neq0$, we have that $I\neq\{0\}$. In addition, $I\neq\EE$, since otherwise $1\in I$, thus $h\,u=1$ for some $u\in\EE$ and hence $h\in\EE^*$. Finally, we show that $I$ is a difference ideal. Let $f\in I$. Then there is a $w\in\EE$ with $f=w\,h$. Since $\sigma(w)\,u\in\EE$, $\sigma(f)=\sigma(w\,h)=\sigma(w)\,u\,h=\,(\sigma(w)\,u)\,h\in I$. Consequently, $I$ is a difference ideal. 
\end{proof}

\noindent Now we are going to find/compute an $h$ described in Lemma~\ref{Lemma:GetDIdeal}. In particular, we will need that $h\notin\EE^*$. In this regard, we will utilize the following lemma.

\begin{lemma}\label{Lemma:LinearNotInvertible}
Let $\AA\ltr{t}$ be a ring of Laurent polynomials, let $a,b\in\AA^*$ and $m\in\ZZ\setminus\{0\}$. Then $a+b\,t^{m}\notin\AA\ltr{t}^*$.
\end{lemma}
\begin{proof}
Denote $f:=a+b\,t^m$. Let $h=\sum_{i=l}^rh_it^i\in\AA\ltr{t}$ with $l\leq r$ and $h_l\neq0\neq h_r$ satisfy $1=h\,f=\sum_{i=l}^ra\,h_i\,t^i+\sum_{i=l+m}^{r+m}b\,h_{i-m}\,t^i$. If $m<0$ then $b\,h_l\,t^{l+m}$ is the smallest possible term in $h\,f$ and $a\,h_r\,t^{r}$ is the largest possible term in $h\,f$. Since $l+m<r$ it follows that $b\,h_l=0$ or $a\,h_r=0$. If $m>0$ then $a\,h_l\,t^l$ is the smallest possible term in $h\,f$ and $b\,h_r\,t^{r+m}$ is the largest possible term in $h\,f$. 
Since $l<r+m$, $a\,h_l=0$ or $b\,h_r=0$. In either case, we have a contradiction with $a,b\in\AA^*$ and $h_l,h_r\neq0$.
\end{proof}

\noindent The next three propositions will provide an $h$ as proposed in Lemma~\ref{Lemma:GetDIdeal} for the scenarios, that one of the $S$-, $P$-, $A$-extensions is not a \sigmaSE-, \piE-, \rE-extension, respectively.

\begin{proposition}\label{Prop:ConstSigmaCase}
Let $\dfield{\EE}{\sigma}$ be a basic $APS$-extension of $\dfield{\FF}{\sigma}$ with~\eqref{Equ:APSOrdered} where $\AA=\FF$ is a field. Suppose that $s_{u}$ is not a \sigmaSE-monomial but that $y_1,\dots,y_{l},p_1,\dots,p_{r},s_{1},\dots,s_{u-1}$ are \rpisiSE-monomials; let $\HH=\FF\lr{y_1}\dots\lr{y_l}\lr{p_1}\dots\lr{p_r}\lr{s_{1}}\dots\lr{s_{u-1}}$. Then there is an $f\in\HH$ with $s_{u}-f\in\const{\EE}{\sigma}\setminus\{0\}$. No such element is a unit.
\end{proposition}
\begin{proof}
Let $\beta:=\sigma(s_{u})-s_{u}\in\HH$. Since $s_{u}$ is not a \sigmaSE-monomial, there is an $f\in\HH$ with 
\begin{equation}\label{Equ:TeleNotSimple}
\sigma(f)-f=\beta
\end{equation}
by Theorem~\ref{Thm:RPSCharacterization}. Thus $\sigma(s_{u}-f)=s_{u}-f$, i.e., $s_{u}-f\in\const{\HH[s_{u}]}{\sigma}\setminus\{0\}$. By part~(1) of Lemma~\ref{Lemma:SInverseElements} it follows that $\HH$ is reduced and thus by part~(2) of Lemma~\ref{Lemma:SInverseElements} we conclude that 
$\EE^*=\HH^*$. 
Thus any $s_{u}-f$ with $f\in\HH$ is not a unit.
\end{proof}

\begin{proposition}\label{Prop:ConstPiCase}
Let $\dfield{\EE}{\sigma}$ be a basic $APS$-extension of $\dfield{\FF}{\sigma}$ with~\eqref{Equ:APSOrdered} where $\AA=\FF$ is a field. Suppose that $p_{u}$ is not a \piE-monomial but that $y_1,\dots,y_{l},p_1,\dots,p_{u-1}$ are \rE\piE-monomials. Then there are
$\xi_1,\dots,\xi_{l}\in\NN$, $\pi_1,\dots,\pi_{u}\in\ZZ$ with $\pi_u\neq0$ and $w\in\FF^*$ with $1-w\,y_1^{\xi_1}\dots y_l^{\xi_l}p_1^{\pi_1}\dots p_{u}^{\pi_u}\in\const{\EE}{\sigma}\setminus\{0\}$. No such element is a unit.
\end{proposition}

\begin{proof}
Let $\HH=\FF\lr{y_1}\dots\lr{y_l}\lr{p_1}\dots\lr{p_{u-1}}$ and let $\alpha:=\sigma(p_{u})/p_{u}$. Note that $\alpha=u\,p_1^{z_1}\dots p_{u-1}^{z_{u-1}}$ with $z_i\in\ZZ$ and $u\in\FF^*$.
Since $p_{u}$ is not a \piE-monomial, there exists an $f\in\HH\setminus\{0\}$ and an $m\in\ZZ\setminus\{0\}$ such that
\begin{equation}\label{Equ:NotRPCharact}
\sigma(f)=\alpha^m\,f
\end{equation}
holds by Theorem~\ref{Thm:RPSCharacterization}.
By Theorem~\ref{Thm:ShiftQuotient}, we conclude that~\eqref{Equ:ProductSol} holds
for some $w\in\FF^*$, $\xi_i\in\NN$ and $\pi_i\in\ZZ$. In particular, $f\in\EE^*$. Consequently, we get
$\tilde{h}:=f\,p_u^{-m}=w\,y_1^{\xi_1}\dots y_l^{\xi_l}p_1^{\pi_1}\dots p_{u-1}^{\pi_{u-1}}p_u^{-m}$
with 
$\sigma(\tilde{h})=\sigma(f)\,\sigma(p_u^{-m})=f\,p_u^{-m}=\tilde{h}.$
Clearly, $1-\tilde{h}\in\const{\EE}{\sigma}\setminus\{0\}$. Now reorder the generators in $\EE$ by putting $p_u$ on top, say we get $\EE=\HH\ltr{p_u}$ for a ring $\HH$ which contains the other generators of $\EE$. Note that $f\in\HH^*$. Thus $1-f\,p_u^{-m}=1-\tilde{h}$ is not a unit in $\HH\ltr{p_u}=\EE$ by Lemma~\ref{Lemma:LinearNotInvertible}.  
\end{proof}

\begin{proposition}\label{Prop:ConstRCase}
Let $\dfield{\EE}{\sigma}$ be a basic $APS$-extension of $\dfield{\FF}{\sigma}$ as given in~\eqref{Equ:APSOrdered} ($\AA=\FF$) where $\dfield{\FF}{\sigma}$ is a constant-stable difference field. Suppose that $y_{u}$ is not an \rE-monomial but that $y_1,\dots,y_{u-1}$ are \rE-monomials. Then there are $\xi_1,\dots,\xi_{u}\in\NN$ with $0<\xi_u<\ord(y_u)$ and $1-y_1^{\xi_1}\dots y_u^{\xi_u}\in\const{\EE}{\sigma}\setminus\{0\}$. No such element is a unit.
\end{proposition}
\begin{proof}
Let $\HH=\FF\lr{y_1}\dots\lr{y_{u-1}}$ and let $\alpha:=\frac{\sigma(y_{u})}{y_{u}}\in\const{\FF}{\sigma}^*$ with $\lambda:=\ord(y_u)$. We start as in the proof of Proposition~\ref{Prop:ConstPiCase}. Since $y_{u}$ is not an \rE-monomial, there exists an $f\in\HH\setminus\{0\}$ and an $m\in\NN$ with $0<m<\lambda$ such that~\eqref{Equ:NotRPCharact} holds
by Theorem~\ref{Thm:RPSCharacterization}.
By Theorem~\ref{Thm:ShiftQuotient}, we conclude that~\eqref{Equ:ProductSol}
holds with $l=u-1$ and $r=0$ for some $w\in\FF^*$ and $\xi_i\in\NN$. Consequently, we get
$\tilde{h}:=f\,y_u^{-m}=w\,y_1^{\xi_1}\dots y_{u-1}^{\xi_{u-1}}y_u^{-m}$
with 
$\sigma(\tilde{h})=\sigma(f)\,\sigma(y_u^{-m})=f\,y_u^{-m}=\tilde{h}$.
Define 
\begin{equation}\label{Equ:DefineaForRProp}
a:=y_1^{\xi_1}\dots y_{u-1}^{\xi_{u-1}} y_u^{-m}.
\end{equation}
Then 
$\sigma(a)=u\,a$
with $u=w/\sigma(w)\in\FF^*$. Since $\dfield{\EE}{\sigma}$ is a basic $APS$-extension of $\dfield{\FF}{\sigma}$, the $\frac{\sigma(y_i)}{y_i}\in\const{\FF}{\sigma}^*$ with $1\leq i\leq u$ are roots of unity, and thus $u=\frac{\sigma(a)}{a}\in\const{\FF}{\sigma}^*$ is a root of unity. More precisely, let $\lambda_i=\ord(y_i)$ and set $\rho=\lcm(\lambda_1,\dots,\lambda_u)>0$. Then $a^{\rho}=u^{\rho}=1$. 
Thus with $\sigma(w)=\frac{1}{u}\,w$ we get $\sigma^{\rho}(w)=u^{-\rho}\,w=w$, i.e., $w\in\const{\FF}{\sigma^{\rho}}$. Since $\dfield{\FF}{\sigma}$ is constant-stable, $w\in\const{\FF}{\sigma}$. With $w\,a=\tilde{h}=\sigma(\tilde{h})=\sigma(w)\,\sigma(a)=w\,\sigma(a),$
we conclude that $\sigma(a)=a$. In particular, $\sigma(1-a)=1-a$. Now take any $a$ as given in~\eqref{Equ:DefineaForRProp} with $1\leq m<\lambda$. As above we can take a $\rho>0$ with $a^{\rho}=1$.
We show that $1-a$ is not a unit. Clearly, $1-a\neq0$ (otherwise, $a=1$ or equivalently $y_u^m=y_1^{\xi_1}\dots y_{u-1}^{\xi_{u-1}}$ with $0<m<\lambda$ and thus $y_u^{\lambda}=1$ would not be the defining relation of our $A$-monomial). Now define $b=1+a+\dots+a^{\rho-1}$. Note that the term $a$ given in~\eqref{Equ:DefineaForRProp} has the coefficient $1$. Hence the term $a^i$ with $i\geq1$ has again the coefficient $1$. This implies that the arising monomials cannot cancel themselves and thus $b\neq0$. 
Finally, observe that $(1-a)\,b=1-a^{\rho}=0$. Thus $1-a$ is a zero divisor which implies that $1-a\in\const{\EE}{\sigma}\setminus\EE^*$.
\end{proof}

\noindent Combining these propositions with Theorem~\ref{Thm:RPiSiIsSimple} produces the main result of this section.

\begin{theorem}[Characterization of \rpisiSE-extensions (I)]\label{Thm:EquivSimpleConst}
Let $\dfield{\EE}{\sigma}$ be a basic $APS$-extension of a difference field $\dfield{\FF}{\sigma}$. Suppose that $\dfield{\FF}{\sigma}$ is constant-stable or that all $A$-monomials are \rE-monomials.
Then the following two statements are equivalent.
\begin{enumerate}
 \item $\dfield{\EE}{\sigma}$ is a basic \rpisiSE-extension of $\dfield{\FF}{\sigma}$ (i.e., $\const{\EE}{\sigma}=\const{\FF}{\sigma}).$
\item $\dfield{\EE}{\sigma}$ is a simple difference ring.
\end{enumerate}
\end{theorem}
\begin{proof}
$(1)\Rightarrow(2)$ follows by Theorem~\ref{Thm:RPiSiIsSimple} ($\dfield{\FF}{\sigma}$ needs not be constant-stable). What remains to show is the direction $(2)\Rightarrow(1)$. Suppose that $\dfield{\EE}{\sigma}$ is not an \rpisiSE-extension of $\dfield{\FF}{\sigma}$. Reorder the monomials in the form~\eqref{Equ:APSOrdered} (with $\AA=\FF$). If one of the $y_i$ is not an \rE-monomial, we 
may assume in addition that $\dfield{\FF}{\sigma}$ is constant-stable and we can apply Proposition~\ref{Prop:ConstRCase}. Otherwise, if all $y_i$ are \rE-monomials, but one of the $p_i$ is not a \piE-monomial, we apply Proposition~\ref{Prop:ConstPiCase}. If all $y_i$ and $p_i$ are \rE\piE-monomials, we apply Proposition~\ref{Prop:ConstSigmaCase}. In any case, it follows that there exists an $h\in\const{\EE}{\sigma}\setminus(\EE^*\cup\{0\})$. Hence by Lemma~\ref{Lemma:GetDIdeal} it follows that $I:=h\,\EE$ is a difference ideal with $\{0\}\subsetneq I\subsetneq\EE$. Consequently $\dfield{\EE}{\sigma}$ is not simple. 
\end{proof}

\begin{remark}
Lemma~1.8 in~\cite{Singer:97} also provides the implication $(2)\Rightarrow(1)$ in the following general setting: it holds if the extension $\dfield{\EE}{\sigma}$ of $\dfield{\FF}{\sigma}$ is a finitely generated $\FF$-algebra, but it relies on the assumption that $\const{\FF}{\sigma}$ is algebraically closed. 
\end{remark}

\subsection{Constructing non-trivial difference ideals in basic $APS$-extensions}\label{SubSec:NonTrivialI}

Let $\dfield{\FF}{\sigma}$ be a difference field as introduced in Remark~\ref{Remark:ConstructiveVersion}. Note that such a  difference field is constant-stable; see Remark~\ref{Remark:ConstantStable}. Moreover, let $\dfield{\EE}{\sigma}$ with $\EE=\FF\lr{t_1}\dots\lr{t_e}$ be an $APS$-extension of $\dfield{\FF}{\sigma}$.
Then the existence proof of a non-trivial difference ideal in Theorem~\ref{Thm:EquivSimpleConst} turns into a constructive version. 
Namely, as worked out in Remark~\ref{Remark:ConstructiveVersion} we can check iteratively, if the $APS$-monomials $t_i$ with $i=1,2,3\,\dots$ are \rpisiSE-monomials. In this way we can check if $\dfield{\EE}{\sigma}$ is an \rpisiSE-extension of $\dfield{\FF}{\sigma}$, i.e., if $\dfield{\EE}{\sigma}$ is simple. 
If it is not an \rpisiSE-extension, we will discover a $k$ with $1\leq k\leq e$ such that the $t_1,\dots,t_{k-1}$ are \rpisiSE-monomials, but $t_k$ is not an \rpisiSE-monomial. 
Namely, exactly 
$f\in\HH:=\FF\lr{t_1}\dots\lr{t_{k-1}}$ with~\eqref{Equ:TeleNotSimple} or $f\in\HH$ and $m\neq0$ ($m\in\ZZ$ or $0<m<\ord(y_u)$) with~\eqref{Equ:NotRPCharact} are calculated to verify that $t_k$ is not an \rpisiSE-monomial. Given such a witness, we can now construct the $h\in\const{\EE}{\sigma}\setminus(\EE^*\cup\{0\})$ as given in the Propositions~\ref{Prop:ConstRCase}, \ref{Prop:ConstPiCase} and~\ref{Prop:ConstSigmaCase}. This finally yields
the difference ideal $I=h\,\EE$ with $\{0\}\subsetneq I\subsetneq\EE$ that witnesses that $\dfield{\EE}{\sigma}$ is not a simple difference ring.

\begin{example}
We demonstrate this constructive aspect for the different cases.
\begin{enumerate}
\item Consider the difference ring from Example~\ref{Exp:ConstructRPiSiExt}. We checked that $\dfield{\AA_5}{\sigma}$ is an \rpisiSE-extension of $\dfield{\KK(x)}{\sigma}$ and thus $\dfield{\AA_5}{\sigma}$ is a simple difference ring. However, $t_2$ is not a \sigmaSE-monomial. Thus $\dfield{\AA_6}{\sigma}$ is not simple by Theorem~\ref{Thm:EquivSimpleConst}. 
E.g., we obtain $c$ as given in~\eqref{Equ:ConstantForSigma} with $\sigma(c)=c$. By Proposition~\ref{Prop:ConstSigmaCase}, $c$ is not a unit and thus $I=\langle c\rangle$ is a difference ideal with $\{0\}\subsetneq I\subsetneq \AA_6$ by Lemma~\ref{Lemma:GetDIdeal}.

\item Consider the difference ring from Example~\ref{Exp:3PiMonomials}. We checked that $\dfield{\AA}{\sigma}$ with $\AA=\KK(x)[y]\ltr{p_1}\ltr{p_2}$ is an \rE\piE-extension of $\dfield{\KK(x)}{\sigma}$ and thus $\dfield{\AA}{\sigma}$ is a simple difference ring. However, $p_3$ is not a \piE-monomial and thus $\dfield{\AA\ltr{p_3}}{\sigma}$ is not simple. 
By Lemma~\ref{Lemma:GetDIdeal} and Proposition~\ref{Prop:ConstPiCase}
we get the difference ideal $I=\langle 1-c\rangle$ with $c$ given in~\eqref{Equ:ConstantForPi} where $\{0\}\subsetneq I\subsetneq\AA\ltr{p_3}$.

\item Consider the difference ring from Example~\ref{Exp:3AMonomials}. We checked that $\dfield{\AA}{\sigma}$ with $\AA=\KK(x)[y_1][y_2]$ is an \rE-extension of $\dfield{\KK(x)}{\sigma}$ and thus $\dfield{\AA}{\sigma}$ is a simple difference ring. However, $y_3$ is not an \rE-monomial and thus $\dfield{\AA[y_3]}{\sigma}$ is not simple. By Lemma~\ref{Lemma:GetDIdeal} and Proposition~\ref{Prop:ConstRCase}
we get the difference ideal $I=\langle 1-c\rangle$ with $c$ given in~\eqref{Equ:ConstantForR} where $\{0\}\subsetneq I\subsetneq\AA[y_3]$.
\end{enumerate}
\end{example}

\subsection{Constructing maximal difference ideals in $S$-extensions}\label{SubSec:SimpleForSExt}

Suppose we are given a basic $APS$-extension $\dfield{\EE}{\sigma}$ of a difference field $\dfield{\FF}{\sigma}$ where $\dfield{\EE}{\sigma}$ is not simple. Moreover, assume that we are given a maximal reflexive difference ideal $I$ in $\dfield{\EE}{\sigma}$. Then by Lemma~1.7 in~\cite{Singer:97} this leads to a simple difference ring $\dfield{\EE/I}{\sigma}$. A natural question is how one can construct such an ideal $I$ and how one can design an \rpisiSE-extension $\dfield{\HH}{\sigma}$ of $\dfield{\FF}{\sigma}$  such that $\dfield{\EE/I}{\sigma}\simeq\dfield{\HH}{\sigma}$ holds. This construction is of particular interest in order to construct a Picard-Vessiot extension explicitly.\\
In the following we skip the $AP$-extension case and refer to related results in~\cite{Schneider:05c,DR2} where it is shown that an expression in terms of finitely many hypergeometric products can be represented within $R\Pi$-extensions; for further details and generalizations see Remark~\ref{Remark:ProductTranslation} below. In this regard we refer also to the special case of $c$-finite sequences~\cite{Kauers:08,Singer:16}. Now we assume that we are given an \rE\piE-extension $\dfield{\AA}{\sigma}$ of $\dfield{\FF}{\sigma}$ in which all our products are formulated. Then using the tools from above we can process a tower of $S$-extensions and can construct a maximal difference ideal with the described properties as follows.

\begin{theorem}\label{Thm:ConstructIso}
Let $\dfield{\AA[t_1]\dots[t_e]}{\sigma}$ be an $S$-extension of $\dfield{\AA}{\sigma}$ with $\sigma(t_i)=t_i+\beta_i$ for $1\leq i\leq e$ where $\KK:=\const{\AA}{\sigma}$ is a field. Then there is a maximal reflexive difference ideal $I$ in $\AA[t_1]\dots[t_e]$ and a \sigmaSE-extension $\dfield{\AA[s_1]\dots[s_r]}{\sigma}$ of $\dfield{\AA}{\sigma}$ with a difference ring isomorphism $\fct{\mu}{\AA[t_1]\dots[t_e]/I}{\AA[s_1]\dots[s_r]}.$
In particular, the following holds.
\begin{enumerate}
\item There are $i_1,\dots,i_{e-r}\in\NN$ with $1\leq i_1<i_2<\dots<i_{e-r}\leq e$ and $g_j\in\AA[t_1,\dots,t_{i_j-1}]$ for $1\leq j\leq e-r$ such that $I=\langle t_{i_1}-g_1,t_{i_2}-g_2,\dots,t_{i_{e-r}}-g_{e-r}\rangle$.
\item Let $\{1,2,\dots,e\}\setminus\{i_1,i_2,\dots,i_{e-r}\}=\{k_1,k_2,\dots,k_r\}$ with $k_1<k_2<\dots<k_r$. 
Then $\mu$ can be given by $\mu(f+I)=f$ for $f\in\AA$, $\mu(t_{i_j}+I)=\mu(g_j+I)$ for\footnote{Note that $\mu$ is constructed iteratively; see the proof below. Particularly, the application of $\mu$ to $g_{j}+I$ is defined before the application of $\mu$ to $t_{i_j}+I$.} $1\leq j\leq e-r$ and $\mu(t_{k_j}+I)=s_j+c_j$ for $1\leq j\leq r$; the $c_j\in\KK$ can be freely chosen.
\item If one can solve the telescoping problem in $\dfield{\AA}{\sigma}$, then the ideal $I$, the \sigmaSE-extension $\dfield{\AA[s_1]\dots[s_{r}]}{\sigma}$ of $\dfield{\AA}{\sigma}$ and the isomorphism $\mu$ can be given explicitly.
\end{enumerate}
\end{theorem}

\begin{proof}
We show the theorem by induction on $e$. For $e=0$ we take $I=\langle\,\rangle=\{0\}$ and the statement holds trivially.
Now suppose that the theorem holds for $e\geq0$ extensions. For $\dfield{\EE}{\sigma}$ with $\EE=\AA[t_1]\dots[t_e]$ we can take a maximal ideal $I=\langle t_{i_1}-g_1,t_{i_2}-g_2,\dots,t_{i_{e-r}}-g_{e-r}\rangle$ with $0\leq r\leq e$ which is closed under $\sigma$ and $\sigma^{-1}$. Further, we can take a \sigmaSE-extension $\dfield{\SA}{\sigma}$ of $\dfield{\AA}{\sigma}$ with $\SA=\AA[s_1]\dots[s_{r}]$ together with the difference ring isomorphism $\fct{\mu}{\AA[t_1]\dots[t_e]/I}{\AA[s_1]\dots[s_{r}]}$ as claimed in the theorem. Consider the $S$-extension $\dfield{\EE[t_{e+1}]}{\sigma}$ of $\dfield{\EE}{\sigma}$ with $\sigma(t_{e+1})=t_{e+1}+\beta_{e+1}$.\\ 
$\bullet$\textit{Case 1.} Suppose that there exists a $g\in\SA$ 
with $\sigma(g)-g=\mu(\beta_{e+1}+I)$. Thus with
$g':=\mu^{-1}(g)\in\EE/I$ we get $\sigma(g')=g'+(\beta_{e+1}+I)$. Moreover,
$\sigma(t_{e+1})=t_{e+1}+\beta_{e+1}$ implies $\sigma(t_{e+1}+I)=(t_{e+1}+I)+(\beta_{e+1}+I)$. Consequently, $\sigma((t_{e+1}+I)-g')=(t_{e+1}+I)-g'$. Take $g_{e-r+1}\in\EE$ with $g'=g_{e-r+1}+I$. Obviously, 
\begin{equation}\label{Equ:ge-r+1Prop}
\sigma(g_{e-r+1}+I)=(g_{e-r+1}+\beta_{e+1})+I.
\end{equation}
Moreover, $\sigma(t_{e+1}-g_{e-r+1})=t_{e+1}-g_{e-r+1}+h$ for some $h\in I$ and thus $\sigma^{-1}(t_{e+1}-g_{e-r+1})=t_{e+1}-g_{e-r+1}-\sigma^{-1}(h)$.
Consequently $\tilde{I}:=(t_{e+1}-g_{e-r+1})\EE+I=\langle t_{i_1}-g_1,t_{i_2}-g_2,\dots,t_{i_{e-r}}-g_{e-r},t_{i_{e-r+1}}-g_{e-r+1}\rangle$ with $i_{e-r+1}=e+1$ is a maximal and reflexive difference ideal. 
Hence we can construct the difference ring $\dfield{\EE[t_{e+1}]/\tilde{I}}{\sigma}$. Identifying the elements of $h+I\in\EE/I$ with $h+\tilde{I}$ turns $\dfield{\EE[t_{e+1}]/\tilde{I}}{\sigma}$ into a difference ring extension of $\dfield{\EE/I}{\sigma}$.
Observe that $\fct{\rho}{\EE[t_{e+1}]/\tilde{I}}{\EE/I}$ with $\rho(f+\tilde{I})=f|_{t_{e+1}\to g_{e-r+1}}+I$ forms a ring isomorphism. Hence $\rho$ is a difference ring isomorphism by
\begin{align*}
\rho(\sigma(t_{e+1}+\tilde{I}))=&\rho(t_{e+1}+\beta_{e+1}+\tilde{I})=\rho(t_{e+1}+\tilde{I})+\rho(\beta_{e+1}+\tilde{I})\\
=&g_{e-r+1}+\beta_{e+1}+I\stackrel{\eqref{Equ:ge-r+1Prop}}{=}\sigma(g_{e-r+1}+I)=\sigma(\rho(t_{e+1}+\tilde{I})).
\end{align*}
Consequently, we obtain the difference ring isomorphism $\fct{\mu'}{\EE[t_{e+1}]/\tilde{I}}{\SA}$ with $\mu':=\mu\circ\rho$. Note that the isomorphism is uniquely determined by 
$\mu'(f+\tilde{I})=\mu(f+I)$ for all $f\in\EE$ and $\mu'(t_{e+1}+\tilde{I})=\mu(g_{e-r+1}+I)(=g)$.\\
$\bullet$\textit{Case 2.} Suppose that there is no $g\in\SA$ with $\sigma(g)=g+\mu(\beta_{e+1}+I)$. 
This implies that one can construct the \sigmaSE-extension $\dfield{\SA[s_{r+1}]}{\sigma}$ of $\dfield{\SA}{\sigma}$ with $\sigma(s_{r+1})=s_{r+1}+\mu(\beta_{e+1}+I)$. Consider $I$ as an ideal in $\EE$ and in $\EE[t_{e+1}]$. Then $\dfield{\EE[t_{e+1}]/I}{\sigma}$ turns into a difference ring which is a difference ring extension of $\dfield{\EE/I}{\sigma}$.
Moreover, we get the ring isomorphism $\fct{\mu'}{\EE[t_{e+1}]/I}{\SA[s_{r+1}]}$ defined by $\mu'|_{\EE/I}=\mu$ and $\mu'(t_{e+1}+I)=s_{r+1}+c$; here $c\in\KK$ can be chosen arbitrarily. Finally, with
$$\mu'(\sigma(t_{e+1}+I))=\mu'(t_{e+1}+\beta_{e+1}+I)=s_{r+1}+c+\mu(\beta_{e+1}+I)=\sigma(s_{r+1}+c)=\sigma(\mu'(t_{e+1}+I))$$
we conclude that $\mu'$ is a difference ring isomorphism. Thus parts~(1) and (2) are proven.\\ Now
suppose that one can solve the telescoping problem in $\dfield{\AA}{\sigma}$. 
By the induction assumption we can construct $I$, $\mu$ and $\dfield{\SA}{\sigma}$.
By iterative application of~\cite[Thm.~7.1]{Schneider:16a} (or Theorem~\ref{Thm:RPSReduction} below) one can solve the telescoping problem in $\dfield{\SA}{\sigma}$. Hence one can decide algorithmically, if there exists a $g\in\SA$ with $\sigma(g)=g+\mu(\beta_{e+1}+I)$. This turns parts~(1) and~(2) to constructive versions and proves part (3) of the theorem.
\end{proof}

\begin{example}\label{Exp:FindIsomorphism}
Consider the difference ring from Example~\ref{Exp:MainDRNaiveDef}. In the following we set $\AA=\QQ(n)(x)[y]\ltr{p_1}\ltr{p_2}$ where $\const{\AA}{\sigma}=\QQ(n)$ and apply the construction given in the proof of Theorem~\ref{Thm:ConstructIso} for the $S$-extension $\dfield{\AA[t_1][t_2][t_3]}{\sigma}$ of $\dfield{\AA}{\sigma}$. As observed in Example~\ref{Exp:ConstructRPiSiExt} the $S$-monomial $t_1$ with 
$\sigma(t_1)=t_1+\beta_1$
and $\beta_1=\frac{-y}{x+1}$ is a \sigmaSE-monomial. Thus we are in case 2 and construct the \sigmaSE-extension $\dfield{\AA[s_1]}{\sigma}$ of $\dfield{\AA}{\sigma}$ with 
\begin{equation}\label{Equ:s1Exp}
\sigma(s_1)=s_1+\beta_1=s_1+\frac{-y}{x+1}.
\end{equation}
Moreover, we set $I_0:=\langle\rangle=\{0\}$ and define the difference ring isomorphism $\fct{\mu}{\AA[t_1]/I_0}{\AA[s_1]}$ with $\mu(f+I_0)=f$ for all $f\in\AA$ and $\mu(t_1+I_0)=s_1+c_1$ for some $c_1\in\QQ(n)$; here we choose $c_1=0$ and get
\begin{equation}\label{Equ:t1IsoThm}
\mu(t_1+I_0)=s_1.
\end{equation}
Now we turn to the $S$-monomial $t_2$ with $\sigma(t_2)=t_2+\beta_2$ and $\beta_2=-\frac{y}{(x+1) (x+2)}$. Here we find $g=\frac{2 (x+1) s_1
-y
}{x+1}+c'\in\AA[s_1]$
with $c'\in\QQ(n)$ such that $\sigma(g)-g=\mu(\beta_2)=\beta_2$ holds. Thus we are in case~1 of our construction. In the following we specialize to $c'=1$. Hence we take $g':=\mu^{-1}(g)=\frac{2 (x+1)t_1
-y
}{x+1}+1+I_0$ and obtain $g_2=\frac{2 (x+1)t_1
-y
}{x+1}+1\in\EE$ with $g'=g_2+I_0$. This gives the reflexive difference ideal 
\begin{equation}\label{Equ:MaxIdealExp}
I=\langle t_2-g_2\rangle=\{h\big(t_2-\tfrac{2 (x+1) t_1
-y
}{x+1}+1)\big|h\in\AA[t_1]\},
\end{equation}
and we can construct the difference ring extension $\dfield{\AA[t_1][t_2]/I}{\sigma}$ of $\dfield{\AA[t_1]}{\sigma}$ together with the difference ring isomorphism $\fct{\mu'}{\AA[t_1][t_2]/I}{\AA[s_1]}$ with $\mu'(f+I)=\mu(f+I_0)$ for all $f\in\EE[t_1]$ and $\mu'(t_2+I)=\mu(g'+I_0)=\mu(g')=g$. For simplicity we use again $\mu$ instead of $\mu'$ and get
\begin{equation}\label{Equ:t2IsoThm}
\mu(t_2+I)=\tfrac{2 (x+1) s_1-y}{x+1}+1.
\end{equation}
Finally, we turn to the $S$-monomial $t_3$ with $\sigma(t_3)=t_3+\beta_3$ and $\beta_3=\tfrac{1-(x^2+3 x+2) y t_ 2}{(x+1)^2 (x+2)}$. This time we do not find a $g\in\AA[s_1]$ with $\sigma(g)-g=\mu(\beta_3+I)=\frac{3
+x-(x+1) (x+2) y(2 s_1+1)}{(x+1)^2 (x+2)}$. Being in case~2 we construct the \sigmaSE-extension $\dfield{\AA[s_1][s_2]}{\sigma}$ of $\dfield{\AA[s_1]}{\sigma}$ with 
\begin{equation}\label{Equ:s2Exp}
\sigma(s_2)=s_2+\mu(\beta_3+I)=s_2+\tfrac{3
+x-(x+1) (x+2) y(2 s_1+1)}{(x+1)^2 (x+2)}.
\end{equation}
Considering $I$ as an ideal in $\AA[t_1][t_2][t_3]$ we can extend $\mu$ to $\fct{\mu}{\AA[t_1][t_2][t_3]/I}{\AA[s_1][s_2]}$ with $\mu(t_3+I)=s_2+c_2$ for some $c_2\in\QQ(n)$. Here we specialize $c_2$ to $0$, i.e, we define
\begin{equation}\label{Equ:t3IsoThm}
\mu(t_3+I)=s_2.
\end{equation}
Summarizing, we constructed the maximal reflexive difference ideal $I$ in $\dfield{\AA[t_1][t_2][t_3]}{\sigma}$ given in~\eqref{Equ:MaxIdealExp} and we constructed the \sigmaSE-extension $\dfield{\AA[s_1][s_2]}{\sigma}$ of $\dfield{\AA}{\sigma}$ with~\eqref{Equ:s1Exp} and~\eqref{Equ:s2Exp} such that $\fct{\mu}{\AA[t_1][t_2][t_3]/I}{\AA[s_1][s_2]}$ is a difference ring isomorphism with $\mu(f+I)=f$ for all $f\in\AA$, and~\eqref{Equ:t1IsoThm} with $I_0$ replaced by $I$, \eqref{Equ:t2IsoThm} and~\eqref{Equ:t3IsoThm}. Note that the construction of $\mu$ (with the choices $c_1=0$, $c'=1$ and $c_2=0$) encodes the identities
\begin{equation}\label{Equ:IdsForEi}
\begin{split}
E_1(k)=&E_1(k),\quad\quad E_2(k)=2 \tsum_{j=1}^k\frac{(-1)^i}{i}
-\frac{(-1)^k}{k+1}+1,\\[-0.1cm]
E_3(k)=&\tsum_{i=1}^k\Big(\frac{-1}{i (1+i)}
+\frac{(-1)^i (1+i)}{i (1+i)} \Big(
        1
        +2\displaystyle\tsum_{j=1}^i\frac{(-1)^j}{j}
\Big)\Big),
\end{split}
\end{equation}
respectively; here the left and right hand sides are the representations in $\dfield{\AA[t_1][t_2][t_3]}{\sigma}$ and $\dfield{\AA[s_1][s_2]}{\sigma}$, respectively.
From the viewpoint of symbolic summation further details will be given in Example~\ref{Exp:SymbolicSummation}. Notably we will explain how the $c',c_1,c_2$ are determined and how the elements from $\dfield{\AA[s_1][s_2]}{\sigma}$ are reinterpreted to get~\eqref{Equ:IdsForEi}. 
\end{example}

\section{Idempotent representations of basic \rpisiSE-extensions}\label{Sec:Copies}

In~\cite[Corollary~1.16]{Singer:97} (compare also~\cite{Singer:08}) it has been worked out that a Picard-Vessiot ring extension, and more generally, a finitely generated difference ring extension $\dfield{\EE}{\sigma}$ of $\dfield{\FF}{\sigma}$ can be decomposed in the form 
\begin{equation}\label{Equ:decomposition}
\EE=e_0\,\EE\oplus e_2\,\EE\oplus\dots\oplus e_{n-1}\,\EE
\end{equation}
for some idempotent elements $e_s\in\EE$ (i.e., $e_s^2=e_s$) with $0\leq s<n$ such that $e_s\,\EE$ forms an integral domain. 
It will be convenient to denote by $s\mmod n$ with $s\in\ZZ$ the unique value $l\in\{0,\dots,n-1\}$ with $n\mid s-l$.
Then given the representation~\eqref{Equ:decomposition} for a simple difference ring $\dfield{\EE}{\sigma}$, it has been shown for $0\leq s<n$ that $\dfield{e_s\,\EE}{\sigma^n}$ is a simple difference ring and that $\sigma$ is a difference ring isomorphism between $\dfield{e_s\,\EE}{\sigma^n}$ and $\dfield{e_{s+1\mmod n}\,\EE}{\sigma^n}$.\\ 
Inspired by this result, we will elaborate in Theorem~\ref{Thm:Interlacing} such a decomposition for basic $RPS$ (and \rpisiSE-extensions) in a very explicit form. This representation will lead to further characterizations of \rpisiSE-extensions in Theorem~\ref{Thm:EquivSimpleCopy}.

Let $\dfield{\FF}{\sigma}$ be a difference field with $\KK=\const{\FF}{\sigma}$, and take a basic $RPS$-extension
$\dfield{\EE}{\sigma}$ of $\dfield{\FF}{\sigma}$ of the form
\begin{equation}\label{Equ:SingleRPS}
\EE=\FF[y]\lr{t_1}\dots\lr{t_e}
\end{equation}
where $y$ is an \rE-monomial of order $n$ with $\alpha:=\frac{\sigma(y)}{y}$ and where the $t_i$ with $1\leq i\leq e$ are $PS$-monomials with $\sigma(t_i)=\alpha_i\,t_i+\beta_i$ (note that either $\alpha_i=1$ or $\alpha_i\in\PPowers{\EE}{\FF}{(\FF^*)}$ with $\beta_i=0$). We remark that the case that several basic \rE-monomials are involved 
can be reduced to this situation by Lemma~\ref{Lemma:RSingleIsoMultiple}.
Now take 
$$\tilde{e}_s=\tilde{e}_s(y):=\tprod_{\substack{j=0\\j\neq n-1-s}}^{n-1}(y-\alpha^{j})$$
for $0\leq s<n$. Since $\alpha$ is a primitive root of unity by
Proposition~\ref{Prop:RExt}, $\tilde{e}_s(\alpha^{n-1-s})\neq0$. Thus we can define 
\begin{equation}\label{Equ:SpecificEi}
e_s=e_s(y):=\frac{\tilde{e}_{s}(y)}{\tilde{e}_{s}(\alpha^{n-1-s})}
\end{equation}
for $0\leq s<n$. 
By construction it follows that for $0\leq s<n$ and $0\leq i<n$ we have that
\begin{equation}\label{Equ:BasicEiProp}
e_s(\alpha\,y)=e_{s+1\mmod n}\quad\quad\text{ and }\quad\quad
e_s(\alpha^{i})=\begin{cases}1&\text{if }i=n-1-s\\ 0&\text{if }i\neq n-1-s.\end{cases}
\end{equation}
First, we will show that the $e_s\in\EE$ are idempotents, pairwise orthogonal and sum up to~$1$. For this result and later considerations we will use the following simple lemma.

\begin{lemma}\label{Lemma:RootImplEqual}
Let $\AA$ be an integral domain and consider the ring $\AA[y]$ subject to the relation $y^n=1$ with $n>1$. Let $f=\sum_{i=0}^{n-1}f_i\,y^i$ and $g=\sum_{i=0}^{n-1}g_i\,y^i$ with $f_i,g_i\in\AA$ and  $f(\lambda_i)=g(\lambda_i)$ for distinct $\{\lambda_1,\dots,\lambda_n\}\subseteq\AA$. Then $f=g$.  
\end{lemma}
\begin{proof}
Consider the ring of polynomials $\AA[z]$ and define $\tilde{f}=\sum_{i=0}^{n-1}f_iz^i, 
\tilde{g}=\sum_{i=0}^{n-1}g_iz^i\in\AA[z]$. Take $h:=\tilde{f}-\tilde{g}\in\AA[z]$. Then $h(\lambda_i)=0$ for all $0\leq i<n$. Since $\AA$ is an integral domain, we have $h=0$.  Therefore $\tilde{f}=\tilde{g}$ and hence $f=g$. 
\end{proof}


\begin{lemma}\label{Lemma:PropertiesOnEi}
Let $\FF$ be a field and let $\alpha$ be a primitive $n$th root of unity. Let $\FF[y]$ be a ring subject to the relation $y^n=1$. Then the
$e_0,\dots,e_{n-1}$ defined in~\eqref{Equ:SpecificEi} are idempotent, they are pairwise orthogonal and we have that $e_0+\dots+e_{n-1}=1$.
\end{lemma}
\begin{proof}
Let $0\leq r<n$ and define $w(y)=e_r(y)^2=\sum_{i=0}^{n-1}h_iy^i$
for some $h_i\in\FF$. By~\eqref{Equ:BasicEiProp} we have that
$w(\alpha^j)=e_r(\alpha^j)^2=e_r(\alpha^j)$ for $0\leq j<n$, i.e., $e_r(y)$ and $w(y)$ have $n$ distinct common evaluation points. By Lemma~\ref{Lemma:RootImplEqual} we get
$e_r(y)=w(y)=e_r(y)^2$.\\
For $0\leq i,j<n$ and $i\neq j$ it is easy to verify that factors in $e_i(y)\,e_j(y)$ produce $y^n-1$ which shows that $e_i(y)\,e_j(y)=0$ holds. Hence the $e_i$ are pairwise orthogonal. Finally, define $f(y)=e_0(y)+\dots+e_{n-1}(y)-1$. Then $f(\alpha^i)=0$ for $0\leq i<n$ by~\eqref{Equ:BasicEiProp}. Thus $0=f=e_0+\dots+e_{n-1}-1$ by Lemma~\ref{Lemma:RootImplEqual}.
\end{proof}

\noindent Given these idempotent elements, we get the decomposition~\eqref{Equ:decomposition} where the $e_s\,\EE$ with $0\leq s<n$ are rings with the multiplicative identity $e_s$. We will elaborate in the following Theorem\footnote{The simplest case $\EE=\FF[y]$ has been handled also in~\cite[Cor.~3.35]{Erocal:11}.}\ref{Thm:Interlacing} that the $\dfield{e_s\,\EE}{\sigma^n}$ are $PS$-extensions of $\dfield{\FF}{\sigma^n}$. In particular, we will show that they form \pisiSE-extensions if $\dfield{\EE}{\sigma}$ is an \rpisiSE-extension of $\dfield{\FF}{\sigma}$. 

\begin{theorem}\label{Thm:Interlacing}
Let $\dfield{\EE}{\sigma}$ be a basic $RPS$-extension of a difference field $\dfield{\FF}{\sigma}$ with~\eqref{Equ:SingleRPS} where $y$ is an \rE-monomial of order $n$ with $\alpha=\frac{\sigma(y)}{y}$. Let $e_0,\dots,e_{n-1}$ be the idempotent, pairwise orthogonal elements defined in~\eqref{Equ:SpecificEi} that sum up to one. Then:
\begin{enumerate}
\item We get the direct sum~\eqref{Equ:decomposition} of the rings $e_s\,\EE$ with the multiplicative identities $e_s$.
\item We have that $e_s\,\EE=e_s\,\tilde{\EE}$ with the integral domain
\begin{equation}\label{Equ:TildeE}
\tilde{\EE}:=\FF\lr{t_1}\dots\lr{t_e}.
\end{equation}
\item For all $0\leq s<n$, $\dfield{e_s\,\tilde{\EE}}{\sigma^n}$ is a $PS$-extension of $\dfield{e_s\,\FF}{\sigma^n}$. 
\item $\sigma$ is a difference ring isomorphism between $\dfield{e_s\,\tilde{\EE}}{\sigma^n}$ and $\dfield{e_{s+1\mmod n}\tilde{\EE}}{\sigma^n}$.
\item If $\dfield{\EE}{\sigma}$ is an \rpisiSE-extension of $\dfield{\FF}{\sigma}$, $\dfield{e_s\,\tilde{\EE}}{\sigma^n}$ is a \pisiSE-extension of $\dfield{e_s\,\FF}{\sigma^n}$ for $0\leq s<n$. Further, if $\dfield{\FF}{\sigma}$ is constant-stable, $\const{e_s\,\EE}{\sigma^n}=e_s\,\const{\FF}{\sigma}$.
\end{enumerate}
\end{theorem}

\noindent One can deduce Theorem~\ref{Thm:Interlacing} by taking~\cite[Corollary~1.16]{Singer:97}, specializing it to the above form with some extra arguments and using Theorem~\ref{Thm:EquivSimpleConst}. However, this result can be proven directly with straightforward arguments. In the remaining part of this section we will tackle the different sub-statements and will illustrate them by concrete examples. Special emphasis will be put on technical details that are relevant for concrete calculations in Section~\ref{Sec:Telescoping}. In addition, further characterizations of \rpisiSE-extensions will be provided.

\smallskip

\noindent\textit{Proof of parts (1) and (2) of Theorem~\ref{Thm:Interlacing}:}
Any $f\in\EE$ can be written in the form $f=\sum_{i=0}^{n-1}f_i\,y^i$ where $f_i$ are entries from the integral domain~\eqref{Equ:TildeE}.
Using this representation consider the map
$\fct{\rho}{\EE}{\EE}$ defined by
$\sum_{i=0}^{n-1}f_iy^i\mapsto \sum_{i=0}^{n-1}e_i\,\tilde{f}_i$
with $\tilde{f}_i=f(\alpha^{n-1-i})$. Now
define $\tilde{f}(y):=\rho(f)$. Then observe that
$\tilde{f}(\alpha^{n-1-s})=\sum_{i=0}^{n-1}\tilde{f}_i\,e_i(\alpha^{n-1-s})=\tilde{f}_{s}=f(\alpha^{n-1-s})$
for all $0\leq s<n$. Since $\alpha$ is primitive, i.e., all $\alpha^s$ with $0\leq s<n$ are distinct, $f=\tilde{f}$ by Lemma~\ref{Lemma:RootImplEqual}. In other words, 
if we consider $\EE$ as an $\tilde{\EE}$-module with the basis $y^0,\dots,y^{n-1}$, then also $e_0,\dots,e_{n-1}$ is a basis and $\rho$ is a basis transformation. In particular, this produces the direct sum
\begin{equation*}
\EE=e_0\,\tilde{\EE}\oplus\dots\oplus e_{n-1}\,\tilde{\EE}
\end{equation*}
of modules. Note that $e_s\,\tilde{\EE}$ with $0\leq s<n$ is an integral domain with the multiplicative identity $e_s$. Hence we obtain a direct sum of rings.
In particular, for $f\in\EE$ we get 
\begin{equation}\label{Equ:NormalizeeiRep}
e_s\,f=e_s\,f(\alpha^{n-1-s})
\end{equation}
with $f(\alpha^{n-1-s})\in\tilde{\EE}$ and thus
$e_s\,\EE=e_s\,\tilde{\EE}$.
This shows parts (1) and (2). $\qed_{(1,2)}$

\medskip

\noindent Before we continue with the remaining parts of Theorem~\ref{Thm:Interlacing}, we summarize the following properties that are useful to carry out basic operations.

\begin{lemma}\label{Lemma:MapSigma}
Let $\dfield{\EE}{\sigma}$ be a basic $RPS$-extension of a difference field $\dfield{\FF}{\sigma}$ of the form~\eqref{Equ:SingleRPS}. Let $f=\sum_{i=0}^{n-1}e_i\,f_i\in\EE$ and $g=\sum_{i=0}^{n-1}e_i\,g_i\in\EE$ with $f_i,g_i\in\tilde{\EE}=\FF\lr{t_1}\dots\lr{t_e}$. Then:
\begin{enumerate}
\item $f+g=\sum_{i=0}^{n-1} e_i\,(f_i+g_i)$ and $f\,g=\sum_{i=0}^{n-1} e_i(f_i\,g_i)$.
\item $\sigma(f)=e_0\,\sigma(f_{n-1})+e_1\,\sigma(f_0)+\dots+e_{n-1}\,\sigma(f_{n-2})=\sum_{i=0}^{n-1}e_i\,h_i$
with $h_i\in\tilde{\EE}$ where
$h_i=\sigma(f_{i-1\mmod n})|_{y\to\alpha^{n-1-i}}.$
\item $\sigma^n(f)=\sum_{i=0}^{n-1}e_i\,\sigma^n(f_i)=\sum_{i=0}^{n-1}e_i\,h_i$ with $h_i=\sigma^n(f_i)|_{y\to\alpha^{n-1-i}}\in\tilde{\EE}$.
\end{enumerate}
\end{lemma}

\begin{proof}
(1) $f+g=\sum_{i=0}^{n-1} e_i\,(f_i+g_i)$ follows by linearity and $f\,g=\sum_{i=0}^{n-1} e_i\,f_i\,g_i$ follows since the $e_i$ are idempotent and pairwise orthogonal.
By~\eqref{Equ:BasicEiProp} we get $\sigma(e_{n-1})=e_0$ and $\sigma(e_i)=e_{i+1}$ for $0\leq i<n-1$. 
With $w_0(y):=\sigma(f_{n-1})$ and $w_i(y):=\sigma(f_{i-1})$ for $1\leq i<n$ we obtain
$w(y):=\sigma(f)=e_0(y)\,w_0(y)+e_1(y)\,w_1(y)+\dots+e_{n-1}(y)\,w_{n-1}(y).$
With $h_i:=w(\alpha^{n-1-i})=w_i(\alpha^{n-1-i})\in\tilde{\EE}$ we get part~(2). Part~(3) is implied by part~(2). 
\end{proof}

\begin{example}\label{Exp:Structure1}
Take the \pisiSE-field $\dfield{\FF}{\sigma}$ over $\QQ(n)$ with $\FF=\QQ(n)(x)$ and $\sigma(x)=x+1$ and 
consider the \rpisiSE-extension $\dfield{\EE}{\sigma}$ of $\dfield{\FF}{\sigma}$ with $\EE=\FF[y]\ltr{p_1}\ltr{p_2}[s_1][s'_2]$ where $\sigma(y)=-y$, $\sigma(p_1)=2\,p_1$, $\sigma(p_2)=\frac{n-x}{x+1}p_2$, $\sigma(s_1)=s_1+\frac{-y}{x+1}$ and $\sigma(s'_2)=s'_2+\frac{1}{(x+1)^2}$. In this difference ring we get the idempotent elements $e_0=\frac{1-y}{2}$ and $e_1=\frac{1+y}{2}$. Thus
$\EE=e_0\,\EE\oplus e_1\,\EE=e_0\,\tilde{\EE}\oplus e_1\,\tilde{\EE}$ with $\tilde{\EE}=\FF\ltr{p_1}\ltr{p_2}[s_1][s'_2]$. E.g., for $f=-\frac{2 y s_1}{x+1}
-\frac{y}{x+1}
+\frac{x+3}{(x+1)^2 (x+2)}
\in\EE$ we get $f=e_0\,f_0+e_1\,f_1$ with
$f_0=f|_{y\to-1}=\frac{2 s_1}{x+1}
+\frac{x^2+4 x+5}{(x+1)^2 (x+2)}\in\tilde{\EE}$ and $f_1=f|_{y\to1}=-\frac{2 s_1}{x+1}
-\frac{x^2+2 x-1}{(x+1)^2 (x+2)}\in\tilde{\EE}$. Similarly, we have that
$\sigma(f)=e_0\,f'_0+e_1\,f'_1$ with $f'_0=\sigma(f_1)|_{y\to-1}=-\frac{2 s_1}{x+2}
-\frac{x^3+7 x^2+16 x+14}{(x+1) (x+2)^2 (x+3)}$ and $f'_1=\sigma(f_0)|_{y\to1}=\frac{2 s_1}{x+2}
+\frac{x^3+5 x^2+6 x-2}{(x+1) (x+2)^2 (x+3)}
$. Furthermore, we get $\sigma^2(f)=\sigma(e_0\,f'_0+e_1\,f'_1)=
e_0\,f''_0+e_1\,f''_1$ with $f''_0=\sigma(f'_1)|_{y\to-1}=\frac{x^4+11 x^3+45 x^2+81 x+58}{(x+1) (x+2) (x+3)^2 (x+4)}
+\frac{2 s_1}{x+3}$ and 
$f''_1=\sigma(f'_0)|_{y\to1}=-\frac{x^4+9 x^3+25 x^2+19 x-10}{(x+1) (x+2) (x+3)^2 (x+4)}
-\frac{2 s_1}{x+3}$. Note that $f''_0=\sigma^2(f_0)|_{y\to-1}$ and $f''_1=\sigma^2(f_1)|_{y\to1}$.
\end{example}

\noindent\textit{Proof of parts (3) and (4) of Theorem~\ref{Thm:Interlacing}:}
$e_s\,\tilde{\EE}$ is a ring with the multiplicative identity $e_s$ for $0\leq s<n$. Moreover, $\fct{\sigma^n}{e_s\,\tilde{\EE}}{e_s\,\tilde{\EE}}$ is a ring automorphism by part (3) of Lemma~\ref{Lemma:MapSigma}. Hence $\dfield{e_s\,\tilde{\EE}}{\sigma^n}$ is a difference ring. Since $e_s\,\FF$ is a subring of $e_s\,\tilde{\EE}$, $\dfield{e_s\,\tilde{\EE}}{\sigma^n}$ is a difference ring extension of $\dfield{e_s\,\FF}{\sigma^n}$. For $1\leq i\leq e$ we can write
$\sigma^n(t_i)=\tilde{\alpha}_i\,t_i+\tilde{\beta}_i$
with $\tilde{\alpha}_i\in\PPowers{\EE}{\FF}{(\FF^*)}\subseteq\tilde{\EE}$
and $\beta=0$ (if $t_i$ is a $P$-monomial), or $\alpha=1$ and $\tilde{\beta}_i\in\EE$ (if $t_i$ is an $S$-monomial). Now define the ring automorphism $\fct{\sigma_s}{\tilde{\EE}}{\tilde{\EE}}$ for each $0\leq s<n$ with 
\begin{equation}\label{Equ:Sigmas}
\sigma_s(t_i)=\tilde{\alpha}_i\,t_i+(\tilde{\beta}_i|_{y\to\alpha^{n-1-s}})
\end{equation}
for all $1\leq i\leq e$. It is immediate that $\dfield{\tilde{\EE}}{\sigma_s}$ is a $PS$-extension of $\dfield{\FF}{\sigma^n}$, and that for $f\in\tilde{\EE}$ we have
\begin{equation}\label{Equ:DefineSigmaS}
\sigma_s(f)={\sigma^n(f)|_{y\to\alpha^{n-1-s}}}.
\end{equation}
Furthermore, we extend $\sigma_s$ from $ \tilde{\EE}$ to $e_s\tilde{\EE}=e_s\,\EE$ with $\sigma_s(e_s)=e_s$ for $0\leq s<n$. Then $\dfield{e_s\,\tilde{\EE}}{\sigma_s}$ is also a $PS$-extension of $\dfield{e_s\,\FF}{\sigma^n}$, and for $f\in\tilde{\EE}$ we have
$$\sigma_s(e_s\,f)=e_s\sigma_s(f)\stackrel{\eqref{Equ:DefineSigmaS}}{=}e_s{\sigma^n(f)|_{y\to\alpha^{n-1-s}}}\stackrel{\eqref{Equ:NormalizeeiRep}}{=}e_s\sigma^n(f)=\sigma^n(e_s\,f).$$
Summarizing, we obtain the difference ring isomorphisms
\begin{equation}\label{Equ:FirstEquiv}
\dfield{e_s\,\EE}{\sigma^n}\simeq\dfield{e_s\,\tilde{\EE}}{\sigma^n}\simeq\dfield{e_s\,\tilde{\EE}}{\sigma_s}
\end{equation}
where the first isomorphism is given  by rewriting the elements  with~\eqref{Equ:NormalizeeiRep} and the second isomorphism is just the identity map. This completes part (3) of Theorem~\ref{Thm:Interlacing}.\\
By Lemma~\ref{Lemma:MapSigma} it follows that for all $f\in e_s\,\EE$ we have that $\sigma(f)\in e_{s+1\mmod n}\,\EE$. Moreover, for all $f\in e_{s+1\mmod n}\,\EE$ we have that $\sigma^{-1}(f)\in e_s\,\EE$. Furthermore, $\sigma(\sigma^n(f))=\sigma^n(\sigma(f))$ for all $f\in e_s\,\EE$. Hence $\sigma$ is a difference ring isomorphism between $\dfield{e_s\,\EE}{\sigma^n}$ and $\dfield{e_{s+1\mmod n}\,\EE}{\sigma^n}$. This establishes part (4) of Theorem~\ref{Thm:Interlacing}. Note that $\dfield{\tilde{\EE}}{\sigma_{s}}$ and $\dfield{\tilde{\EE}}{\sigma_{s+1\mmod n}}$ are isomorphic by $\fct{\tau_s}{\tilde{\EE}}{\tilde{\EE}}$ with 
$\tau_s(f)=\sigma(f)|_{y\to\alpha^{n-2-s}}.$
\hfill$\qed_{(3,4)}$

\begin{example}\label{Exp:Structure2}
We continue with Example~\ref{Exp:Structure1}. We obtain the $PS$-extension $\dfield{\EE}{\sigma^2}$ of $\dfield{\FF}{\sigma^2}$ with $\sigma^2(x)=x+2$, $\sigma^2(y)=y$ and
\begin{align*}
\sigma^2(p_1)&=4\,p_1,&\sigma^2(s_1)&=s_1-\frac{y}{x+1}
        +\frac{y}{x+2},\\
\sigma^2(p_2)&=\frac{(n
        -x
        ) (-1
        +n
        -x
        )}{(x+1) (x+2)}p_2,&\sigma^2(s'_2)&=s'_2-\frac{1}{(x+1)^2}
        -\frac{1}{(x+2)^2}.
\end{align*}
Furthermore, with $\tilde{\EE}=\FF\lr{p_1}\lr{p_2}[s_1][s'_2]$, we get the $PS$-extension $\dfield{\tilde{\EE}}{\sigma_0}$ of $\dfield{\FF}{\sigma^2}$ with
\begin{align*}
\sigma_0(p_1)&=4\,p_1,&\sigma_0(s_1)&=s_1+\frac{1}{(x+1) (x+2)},\\
\sigma_0(p_2)&=\frac{(n
        -x
        ) (-1
        +n
        -x
        )}{(x+1) (x+2)}p_2,&\sigma_0(s'_2)&=s'_2-\frac{2 x^2+6 x+5}{(x+1)^2 (x+2)^2}
\end{align*}
and the $PS$-extension $\dfield{\tilde{\EE}}{\sigma_1}$ of $\dfield{\FF}{\sigma^2}$ with
\begin{align*}
\sigma_1(p_1)&=4\,p_1,&\sigma_1(s_1)&=s_1-\frac{1}{(x+1) (x+2)},\\
\sigma_1(p_2)&=\frac{(n
        -x
        ) (-1
        +n
        -x
        )}{(x+1) (x+2)}p_2,&\sigma_1(s'_2)&=s'_2-\frac{2 x^2+6 x+5}{(x+1)^2 (x+2)^2}.
\end{align*}
In particular, it follows that $\sigma^2(f)=e_0\,f''_0+e_1\,f''_1$ with $f''_0=\sigma_0(f_0)$ and $f''_1=\sigma_1(f_1)$ as given in Example~\ref{Exp:Structure1}. Furthermore, with $\sigma_s(e_s):=e_s$ we get
$\dfield{e_s\tilde{\EE}}{\sigma_s}\simeq\dfield{e_s\EE}{\sigma^2}$ with $s\in\{0,1\}$. Moreover,
$\dfield{e_0\,\tilde{\EE}}{\sigma^2}$ and $\dfield{e_1\,\tilde{\EE}}{\sigma^2}$ are isomorphic by $\sigma$. In addition, $\dfield{\tilde{\EE}}{\sigma_0}$ and $\dfield{\tilde{\EE}}{\sigma_1}$ are isomorphic by $\fct{\tau_0}{\tilde{\EE}}{\tilde{\EE}}$ with 
$\tau_0(f)=\sigma(f)|_{y\to1}$.
\end{example}

\noindent In order to show part~(5) of Theorem~\ref{Thm:Interlacing} we rely on the following lemma.

\begin{lemma}\label{Lemma:Zero-Divisors}
Let $\dfield{\EE}{\sigma}$ be the $RPS$-extension of a difference field $\dfield{\FF}{\sigma}$ of the form~\eqref{Equ:SingleRPS}. Furthermore,  
let $f=\sum_{i=0}^{n-1}{e_i\,f_i}\in\EE\setminus\{0\}$ with $f_i\in\tilde{\EE}$ where $\tilde{\EE}:=\FF\lr{t_1}\dots\lr{t_e}$. Then:
\begin{enumerate}
\item $f$ is a zero-divisor iff $f_s=0$ for some $0\leq s<n$. 
\item If $f$ is not a zero-divisor, then $f$ is a unit in $Q(\tilde{\EE})[y]$.
\item $f\in\EE$ iff $f_s\in\tilde{\EE}^*=\PPowers{\tilde{\EE}}{\FF}{(\FF^*)}$ for $0\leq s<n$. 
\end{enumerate}
\end{lemma}
\begin{proof}
(1) If $f_s=0$ for some $0\leq s<n$, take $h\in e_s\EE\setminus\{0\}$. Since $f\,h=0$ by Lemma~\ref{Lemma:MapSigma}, $f$ is a zero-divisor. Conversely, suppose that $f_s\neq0$ for all $0\leq s<n$ and let $h=\sum_{i=0}^{n-1}e_i\,h_i$ with $h_i\in\tilde{\EE}$ and $0=f\,h=e_0\,(f_0\,h_0)+\dots+e_{n-1}\,(f_{n-1}\,h_{n-1})$. Hence $f_s\,h_s=0$ and thus $h_s=0$ for all $0\leq s<n$. Therefore $h=0$, and consequently $f$ is not a zero-divisor.\\
(2) Suppose that $f$ is not a zero-divisor. Then $f_i\in\tilde{\EE}\setminus\{0\}$ for $0\leq i<n$ by the first part of the lemma. Hence $f_i$ is invertible in the field of fractions $Q(\tilde{\EE})$ and we define $h=e_0\,f_0^{-1}+\dots+e_{n-1}\,f_{n-1}^{-1}\in Q(\tilde{\EE})[y](=Q(\EE[y])$; compare~\eqref{Equ:QuotRep} below. Multiplying $f$ with $h$ in $Q(\tilde{\EE})[y]$ and using the orthogonality property we obtain 
$f\,h=e_0(f_0\,f_0^{-1})+\dots+e_{n-1}(f_{n-1}\,f_{n-1}^{-1})=e_0+\dots+e_{n-1}=1$. Thus $f$ is invertible in $Q(\tilde{\EE})[y]$.\\
(3) For $h=e_0\,h_0+\dots+e_{n-1}\,h_{n-1}$ with $h_s\in\tilde{\EE}$ for $0\leq s<n$ we have that $f\,h=e_0(f_0\,h_0)+\dots+e_{n-1}(f_{n-1}\,h_{n-1})$. Thus by $1=e_0\,1+\dots+e_{n-1}\,1$ and the unique representation of the decomposition it follows that
$f\,h=1$ iff $h_s\,f_s=1$ for all $0\leq s<n$ iff $f_s\in\EE^*$ for all $0\leq s<n$. Hence $f\in\EE$ iff $f_s\in\EE^*$ for all $0\leq s<n$. $\EE^*=\PPowers{\tilde{\EE}}{\FF}{(\FF^*)}$ follows by part~(3) of Lemma~\ref{Lemma:SInverseElements} and the fact that $\dfield{\tilde{\EE}}{\sigma_s}$ is a $PS$-extension of $\dfield{\FF}{\sigma^n}$.
\end{proof}

\noindent\textit{Proof of part (5) of Theorem~\ref{Thm:Interlacing}:}
Let $\KK:=\const{\FF}{\sigma}=\const{\EE}{\sigma}$. Let $s$ satisfy $0\leq s<n$ and assume that we can take an $e_s\,f\in\const{e_s\,\tilde{\EE}}{\sigma^n}\setminus\const{e_s\,\FF}{\sigma^n}$ with $f\in\tilde{\EE}$. This means that $f\in\tilde{\EE}\setminus\FF$. 
Now define $h_{s+k\mmod n}:=\sigma^k(f)|_{y\to\alpha^{n-1-s-k}}\in \tilde{\EE}$ for $0\leq k<n$. Note that
$\sigma^k(e_{s}\,f)=e_{s+k\mmod n}h_{s+k\mmod n}$ and that $h_i\neq0$ for all $0\leq i<n$. Further, observe that $\sigma^k$ is a difference ring isomorphism between $\dfield{e_s\,\tilde{\EE}}{\sigma^n}$ and $\dfield{e_{s+k\mmod n}\tilde{\EE}}{\sigma^n}$ by part (4) of Theorem~\ref{Thm:Interlacing}. Hence the constant property is carried over and we get $e_{s+k\mmod n}h_{s+k\mmod n}\in\const{e_{s+k\mmod n}\tilde{\EE}}{\sigma^n}$
for $0\leq k<n$. Thus for 
\begin{equation}\label{Equ:hPDecomp}
h:=e_0\,h_0+\dots+e_{n-1}\,h_{n-1}\in\tilde{\EE}[y]\setminus\FF[y]
\end{equation}
we get $\sigma^n(h)=h$. 
This construction with  $\sigma^k(e_s\,h_s)=e_{s+k\mmod n} h_{s+k\mmod n}$ for all $0\leq k<n$ and $h(\alpha^{n-1-i})=h_i\neq0$ for all $0\leq i<n$ will lead to a contradiction.\\
Namely, by part~(1) of Lemma~\ref{Lemma:Zero-Divisors} it follows that $h$ is not a zero-divisor. Even more, by part~(2) of the lemma $h$ is invertible in $\HH[y]$ where $\HH$ is the field of fractions of $\tilde{\EE}$. In particular, also $\sigma^i(h)$ is not a zero-divisor for all $0\leq i<n$.
Now define $g:=h\,\sigma(h)\dots\sigma^{n-1}(h)$. 
Since the $\sigma^i(h)$ are not zero-divisors, $g\neq0$. 
Moreover, we get $h\,\sigma(g)=g\,\sigma^n(h)=g\,h$.
Thus multiplying by $h^{-1}$ in $\HH[y]$ on both sides, we obtain $\sigma(g)=g$ in $\HH[y]$. Since $\sigma(g),g\in\tilde{\EE}[y]$, this identity holds also in $\tilde{\EE}[y]=\EE$.
Consequently, $k:=g\in\KK^*$ and
hence $h\,w=1$ with $w:=\frac1k\sigma(h)\dots\sigma^{n-1}(h)\in\EE$. Therefore $h$ is invertible not only in $\HH[y]$, but in $\EE$.
W.l.o.g.\ suppose that the \piE-monomials in~\eqref{Equ:SingleRPS} are $t_1,\dots,t_r$ and the \sigmaSE-monomials are $t_{r+1},\dots,t_e$.
By part (3) of Lemma~\ref{Lemma:Zero-Divisors}
it follows that $h_s\in\tilde{\EE}^*$, and even more that $h_s=d\,t_1^{\pi_1}\dots t_r^{\pi_r}$ with $\pi_i\in\ZZ$ for $1\leq i\leq r$ and $d\in\FF^*$. Since $h_s$ is free of the $R\Sigma$-monomials, 
it follows that $\sigma^k(h_s)\in\tilde{\EE}$. Thus
the relation $\sigma^k(e_s\,h_s)=e_{s+k\mmod n} h_{s+k\mmod n}$ implies $\sigma^k(h_s)=h_{s+k\mmod n}$ for all $0\leq k<n$. 
Thus~\eqref{Equ:hPDecomp} gives
$h=u\,t_1^{\pi_1}\dots t_r^{\pi_r}$
for some $u\in\FF[y]^*$. There is a $\lambda$ with $1\leq \lambda\leq r$ and $\pi_{\lambda}\neq0$ since $h\notin\FF[y]$. Take the maximal $\lambda$ with this property. Since $\frac{\sigma(t_i)}{t_i}$ is free of $t_{\lambda}$ for all $1\leq i\leq\lambda$, $g=q\,(t_{\lambda}^{\pi_{\lambda}})^{n}$ for some $q\in\FF[y]\lr{t_1}\dots\lr{t_{\lambda-1}}\setminus\{0\}$. Since $g=k\in\KK^*$ and $n>1$, we get $\pi_{\lambda}=0$, a contradiction. Summarizing, $\dfield{e_s\tilde{\EE}}{\sigma^n}$ is a \pisiSE-extension of $\dfield{e_s\FF}{\sigma^n}$ for all $0\leq s<n$.\\
Suppose in addition that $\dfield{\FF}{\sigma}$ is constant-stable. Then it follows that
$\const{e_s\,\EE}{\sigma^n}=\const{e_s\,\FF}{\sigma^n}=e_s\,\const{e_s\,\FF}{\sigma^n}=e_s\,\const{\FF}{\sigma}$ which completes part~(5).
$\qed_{(5)}$

\begin{example}\label{Exp:Structure3}
$\dfield{\tilde{\EE}}{\sigma_0}$ and $\dfield{\tilde{\EE}}{\sigma_1}$ from Example~\ref{Exp:Structure2} are \pisiSE-extensions of $\dfield{\FF}{\sigma^2}$.  
\end{example}

Summarizing, suppose that we are given an \rpisiSE-extension $\dfield{\EE}{\sigma}$ of $\dfield{\FF}{\sigma}$ in the form~\eqref{Equ:SingleRPS} with the idempotent elements $e_0,\dots,e_{n-1}$ given by~\eqref{Equ:SpecificEi}. Then the difference rings~\eqref{Equ:FirstEquiv}
with $0\leq s< n$ are \pisiSE-extensions of $\dfield{e_s\FF}{\sigma^n}$. In particular, the $\dfield{\tilde{\EE}}{\sigma_s}$ with $0\leq s<n$ are \pisiSE-extensions of $\dfield{\FF}{\sigma^n}$. Further, one can write $\dfield{\EE}{\sigma}$ as follows:
\begin{equation}\label{Equ:DiagramCopies}
\xymatrix@!C=0.1cm{
\EE\,=\,e_0\,\tilde{\EE}\ar@/^1pc/^{\sigma}[rr] &\oplus& e_2\,\tilde{\EE}\ar@/^1pc/^{\sigma}[rr] &\oplus& e_3\,\tilde{\EE}\ar@/^1pc/^{\sigma}[rr]&\oplus&\dots\ar@/^1pc/^{\sigma}[rr]&\oplus& e_{n-2}\,\tilde{\EE}\ar@/^1pc/^{\sigma}[rr]&
\oplus& e_{n-1}\,\tilde{\EE}\ar@/^{1.3pc}/^{\sigma}[llllllllll]}.
\end{equation}
In short, shifting $f=e_0\,f_0+\dots+e_{n-1}\,f_{n-1}\in\EE$ with $\sigma$ means that the component $e_sf_s$ is moved cyclically to $\dfield{e_{{s+1}\mmod n}\,\tilde{\EE}}{\sigma^n}$ with $\sigma(e_s\,f_s)=e_{s+1\mmod n}\sigma(f_s)|_{y\to\alpha^{n-2-s}}$. 
We note that the difference ring $\dfield{\EE}{\sigma}$ can be considered as the ``interlacing'' of the difference rings $\dfield{e_s\,\EE}{\sigma^n}=\dfield{e_s\,\tilde{\EE}}{\sigma_s}$. This idea will be made very precise in Lemma~\ref{Lemma:Homoidempotent} by embedding the decomposition~\eqref{Equ:DiagramCopies} into the ring of sequences.

\medskip

\noindent We conclude this section with the following refined characterization of \rpisiSE-extensions.

\begin{theorem}[Characterization of \rpisiSE-extensions (II)]\label{Thm:EquivSimpleCopy}
Suppose that the difference field $\dfield{\FF}{\sigma}$ is constant-stable.
Let $\dfield{\EE}{\sigma}$ be a basic $RPS$-extension of $\dfield{\FF}{\sigma}$ given in the form~\eqref{Equ:SingleRPS} where the \rE-monomial $y$ has order $n$ with $\alpha=\frac{\sigma(y)}{y}$. Let $e_0,\dots,e_{n-1}$ be the idempotent elements defined in~\eqref{Equ:SpecificEi}
Then the following statements are equivalent:
\begin{enumerate}
 \item $\dfield{\EE}{\sigma}$ is a basic \rpisiSE-extension of $\dfield{\FF}{\sigma}$ (i.e., $\const{\EE}{\sigma}=\const{\FF}{\sigma}$).
 \item $\dfield{\EE}{\sigma}$ is a simple difference ring.
 \item $\EE=e_0\,\EE\oplus\dots\oplus e_{n-1}\EE$ where $\dfield{e_s\EE}{\sigma^n}$ is a \pisiSE-extension of $\dfield{e_s\,\FF}{\sigma^n}$ for all $0\leq s<n$.
 \item $\EE=e_0\,\EE\oplus\dots\oplus e_{n-1}\,\EE$ where $\dfield{e_s\,\EE}{\sigma^n}$ is simple for all $0\leq s<n$.
 \item There is an $s\in\{0,\dots,n-1\}$ such that  $\dfield{e_s\EE}{\sigma^n}$ is a \pisiSE-extension of $\dfield{e_s\,\FF}{\sigma^n}$.
 \item There is an $s\in\{0,\dots,n-1\}$ such that  $\dfield{e_s\EE}{\sigma^n}$ is simple.
 \end{enumerate}
\end{theorem}
\begin{proof}
$(1)\Leftrightarrow(2)$ follows by Theorem~\ref{Thm:EquivSimpleConst}. Moreover, $(1)\Rightarrow(3)$ follows by Theorem~\ref{Thm:Interlacing}.\\ 
$(3)\Rightarrow(1)$: Suppose that~(3) holds and let $\KK:=\const{\FF}{\sigma}$. Now let $f=e_0\,f_0+\dots+e_{n-1}\,f_{n-1}\in\EE$ with $f_s\in\tilde{\EE}$ for $0\leq s<n$ such that $\sigma(f)=f$ holds. Then $\sigma^n(f)=f$ and thus $\sigma_s(f_s)=f_s$ for all $0\leq s<n$. 
Since $\dfield{\tilde{\EE}}{\sigma_s}$ is a \pisiSE-extension of $\dfield{\FF}{\sigma^n}$, $f_s\in\const{\FF}{\sigma^n}$. Further, since $\dfield{\FF}{\sigma}$ is constant-stable, $f_s\in\KK$ for all $0\leq s<n$. With $\sigma(f)=f$ we conclude that $c:=f_0=\dots=f_{n-1}\in\KK$ and thus we get $f=(e_0+\dots+e_{n-1})c=c\in\KK$. This shows that the statements~(1)--(3) are equivalent.\\ 
Clearly, also $\dfield{\FF}{\sigma^n}$ is constant-stable. Thus by Theorem~\ref{Thm:EquivSimpleConst} the statements (3) and (4) are equivalent and the statements (5) and (6) are equivalent. Obviously, (3) implies (5).\\ 
$(5)\Rightarrow(3)$: Suppose that (5) holds and take $j$ with $0\leq j<n$. 
By Theorem~\ref{Thm:Interlacing} $\dfield{e_s\,\EE}{\sigma^n}=\dfield{e_s\,\tilde{\EE}}{\sigma^n}$ is isomorphic to $\dfield{e_j\,\EE}{\sigma^n}=\dfield{e_j\,\tilde{\EE}}{\sigma^n}$ by the isomorphism $\sigma^{\lambda_j}$ with $\lambda_j=(j-s)\mmod n\in\{0,\dots,n-1\}$. Let $e_j\,c\in\const{e_j\,\EE}{\sigma^n}$ with $c\in\tilde{\EE}$, i.e., $\sigma^n(e_j\,c)=e_j\,c$. Then $\sigma^n(\sigma^{-\lambda_j}(e_j\,c))=\sigma^{-\lambda_j}(\sigma^n(e_j\,c))=\sigma^{-\lambda_j}(e_j\,c)$ and thus $\sigma^{-\lambda_j}(e_j\,c)\in\const{e_s\,\EE}{\sigma^n}=e_s\,\const{\FF}{\sigma^n}$; the last equality follows since $\dfield{e_s\,\EE}{\sigma^n}$ is a \pisiSE-extension of $\dfield{e_s\,\FF}{\sigma^n}$. Hence $e_j\,c\in e_j\,\const{\FF}{\sigma^n}$ and therefore $\const{e_j\,\EE}{\sigma^n}=e_j\,\const{\FF}{\sigma^n}$. Summarizing, $\dfield{e_j\,\EE}{\sigma^n}$ is a \pisiSE-extension of $\dfield{e_j\,\FF}{\sigma^n}$ for all $0\leq j<n$.
By part (1) of Theorem~\ref{Thm:Interlacing} we get $\EE=e_0\,\EE\oplus\dots\oplus e_{n-1}\,\EE$ which proves part~(3). 
\end{proof}

\section{Embedding of basic \rpisiSE-extensions into the ring of sequences}\label{Sec:Embedding}

In Subsection~\ref{Sec:EmbeddingConstruction} (Theorem~\ref{Thm:RPiSiImpliedEmbedding}) we will 
show that an \rpisiSE-extension $\dfield{\EE}{\sigma}$ of a difference field $\dfield{\FF}{\sigma}$ can be embedded into the difference ring of sequences provided that $\dfield{\FF}{\sigma}$ can be embedded in the ring of sequences. In particular, we will work out that this construction is algorithmic if certain properties hold for the ground field $\FF$. Further, in Subsection~\ref{Sec:CharacteriationEmbedding} (Theorems~\ref{Thm:EquivSimpleConstIso} and~\ref{Thm:EquivSimpleConstIso2}) we will provide illuminating characterizations of \rpisiSE-extensions based on these constructions.

\subsection{The algorithmic construction of $\KK$-embeddings}\label{Sec:EmbeddingConstruction}

Let $\KK$ be a field and consider the set of sequences
$\KK^{\NN}$ with elements $\langle a_n\rangle_{n\geq0}=\langle
a_0,a_1,a_2,\dots\rangle$, $a_i\in\KK$. With component-wise addition and multiplication we obtain a commutative ring; the field $\KK$ can be
naturally embedded by identifying $k\in\KK$ with the sequence
$\vect{k}:=\langle k,k,k,\dots\rangle$; note that $\vect{0}=\langle0,0,0,\dots\rangle$.

\noindent We follow the construction
from~\cite[Sec.~8.2]{AequalB} in order to turn the
shift
\begin{equation}\label{Equ:ShiftOp}
\Shift:{\langle a_0,a_1,a_2,\dots\rangle}\mapsto{\langle
a_1,a_2,a_3,\dots\rangle}
\end{equation}
into an automorphism: we define an equivalence relation
$\sim$ on $\KK^{\NN}$ by $\langle a_n\rangle_{n\geq0}\sim
\langle b_n\rangle_{n\geq0}$ if there exists a $d\geq 0$ such that
$a_n=b_n$ holds for all $n\geq d$. The equivalence classes form a ring
which is denoted by $\seqK$; the elements of $\seqK$ (also called germs) will be
denoted, as above, by sequence notation.
Now it is immediate that
$\fct{\Shift}{\seqK}{\seqK}$ with~\eqref{Equ:ShiftOp} forms a ring automorphism. The difference ring $\dfield{\seqK}{\Shift}$ is called the \notion{(difference) ring of sequences (over $\KK$)}.\\
Let $\dfield{\AA}{\sigma}$ be a difference ring with constant field
$\KK$. Then a difference ring homomorphism (resp.\ difference ring monomorphism) $\fct{\tau}{\AA}{\seqK}$ is called a \notion{$\KK$-homomorphism} (resp.\ \notion{$\KK$-monomorphism or $\KK$-embedding}) if for all $c\in\KK$  we have that
$\tau(c)=\vect{c}$.

\smallskip

In the following we will construct an (injective) $\KK$-homomorphism from a basic \rpisiSE-exten\-sion into $\dfield{\seqK}{\Shift}$ by generalizing the ideas of~\cite{Schneider:10c}. We start with the following simple observation.
If $\fct{\tau}{\AA}{\seqK}$ is a $\KK$-homomorphism, there is
a map $\fct{\ev}{\AA\times\NN}{\KK}$ with
\begin{equation}\label{Equ:EvDef}
\tau(f)=\langle\ev(f,0),\ev(f,1),\dots\rangle
\end{equation}
for all $f\in\AA$ which has the following properties. For all
$c\in\KK$ there is a $\delta\geq0$ with
\begin{align}\label{Ev:Const}
\forall i\geq\delta:&\;\ev(c,i)=c;\\
\intertext{for all $f,g\in\AA$ there is a $\delta\geq 0$ with}
\forall i\geq\delta:&\;\ev(f\,g,i)=\ev(f,i)\,\ev(g,i),\label{Ev:Mult}\\
\forall i\geq\delta:&\;\ev(f+g,i)=\ev(f,i)+\ev(g,i)\label{Ev:Add};\\
\intertext{for all $f\in\AA$ and $j\in\ZZ$ there is a $\delta\geq 0$ with}
\label{Ev:Shift}
\forall i\geq\delta:&\;\ev(\sigma^j(f),i)=\ev(f,i+j);\\
\intertext{and for all $f\in\AA^*$ there is a $\delta\geq 0$ with}
\label{Ev:Zero}
\forall i\geq\delta:&\;\ev(f,i)\neq0.
\end{align}
The last property follows since $f^{-1}\in\AA^*$ and thus $\tau(f)\tau(f^{-1})=\tau(f\,f^{-1})=\tau(1)=\vect{1}$.\\
Conversely, if there is a function $\fct{\ev}{\AA\times\NN}{\KK}$
with~\eqref{Ev:Const},~\eqref{Ev:Mult}, \eqref{Ev:Add}
and~\eqref{Ev:Shift}, then the function $\fct{\tau}{\AA}{\seqK}$ defined
by~\eqref{Equ:EvDef} forms a $\KK$-homomorphism.

Subsequently, we assume that a $\KK$-homomorphism/$\KK$-embedding is always defined by such a function $\ev$; $\ev$ is also called a \notion{defining function of $\tau$}. To take into account the constructive aspects, we
introduce the following functions for $\ev$; compare~\cite{Schneider:10c}.

\begin{definition}
Let $\dfield{\AA}{\sigma}$ be a difference ring and
let $\fct{\tau}{\AA}{\seqK}$ be a $\KK$-homomor\-phism given
by a defining function $\ev$ as introduced in~\eqref{Equ:EvDef}. $\ev$ is called \notion{operation-bounded} by
$\fct{L}{\AA}{\NN}$ if for all $f\in\AA$ and $j\in\ZZ$ with
$\delta=\delta(f,j):=L(f)+\max(0,-j)$ we have~\eqref{Ev:Shift} and
for all $f,g\in\AA$ with $\delta=\delta(f,g):=\max(L(f),L(g))$ we
have~\eqref{Ev:Mult} and~\eqref{Ev:Add}; moreover, we require that for all $f\in\AA$ and all $j\in\ZZ$ we have $L(\sigma^j(f))\leq L(f)+\max(0,-j)$. Such a function is also
called an \notion{$o$-function} for $\ev$.\\ 
Let $M\subseteq\AA\setminus\{0\}$.
$\ev$ is called \notion{zero-bounded} by
$\fct{Z}{M}{\NN}$ if for all $f\in M$ and all $i\geq Z(f)$ we
have $\ev(f,i)\neq0$; such a function is also called a
\notion{$z$-function} of $M$ for $\ev$.
\end{definition}

\noindent Note: a $z$-function of $\AA^*$ (or of any subgroup $G$ of $\AA^*$) for a defining function $\ev$ of $\tau$ always exists by~\eqref{Ev:Zero}. Moreover,
if there is an $o$-function for $\ev$, this implies that
for all $f\in\AA$ there is a $\delta\in\NN$ such that for all $j\in\ZZ$ we have that
$$\forall i\geq\delta+\max(0,-j):\;\ev(\sigma^j(f),i)=\ev(f,i+j).$$
This bound on the shift is a stronger assumption than~\eqref{Ev:Shift}. Subsequently, we call a $\KK$-homomorphism \notion{$\tau$ shift-bounded} if the underlying defining function $\ev$ has an $o$-function.

\medskip

\noindent For concrete applications we will assume that we are given a $\KK$-embedding for the ground field $\dfield{\FF}{\sigma}$. In this article we will restrict to the following two examples.

\begin{example}\label{Exp:Embedding1}
Take the \pisiSE-field $\dfield{\KK(x)}{\sigma}$ over $\KK$ with $\sigma(x)=x+1$ from part~(1) of Example~\ref{Exp:qRat}. We obtain a $\KK$-homomorphism $\fct{\tau}{\KK(x)}{\seqK}$ by introducing the defining function 
\begin{equation}\label{Equ:EvalRat}
\ev(\tfrac{p}{q},k)=\begin{cases}
0&\text{if }q(k)=0\\
\frac{p(k)}{q(k)}&\text{if }q(k)\neq0
\end{cases}
\end{equation}
where $p,q\in\KK[x]$, $q\neq0$ and $p,q$ are co-prime;
here $p(k),q(k)$ is the usual evaluation of polynomials at $k\in\NN$. 
For the $o$-function $L(f)$ we take the minimal non-negative integer $l$ with $q(k+l)\neq0$ for all $k\in\NN$, and as $z$-function we take $Z(f)=L(p\,q)$. By construction $\tau$ is shift-bounded. In addition, since $p(x)$ and $q(x)$ have only finitely
many roots, $\tau(\frac{p}{q})=\vect{0}$ iff
$\frac{p}{q}=0$. Hence $\tau$ is injective.  
Summing up, we have constructed a $\KK$-embedding $\fct{\tau}{\KK(x)}{\seqK}$. In particular, up to the renaming of the elements of $\KK(x)$, the difference field $\dfield{\KK(x)}{\sigma}$ is contained in $\dfield{\seqK}{\Shift}$ as the sub-difference ring $\dfield{\tau(\KK(x))}{\Shift}$. $\dfield{\tau(\KK(x))}{\Shift}$ is also called the \notion{difference field of rational sequences}.
\end{example}

\begin{example}\label{Exp:Embedding2}
Take the \pisiSE-field $\dfield{\FF}{\sigma}$ over $\KK$ from part~(2) of Example~\ref{Exp:qRat}.
Then for $f=\frac{p}{q}$ with $p,q\in\KK[x,x_1,\dots,x_v]$, $q\neq0$ and $p,q$ being co-prime we define
\begin{equation}\label{Equ:EvalMixed}
\ev(f,k)=\begin{cases}
0&\text{if }q(k,q_1^k,\dots,q_v^k)=0\\
\frac{p(k,q_1^k,\dots,q_v^k)}{q(k,q_1^k,\dots,q_v^k)}&\text{if }q(k,q_1^k,\dots,q_v^k)\neq0.
\end{cases}
\end{equation}
There is an algorithm~\cite[Sec.~3.7]{Bauer:99} that determines a
$\delta\in\NN$ with $q(k,q_1^k,\dots,q_v^k)\neq0$ for all $k\geq\delta$. 
Hence for the $o$-function $L(f)$ we define (and compute) the value $\delta\in\NN$ which is minimal such that $q(k,q_1^k,\dots,q_v^k)\neq0$ holds for all $k\geq\delta$; as $z$-function we take $Z(f)=L(p\,q)$. 
Since $p(k,q_1^{k},\dots,q_v^{k})$ and $q(k,q_1^{k},\dots,q_v^{k})$ are non-zero for all $k\geq Z(f)$, $\tau(\frac{p}{q})=\vect{0}$ iff
$\frac{p}{q}=0$. Hence $\tau$ is injective and equipped with computable $o$- and $z$-functions.
Summarizing, we have constructed a $\KK$-embedding $\fct{\tau}{\FF}{\seqK}$. The difference field $\dfield{\tau(\FF)}{\Shift}$ is called the difference field of $q$--mixed rational sequences. Similarly, if we perform this construction without $x$ and with $v=1$ where $q=q_1$, we obtain the \notion{difference field of $q$--rational sequences} $\dfield{\tau(\KK(x_1))}{\Shift}$.
\end{example}

\noindent As base case we will start with our ground fields $\dfield{\FF}{\sigma}$ from above (or more generally with a difference ring $\dfield{\AA}{\sigma}$) with constant field $\KK$ together with a $\KK$-homomorphism $\tau$. Then our goal is to construct a $\KK$-homo\-morphism for a given basic $APS$-extension. Here we will treat each $APS$-monomial separately by applying the following lemma.

\begin{lemma}\label{Lemma:LiftEvToPoly}
Let $\dfield{\AA}{\sigma}$ be a difference ring with constant field $\KK$ and let $G$ be a subgroup of $\AA^*$.
Let $\dfield{\AA\lr{t}}{\sigma}$ be a $G$-basic $APS$-extension of $\dfield{\AA}{\sigma}$ with $\sigma(t)=\alpha\,t+\beta$ ($\alpha=1$ and $\beta\in\AA$ or $\alpha\in G$ and $\beta=0$).
Let
$\fct{\tau}{\AA}{S(\KK)}$ be a $\KK$-homomorphism with a defining function $\ev$ as given in~\eqref{Equ:EvDef} and with a $z$-function of $Z$ of $G$ for $\ev$.
Then the following holds.
\begin{enumerate}
\item Take $c\in\KK$ and $r\in\NN$. If $\beta=0$, we suppose that $c\neq0$ and $r>Z(\alpha)$; if $t^{\lambda}=1$ for some $\lambda>1$, we assume in addition that $c^{\lambda}=1$ holds.
Then one gets a $\KK$-homomorphism  $\fct{\tau'}{\AA\lr{t}}{S(\KK)}$ with a defining function $\ev'$ given by
\begin{align}\label{Equ:DefineEvExt}
\ev'(\sum_{i}f_i\,t^i,k)&=\sum_{i}\ev(f_i,k)\ev'(t,k)^i\quad\forall k\in\NN\\[-0.3cm]
\intertext{with\vspace*{-0.2cm}} 
\label{Equ:SumProdHom}
\ev'(t,k)&=\begin{cases}
\displaystyle c\,\tprod_{i=r}^k\ev(\alpha,i-1)&\text{ if $\sigma(t)=\alpha\,t$}\\
\displaystyle\tsum_{i=r}^k\ev(\beta,i-1)+c&\text{ if  $\sigma(t)=t+\beta$.}
\end{cases}
\end{align}
\item Any other difference ring homomorphism $\fct{\tau''}{\AA\lr{t}}{S(\KK)}$ with $\tau''|_{\AA}=\tau$ has a defining function $\ev''$ of the form~\eqref{Equ:SumProdHom} up to the choice of $r\in\NN$ and $c\in\KK$ with the requirements stated in part (1). 
\item If there is an $o$-function $L$ of $\ev$, 
take any $r\in\NN$ in~\eqref{Equ:SumProdHom} with
\begin{equation}\label{Equ:ChooseLowerBound}
r>\begin{cases}
\max(L(\alpha),Z(\alpha))&\text{ if } \sigma(t)=\alpha\,t\\
L(\beta)&\text{ if } \sigma(t)=t+\beta.
\end{cases}
\end{equation}
For this choice there is an $o$-function $L'$ for $\ev'$ with $L'|_{\AA}=L$ and a $z$-function $Z'$ of $\PPowers{\AA\lr{t}}{\AA}{G}$ for $\ev'$ with $Z'|_{G}=Z$. 
\item If the defining function $\ev$ of $\tau$, and the functions $L$ and $Z$ from part~(3) are computable, then a defining function  $\ev'$ for $\tau'$ and a $o$-function $L'$ and a $z$-function $Z'$ for $\ev'$ as in part (3) are computable. 
\end{enumerate}
\end{lemma}
\begin{proof}
{\bf(1)}:~let $\tau$ be
defined by~\eqref{Equ:EvDef}. First suppose that $\alpha=1$. In this case, take any $r\in\NN$ and $c\in\KK$ and
extend $\ev$ from $\AA$ to $\AA\lr{t}$ by~\eqref{Equ:DefineEvExt} and~\eqref{Equ:SumProdHom}.
Let $j\in\ZZ$. Then we can choose a $\delta\in\NN$ with $\delta\geq r$ such that $\ev(\sigma^j(\beta),k)=\ev(\beta,k+j)$ holds for all $k\geq\delta$. Hence for all $k\geq\delta$ we have that
\begin{equation}\label{Equ:ShiftTEvSum}
\begin{split}
\ev'&(\sigma^j(t),k)=\ev'(t+\ssum{i=0}{j-1}\sigma^i(\beta),k)=\ev'(t,k)+
\ssum{i=0}{j-1}\ev(\sigma^i(\beta),k)\\
&=\ssum{i=r}k\ev(\beta,i-1)+\ssum{i=0}{j-1}\ev(\beta,k+i)+c=\ssum{i=r}{k+j}\ev(\beta,i-1)+c=\ev'(t,k+j).
\end{split}
\end{equation}
Now suppose that $\beta=0$. Thus we take $r\in\NN$ with $r>Z(\alpha)$. This means that $\ev(\alpha,k-1)\neq0$ holds for all $k\geq r$. 
Further, we take $c\in\KK^*$; in particular, if $t^{\lambda}=1$ for some $\lambda>1$, then we assume in addition that $c^{\lambda}=1$.
Now
extend $\ev$ from $\AA$ to $\AA\lr{t}$ by~\eqref{Equ:DefineEvExt} and~\eqref{Equ:SumProdHom}.
By definition, $\ev'(t,k)\neq0$ for all $k\in\NN$. 
Let $j\in\ZZ$. Then there is a $\delta\in\NN$ with $\delta\geq r$ such that
$\ev(\sigma^j(\alpha),k)=\ev(\alpha,k+j)$ holds 
for all $k\geq\delta$. Hence completely analogously to the sum case it follows that for all $k\geq\delta$ we have that \begin{equation}\label{Equ:ShiftTEvProd}
\ev'(\sigma^j(t),k)=\ev'(t,k+j)
\end{equation}
for all $k\geq\delta$.
Further, if we choose $\delta\geq r$ big enough (depending on the
$f_i$ and being larger than the chosen versions for~\eqref{Equ:ShiftTEvSum} or~\eqref{Equ:ShiftTEvProd}), we get~\eqref{Ev:Shift} for $f=\sum_{i}f_it^i$. Summarizing, property~\eqref{Ev:Shift} is established for an $APS$-monomial.\\
Now let $f=\sum_{i=a}^mf_it^i$, $g=\sum_{i=b}^ng_it^i\in\AA\lr{t}$ be arbitrary but fixed. Then we can choose a $\delta\geq0$ big enough
such that for all $k\geq\delta$ we have that
\begin{equation}\label{Equ:SumHom}
\begin{split}
\ev'&(f+g,k)=\ev'(\tsum_i (f_i+g_i)\,t^i,k)=\tsum_i((\ev(f_i,k)+\ev(g_i,k))\ev'(t,k)^i\\
&=\tsum_i\ev(f_i,k)\ev'(t,k)^i+\tsum_i\ev(g_i,k)\ev'(t,k)^i=\ev'(f,k)+\ev'(g,k).
\end{split}
\end{equation}
If $t$ is a $PS$-monomial, it follows for a $\delta\geq0$ chosen large enough that
\begin{equation}\label{Equ:ProdHom}
\begin{split}
\ev'(f\,g,k)&=
\ev'(\tsum_{j=a+b}^{m+n}t^j\tsum_{i}f_ig_{j-i},k)=\tsum_{j=a+b}^{m+n}\ev'(t,k)^j\tsum_{i}\ev(f_ig_{j-i},k)\\
&=
\Big(\tsum_{i=a}^{m}\ev(f_i,k)\ev'(t,k)^i\Big)\Big(\tsum_{j=b}^{n}\ev(g_j,k)\ev'(t,k)^j\Big)=
\ev'(f,k)\ev'(g,k)
\end{split}
\end{equation}
holds for all $k\geq\delta$; outside of the support of $f,g$ the $f_i,g_i$ are zero. In addition, if
$t$ is an $A$-monomial of order $\lambda$, we have $c^{\lambda}=\alpha^{\lambda}=1$ which implies that
$\ev'(t,k)^{\lambda}=c^{\lambda}\prod_{i=r}^k\ev(\alpha^{\lambda},i-1)=1$.
With this property, we can choose a $\delta\geq0$ sufficiently large and can again verify~\eqref{Equ:ProdHom} with $m=n=\lambda-1$. This establishes~\eqref{Ev:Add} and~\eqref{Ev:Mult}.
Moreover,~\eqref{Ev:Const}
holds, since $\ev'|_{\AA\times\NN}=\ev$. Summarizing,
if we define $\fct{\tau'}{\AA\lr{t}}{\seqK}$ by $\tau'(f)=\langle\ev'(f,i)\rangle_{i\geq0}$
for all $f\in\AA\lr{t}$ then $\tau'$ forms a $\KK$-homomorphism.\\
{\bf(2)}: This construction is unique up to $c$ and $r$ in~\eqref{Equ:SumProdHom}. Namely, take any other $\tau_2$ with $\tau_2(f)=\tau(f)$ for $f\in\AA$ and define $T:=\tau_2(t)$; if $\beta=0$, then we require in addition that $T$ is non-zero from a certain point on. Then $S(T)=S(\tau_2(t))=\tau_2(\sigma(t))=\tau_2(\alpha\,t+\beta)=\tau_2(\alpha)T+\tau(\beta)$.
Note that $\tau_2(1)=\vect{1}$ and $\tau_2(0)=\vect{0}$. Hence, if $\alpha=1$, then $S(T)=T+\tau(\beta)$, and therefore $S(T-\tau(t))=T-\tau(t)$, i.e., $T=\tau(t)+\vect{d}$ for some constant $d\in\KK$. Similarly, if $\beta=0$, then $S(T)=\tau(\alpha)T$. Since $\tau'(t)$ is non-zero from the point $r$ on, one can take the inverse $1/\tau'(t)\in\seqK$ and gets $S(\frac{T}{\tau'(t)})=\frac{T}{\tau'(t)}$. Hence $\frac{T}{\tau'(t)}=\vect{d}$ with $d\in\KK$, i.e., $T=\vect{d}\tau'(t)$. Since $T$ is non-zero for almost all entries, $d\neq0$. This shows that $\tau_2$ can be defined by~\eqref{Equ:SumProdHom} up to a constant $d\in\KK$; $d\neq0$ if $\beta=0$. In addition, if $t$ is an $A$-extension of order $\lambda$, it follows that $\vect{1}=\tau_2(1)=\tau_2(t^{\lambda})=\tau_2(t)^{\lambda}=T^{\lambda}=\vect{d}^{\lambda}\tau'(t)^{\lambda}=\vect{d}^{\lambda}\tau'(t^{\lambda})=\vect{d}^{\lambda}\tau'(1)=\vect{d}^{\lambda}$ which implies $d^{\lambda}=1$.
Note: a different $r$ for an $APS$-monomial can be compensated by an appropriate choice of $d$.\\
\noindent{\bf(3)} Now suppose that we are given an $o$-function $L$ for $\ev$. Take $r\in\NN$ for~\eqref{Equ:SumProdHom} such that~\eqref{Equ:ChooseLowerBound} holds. Then for all $k\geq r+\max(-j,0)$ we have that~\eqref{Equ:ShiftTEvSum} and~\eqref{Equ:ShiftTEvProd}, respectively. 
Define $\fct{L'}{\AA\lr{t}}{\NN}$ by
$$L'(f)=\begin{cases}
L(f)&\text{if }f\in\AA,\\
\max(r,L(f_a),\dots,L(f_m))&\text{if
$f=\sum_{i=a}^mf_it^i\notin\AA\lr{t}\setminus\AA$}.
\end{cases}$$
Then one can check that $L'$ is an o-function for $\tau'$ which extends $L$. Namely, let $f=\sum_{i=a}^mf_it^i$ and $g=\sum_{i=b}^ng_it^i\in\AA\lr{t}$. Then for all 
$$k\geq\max(r,L(f_a),\dots,L(f_m),L(g_b),\dots,L(g_n))=\max(L'(f),L'(g))$$ we have that~\eqref{Equ:SumHom} and~\eqref{Equ:ProdHom}. For $f\in\PPowers{\AA\lr{t}}{\AA}{G}$ we get $f=t^m\,h$ with $m\in\ZZ$ and $h\in G$ ($m=0$ if $t$ is an $AS$-monomial). We define $\fct{Z'}{\PPowers{\AA\lr{t}}{\AA}{G}}{\NN}$ by $Z'(f)=Z(h)$. If $t$ is an $RS$-monomial, $Z'$ is clearly a $z$-function of $\PPowers{\AA\lr{t}}{\AA}{G}=G$. Otherwise, if $t$ is a $P$-monomial, 
$\ev(t,k)\neq0$ for all $k\geq0$ and $\ev(h,k)\neq0$ for all $k\geq Z(h)$. Thus $\ev(f,k)=\ev(t^m\,h,k)\neq0$ for all $k\geq Z(h)=Z'(f)$. Hence $Z'$ is a $z$-function of $\PPowers{\AA\lr{t}}{\AA}{G}$ extending $Z$.\\
{\bf(4)}~In particular, if $L$ and $Z$ are computable, then also $L'$ and $Z'$ are computable. Moreover, if $\ev$ is computable, $\ev'$ with~\eqref{Equ:DefineEvExt} is computable.
\end{proof}

\noindent The iterative application of the (algorithmic) construction given in Lemma~\ref{Lemma:LiftEvToPoly} yields

\begin{proposition}\label{Prop:HomomorphsimExistence}
Let $\dfield{\AA}{\sigma}$ be a difference ring with constant field $\KK$, and let $G$ be a subgroup of $\AA^*$.
Let $\dfield{\EE}{\sigma}$ be a $G$-basic $APS$-extension of $\dfield{\AA}{\sigma}$, and let $\fct{\tau}{\AA}{\seqK}$ be a (shift-bounded) $\KK$-homomorphism. Then:
\begin{enumerate}
\item There exists a (shift-bounded)
$\KK$-homomorphism $\fct{\tau'}{\EE}{\seqK}$ with $\tau'|_{\AA}=\tau$. 
\item Let $\ev$ be a defining function of $\tau$ with a $z$-function $Z$ of $G$, and suppose that there is an $o$-function $L$ for $\ev$. Then there is a defining function $\ev'$ for $\tau'$ with an $o$-function $L'$ and a $z$-function $Z'$ of $\PPowers{\EE}{\AA}{G}$ where
$\ev'|_{\AA\times\NN}=\ev$, $L'|_{\AA}=L$, $Z'|_{G}=Z$.
\item If $\ev$, $L$ and $Z$ are computable, such functions $\ev'$, $L'$ and $Z'$ are computable. 
\end{enumerate}
\end{proposition}
\begin{proof}
We apply part (1) of Lemma~\ref{Lemma:LiftEvToPoly} iteratively (for $\EE=\AA\lr{t_1}\dots\lr{t_e}$ we replace $G$ by $\PPowers{\AA\lr{t_1}\dots\lr{t_r}}{\AA}{G}$ for the $r$th iteration step). This shows that there is a difference ring homomorphism from $\fct{\tau'}{\EE}{\seqK}$ with $\tau'|_{\AA}=\tau$. Now let $\ev$ be a defining function of $\tau$ with a $z$-function $Z$ for $G$ (which exists by property~\eqref{Ev:Zero}). In addition, suppose that $\tau$ is shift-bounded. Then there is an $o$-function $L$ for $\ev$. 
Thus applying parts~(1) and~(3) of Lemma~\ref{Lemma:LiftEvToPoly} iteratively shows that there is a defining function $\ev'$ of $\tau'$ with $\ev'_{\AA\times\NN}=\ev$ together with an $o$-function $L'$ for $\ev'$ with $L'|_{\AA}=L$ and a $z$-function $\fct{Z'}{\PPowers{\EE}{\AA}{G}}{\NN}$ for $\ev'$ with $Z'|_{G}=Z$. By construction $\tau'$ is shift-bounded. This establishes parts (1) and~(2). In addition, if $\ev,$ $L$ and $Z$ are computable, one obtains also computable functions $L'$ and $Z'$ by part~(4) of Lemma~\ref{Lemma:LiftEvToPoly}. This completes the proof.
\end{proof}

\begin{example}\label{Exp:ConstructEmbedding}
Take the \pisiSE-field $\dfield{\QQ(n)(x)}{\sigma}$ over $\QQ(n)$ with $\sigma(x)=x+1$. With Example~\ref{Exp:Embedding1} we get the $\QQ(n)$-embedding $\fct{\tau}{\QQ(n)(x)}{\seqP{\QQ(n)}}$. Now consider the \rpisiSE-extension $\dfield{\AA[s_1]}{\sigma}$ of $\dfield{\QQ(n)(x)}{\sigma}$ with $\AA=\QQ(n)(x)[y]\ltr{p_1}\ltr{p_2}$ and~\eqref{Equ:s1Exp} 
from Example~\ref{Exp:FindIsomorphism} (or Example~\ref{Exp:MainDRNaiveDef} with $s_1=t_1$).
Starting with $\AA=\QQ(n)(x)$ and $G=\QQ(n)(x)^*$ we apply Lemma~\ref{Lemma:LiftEvToPoly} iteratively and extend $\tau$ successively by 
\begin{equation}\label{Equ:s1Embed}
\begin{aligned}
\tau(y)&=\langle \tprod_{i=1}^k(-1)\rangle_{k\geq0}=\langle (-1)^k\rangle_{k\geq0},&
\tau(p_1)&=\langle\tprod_{i=1}^k2\rangle_{k\geq0}=\langle 2^k\rangle_{k\geq0},\\
\tau(p_2)&=\langle\tprod_{i=1}^k\frac{n+1-i}{i}\rangle_{k\geq0}=\langle\tbinom{n}{k}\rangle_{k\geq0},&
\tau(s_1)&=\langle\tsum_{i=1}^k\frac{(-1)^i}{i}\rangle_{k\geq0}.
\end{aligned}
\end{equation}
\end{example}

\begin{example}\label{Exp:NestedP3}
Take the basic \piE-extension $\dfield{\KK(x)\ltr{p_1}\ltr{p_2}\ltr{p_3}}{\sigma}$ of $\dfield{\KK(x)}{\sigma}$ from Example~\ref{Exp:NestedP2}. 
With Example~\ref{Exp:Embedding1} we get the $\KK$-embedding $\fct{\tau}{\KK(x)}{\seqP{\KK}}$. 
Now we apply Lemma~\ref{Lemma:LiftEvToPoly} iteratively (starting with $\AA=\KK(x)$ and $G=\KK(x)^*$) and extend $\tau$ to the $\KK$-homomorphism $\fct{\tau}{\KK(x)\ltr{p_1}\ltr{p_2}\ltr{p_3}}{\seqK}$ with
\begin{align}\label{Equ:NestedPiEmb}
\tau(p_1)&=\langle k!\rangle_{k\geq0},&\tau(p_2)&=\langle \tprod_{i=1}^k i!\rangle_{k\geq0},&\tau(p_3)&=\langle \tprod_{i=1}^k \tprod_{j=1}^i j!\rangle_{k\geq0}. 
\end{align}
\end{example}

\noindent Finally, we concentrate on the question when the obtained $\KK$-homomorphism is injective. 

\begin{lemma}\label{Lemma:SimpleImpliesEmbedding}
Let $\dfield{\AA}{\sigma}$ be a difference ring with constant field $\KK$ and let $\fct{\tau}{\AA}{\seqK}$ be a $\KK$-difference ring homomorphism. If $\dfield{\AA}{\sigma}$ is simple, $\tau$ is a $\KK$-embedding. 
\end{lemma}
\begin{proof}
Take the ideal $I:=\ker(\tau)=\{f\in\AA|\tau(f)=\vect{0}\}$. Let $f\in I$. Then 
$\tau(\sigma(f))=S(\tau(f))=S(\vect{0})=\vect{0}$
and hence $\sigma(f)\in I$. This proves that $I$ is a difference ideal. Since $\tau(1)=\vect{1}$, $1\notin I$ and hence $I\neq\EE$. Since $\dfield{\EE}{\sigma}$ is simple, $I=\{0\}$, i.e., $\tau$ is injective.
\end{proof}

\noindent In a nutshell, we end up at the following central result by exploiting Theorem~\ref{Thm:RPiSiIsSimple} from Section~\ref{Sec:SimpleRings}; for an integral domain version see~\cite{Schneider:10c}.

\begin{theorem}\label{Thm:RPiSiImpliedEmbedding}
Let $\dfield{\EE}{\sigma}$ be a basic \rpisiSE-extension of a difference field $\dfield{\FF}{\sigma}$ with a (shift-bounded) $\KK$-embedding $\fct{\tau}{\FF}{\seqK}$. Then: 
\begin{enumerate}
\item There exists a (shift-bounded) $\KK$ embedding $\fct{\tau'}{\EE}{\seqK}$ with $\tau'|_{\FF}=\tau$.
\item If there is a computable defining function for $\tau$ equipped with a computable $o$-function and a computable $z$-function of $\FF^*$, such a $\KK$-embedding $\tau'$ can be given explicitly (i.e., one can construct a computable defining function for $\tau'$ equipped with a computable $o$-function and a computable $z$-function of $\PPowers{\EE}{\FF}{(\FF^*)}$). 
\end{enumerate}
\end{theorem}
\begin{proof}
By part~(1) of Proposition~\ref{Prop:HomomorphsimExistence} (with $\AA=\FF$ and $G=\FF^*$) there exists a (shift-bounded) $\KK$-homomorphism $\fct{\tau'}{\EE}{\seqK}$, and it can be constructed explicitly if a defining function of $\tau$ is computable and is equipped with a computable $o$-function and $z$-function. Since $\dfield{\EE}{\sigma}$ is a basic \rpisiSE-extension of $\dfield{\FF}{\sigma}$, $\dfield{\EE}{\sigma}$ is simple by Theorem~\ref{Thm:RPiSiIsSimple}. Hence $\tau'$ is a $\KK$-embedding by Lemma~\ref{Lemma:SimpleImpliesEmbedding}.
\end{proof}

\begin{example}
(1) Take the difference ring $\dfield{\EE}{\sigma}$ with $\EE=\QQ(n)(x)[y]\ltr{p_1}\ltr{p_2}[s_1]$ from Example~\ref{Exp:ConstructEmbedding} and take the shift-bounded $\QQ(n)$-homomorphism $\fct{\tau}{\EE}{\seqP{\QQ(n)}}$. Since $\dfield{\EE}{\sigma}$ is an \rpisiSE-extension of $\dfield{\QQ(n)(x)}{\sigma}$, $\tau$ is a $\QQ(n)$-embedding.\\
(2) Similarly, it follows that the $\KK$-homomorphism $\fct{\tau}{\KK(x)\ltr{p_1}\ltr{p_2}\ltr{p_3}}{\seqK}$ from Example~\ref{Exp:NestedP3} is a $\KK$-embedding.
\end{example}

\begin{remark} An existence statement related to Theorem~\ref{Thm:RPiSiImpliedEmbedding} has been proved in~\cite[Prop~4.1]{Singer:99}: there exists an embedding of a Picard-Vessiot extension over $\dfield{\FF}{\sigma}$ with $\KK=\const{\FF}{\sigma}$ into $\dfield{\seqK}{\Shift}$ if the following rather strong conditions hold: $\KK$ is algebraically closed and the algebraic closure of $\FF$ can be embedded into $\dfield{\seqK}{\Shift}$.\\
We remark further that Theorem~\ref{Thm:RPiSiImpliedEmbedding} has been exploited in~\cite{AS:15} to show that the class of harmonic sums~\cite{Bluemlein:99,Vermaseren:99} and cyclotomic harmonic sums~\cite{ABS:11} can be embedded into the ring of sequences; an interesting consequence is that the corresponding \rpisiSE-extension is obtained by simply using the underlying quasi-shuffle algebra~\cite{Bluemlein:04,ABS:11} of the nested sums. For further applications of Theorem~\ref{Thm:RPiSiImpliedEmbedding} within symbolic summation we refer to  Section~\ref{Sec:Application}. 
\end{remark}

\subsection{Further characterizations of \rpisiSE-extensions: the interlacing property}\label{Sec:CharacteriationEmbedding}

We will enhance the characterizations given in Theorem~\ref{Thm:EquivSimpleConst} with further equivalent statements in terms of $\KK$-embeddings. In Theorem~\ref{Thm:EquivSimpleConst} we assumed that the underlying difference field is constant-stable. First, we will show that this property is implied by the assumption that the underlying difference field is embedded into the ring of sequences.

\begin{lemma}\label{Lemma:EmbeddingImpliesConstantStable}
Let $\dfield{\AA}{\sigma}$ be a difference ring with constant field $\KK$. If $\AA$ is an integral domain and there is a $\KK$-embedding $\fct{\tau}{\AA}{\seqK}$, then $\dfield{\AA}{\sigma}$ is constant-stable.
\end{lemma}
\begin{proof}
Suppose there is a $\KK$-embedding $\tau$ and suppose that $\AA$ is an integral domain, but suppose that $\dfield{\AA}{\sigma}$ is not constant-stable.
Hence we can take an $f\in\AA\setminus\KK$ and $k>1$ with $\sigma^k(f)=f$ and $\sigma^i(f)\neq f$ for all  $1\leq i<k$. Then $\tau(f)=\tau(\sigma^k(f))=S^k(\tau(f))$. Thus $\tau(f)$ is the interlacing of $k-1$ constant-sequences and we get
$$\tau(f)=\langle c_1,c_2,\dots,c_{k-1},c_1,c_2,\dots,c_{k-1},\dots\rangle$$
for some $c_1,\dots,c_{k-1}\in\KK$.
Suppose that $c_r=0$ for some $r$ with $1\leq r<k$. Define $h_l=f\,\sigma(f)\dots\sigma^l(f)$ with $0\leq l<k$. Then observe that $\tau(h_{k-1})=\prod_{i=0}^{k-1}\Shift^i\tau(f)=\vect{0}$ since $c_r$ will be multiplied to each component. Note that $h_{k-1}=0$ since $\tau$ is injective. Since $h_0=f\neq0$, there is an $s$ with $0\leq s<k-1$ and $h_s\neq0$. For the maximal $s$ we have  $0=h_{s+1}=h_s\,\sigma^{s+1}(f)$.
Since $f\neq0$, we get $\sigma^{s+1}(f)\neq0$, and thus $h_s$ and $\sigma^{s+1}(f)$ are zero-divisors; a contradiction since $\AA$ is integral. Consequently, $c_r\neq0$ for all $1\leq r<k$.
Now define
$$w_i:=c_{i}\,\sigma(f)-c_{i+1}\,f\text{ and }
w_k:=c_k\,\sigma(f)-c_1\,f_1$$
for $1\leq i\leq k-1$. Suppose that there is a $j$ with $1\leq j\leq k$ such that $w_j=0$ holds. Then there exists a $d\in\KK^*$ with
$\sigma(f)=d\,f$. With $\sigma^k(f)=f$, it follows that
$f=\sigma^k(f)=d^k\,f$. Since $\AA$ is an integral domain, $d^k=1$, i.e., $d$ is a $k$th root of unity.
Now suppose that $d^l=1$ for some $1\leq l<k$. Then $\sigma^l(f)=d^l\,f=f$, a contradiction to the minimality of $k$. Thus $d$ is a primitive $k$th root of unity and $d,d^2,\dots,d^k\in\KK$ are distinct roots of the polynomial $x^k-1\in\AA[x]$. Now observe that $\sigma(f^k)=d^k\,f^k=f^k$ where $f^k\neq0$ and thus $f^k=c$ for some $c\in\KK^*$.
Since $f\notin\KK$, we obtain $h:=\frac{f}{c}\in\AA\setminus\KK$ with $h^k=1$. Therefore also $h$ is a root of $x^k-1$. Thus $x^k-1$ has $k+1$ roots ($k$ distinct roots from $\KK$ and one extra root from $\AA\setminus\KK$), a contradiction to the assumption that $\AA$ is an integral domain.\\
Consequently, we may suppose that $w_j\neq0$ for all $1\leq j\leq k$. Observe that in $\tau(w_j)$ always the $j$th entry is zero, but $\tau(w_j)\neq\vect{0}$ (otherwise $w_j$ would be $0$ using the fact that $\tau$ is injective). Define $g_k=w_1,\dots,w_k$. Consequently $\vect{0}=\tau(w_1)\,\tau(w_2)\dots\tau(w_k)=\tau(w_1\,w_2\dots w_k)=\tau(g_k)$. Since $\tau$ is injective, $g_k=0$. By construction $g_1=w_1\neq0$. Now let $l$ with $1\leq l<k$ be maximal with $g_l\neq0$. Then $0=g_{l+1}=g_l\,w_{l+1}$ with $w_{l+1}\neq0$ and $g_l\neq0$. Therefore $g_l$ and $w_{l+1}$ are zero-divisors, again a contradiction to the assumption that $\AA$ is integral. This completes the proof.
\end{proof}

\noindent In addition, we will need the following simple observation.

\begin{lemma}\label{Lemma:EmbeddingImpliesConstants}
Let $\dfield{\AA}{\sigma}$ be a difference ring, let $\KK$ be a field with $\KK\subseteq\const{\AA}{\sigma}$, and let $\fct{\tau}{\AA}{\seqK}$ be a $\KK$-embedding. Then $\const{\AA}{\sigma}=\KK$.
\end{lemma}
\begin{proof}
Let $f\in\const{\AA}{\sigma}$. Then 
$\tau(f)=\tau(\sigma(f))=S(\tau(f))$. Thus $\tau(f)=\vect{k}$ for some $k\in\KK$. Since $\tau$ is a $\KK$-embedding, $\tau(k)=\vect{k}=\tau(f)$. Since $\tau$ is injective, $f=k\in\KK$. Therefore $\const{\AA}{\sigma}\subseteq\KK$. The assumption $\KK\subseteq\const{\AA}{\sigma}$ implies  $\KK=\const{\AA}{\sigma}$.
\end{proof}

\noindent Using these lemmas we can enrich the defining properties of \rpisiSE-extensions as follows.

\begin{theorem}[Characterization of \rpisiSE-extensions (III)]\label{Thm:EquivSimpleConstIso}
Let $\dfield{\FF}{\sigma}$ be a difference field with $\KK=\const{\FF}{\sigma}$ and with a $\KK$-embedding $\fct{\tau}{\FF}{\seqK}$.
Let $\dfield{\EE}{\sigma}$ be a basic $APS$-extension of $\dfield{\FF}{\sigma}$. Then the following statements are equivalent.
\begin{enumerate}
 \item $\dfield{\EE}{\sigma}$ is a basic \rpisiSE-extension of $\dfield{\FF}{\sigma}$ (i.e., $\const{\EE}{\sigma}=\const{\FF}{\sigma}).$
\item $\dfield{\EE}{\sigma}$ is a simple difference ring.
 \item There is a $\KK$-embedding $\fct{\tau'}{\EE}{\seqK}$ with $\tau'|_{\FF}=\tau$.
 \item There is a $\KK$-embedding $\fct{\tau'}{\EE}{\seqK}$.
\end{enumerate}
\end{theorem}
\begin{proof}
Since $\dfield{\FF}{\sigma}$ is constant-stable by Lemma~\ref{Lemma:EmbeddingImpliesConstantStable},
the equivalence $(1)\Leftrightarrow(2)$ follows by Theorem~\ref{Thm:EquivSimpleConst}. $(2)\Rightarrow(3)$ follows by Theorem~\ref{Thm:RPiSiImpliedEmbedding}.
$(3)\Rightarrow(4)$ is obvious. Suppose that statement (4) holds.
Since $\KK\subseteq\const{\EE}{\sigma}$ and $\tau'$ is a $\KK$-embedding, it follows that $\const{\EE}{\sigma}=\KK$ by Lemma~\ref{Lemma:EmbeddingImpliesConstants}. Hence statement $(1)$ holds.  
\end{proof}

\noindent Finally, we will refine these characterizations further by linking them to the idempotent representation given in Section~\ref{Sec:Copies}. First, we work out how the idempotent elements are evaluated within a $\KK$-embedding.

\begin{lemma}\label{Lemma:EvEi}
Let $\dfield{\FF[y]}{\sigma}$ be an \rE-extension of $\dfield{\FF}{\sigma}$ of order $n$ with $\sigma(y)=\alpha\,y$. Let $\fct{\tau}{\FF}{\seqK}$ be a $\KK$-embedding with $\KK=\const{\FF}{\sigma}$ and let $\fct{\tau'}{\FF[y]}{\seqK}$ be a $\KK$-embedding with $\tau'|_{\FF}=\tau$. Then there is a defining function $\ev$ of $\tau'$ such that for all $k\in\NN$ we have  $\ev(y,k)=\alpha^u\,\alpha^k$ for some $u\in\NN$ with $0\leq u<n$. In particular, for the $e_s$ with $0\leq s<n$ as defined in~\eqref{Equ:SpecificEi} we get
\begin{equation}\label{Equ:EvEi}
\ev(e_{s},k)=
\begin{cases}
0&\text{if }n\nmid k+u+s+1\\
1&\text{if }n\mid k+u+s+1.
\end{cases}
\end{equation}
\end{lemma}
\begin{proof}
By Lemma~\ref{Lemma:LiftEvToPoly} (parts (1) and (2)) it follows that $\tau(y)=\langle c\,y^k\rangle_{k\geq0}$ for some $c\in\KK^*$ with $c^n=1$. Thus $c=\alpha^u$ for some $u\in\{0,\dots,n-1\}$. Hence we can choose a defining function for $\tau'$ with $\ev(y,k)=\alpha^u\,\alpha^k$. With~\eqref{Equ:BasicEiProp} we conclude that~\eqref{Equ:EvEi} holds.
\end{proof}

\noindent Now we are ready to show that the representation in terms of our idempotent elements encodes the interlacing of certain sequences. In this context, it will be convenient to introduce for $0\leq s<n$ the function
$\fct{\lambda_s}{\seqK}{\seqK}$ defined by 
$$\lambda_s(\langle a_k\rangle_{k\geq0})=\langle a_{n\,r+(n-1-s)}\rangle_{r\geq0}.$$ 
In short, $\lambda_s$ picks out the $n$th entries with offset $n-1-s$. It is easily checked that $\lambda_s$ forms a difference ring homomorphism from $\dfield{\seqK}{\Shift^n}$ to $\dfield{\seqK}{\Shift}$.

\begin{lemma}\label{Lemma:Homoidempotent}
Let $\dfield{\FF}{\sigma}$ be a difference field with $\KK=\const{\FF}{\sigma}$,
and let $\dfield{\EE}{\sigma}$ be a basic \rpisiSE-extension of $\dfield{\FF}{\sigma}$ given in the form~\eqref{Equ:SingleRPS} where the \rE-monomial $y$ has order $n$ with $\alpha=\frac{\sigma(y)}{y}\in\KK$. Let $e_0,\dots,e_{n-1}$ be the idempotent elements defined in~\eqref{Equ:SpecificEi}.
\begin{enumerate}
 \item (Subsequences) Let $\fct{\tau}{\EE}{\seqK}$ be a $\KK$-embedding with a defining function $\ev$ where $\ev(y,k)=\alpha^{u+k}$ for some $0\leq u<n$. Then for $0\leq s<n$ and $l:=s+u\mod n$ we get the $\KK$-embedding $\tau_s=\lambda_l\circ\tau|_{e_s\,\EE}$ from $\dfield{e_s\,\EE}{\sigma^n}$ to $\dfield{\seqK}{\Shift}$ where 
 \begin{equation}\label{Equ:SequCut}
 \tau_s(e_s\,f)=\langle\ev(e_s\,f,r\,n+n-1-l)\rangle_{r\geq0}.
 \end{equation}
 \item (Interlacing) Let $\fct{\tau_s}{e_s\,\EE}{\seqK}$ be a $\KK$-embedding from $\dfield{e_s\,\EE}{\sigma^n}$ to $\dfield{\seqK}{\Shift}$ with a defining function $\ev_s$ for some $0\leq s<n$. Then for $0\leq j< n$ the maps $\fct{\tau_{j}}{e_j\,\EE}{\seqK}$ defined by $\tau_j(e_j\,f)=\tau_s(e_s\,\sigma^{s-j}(f))$ are $\KK$-embeddings from $\dfield{e_j\,\EE}{\sigma^n}$ to $\dfield{\seqK}{\Shift}$ with the defining functions $\ev_j(e_j\,f,k)=\ev_s(e_s\,\sigma^{s-j}(f),k)$ for all $f\in\EE$. Moreover, the interlacing of the sequences 
 $\tau_{n-1}(e_{n-1}\,f),\dots,\tau_{0}(e_{0}\,f),$
 i.e., the map $\fct{\tau}{\EE}{\seqK}$ defined by\footnote{At a first glance this construction seems odd. However note that the shift in the representation~\eqref{Equ:DiagramCopies} is left to right. But the shift in $\dfield{\seqK}{\Shift}$ means that a sequence is moved from right to left dropping the first term. As a consequence also the interlacing of the sequences must be given in the order $\ev_{n-1},\ev_{n-2},\dots,\ev_{0}$ in order to construct a difference ring homomorphism.}
\begin{equation}\label{Equ:TauInterlacing}
 \tau(f)=\big\langle\begin{array}[t]{ccccc}
 \ev_{n-1}(e_{n-1}\,f,0),&\ev_{n-2}(e_{n-2}\,f,0),&\dots,&\ev_{0}(e_{0}\,f,0),&\\
 \ev_{n-1}(e_{n-1}\,f,1),&\ev_{n-2}(e_{n-2}\,f,1),&\dots,&\ev_{0}(e_{0}\,f,1),&\\
 \vdots&\vdots&\dots&\vdots&\big\rangle
 \end{array}
 \end{equation}
 is a $\KK$-embedding from $\dfield{\EE}{\sigma}$ into $\dfield{\seqK}{\Shift}$; here we have that $\tau(y)=\langle \alpha^k\rangle_{k\geq0}$.
\end{enumerate}
\end{lemma}
\begin{proof}
(1) Note that $\tau$ is a $\KK$-embedding of $\dfield{\EE}{\sigma^n}$ into $\dfield{\seqK}{\Shift^n}$ since $\tau(\sigma^n(f))=\Shift^n(\tau(f))$ for all $f\in\EE$. In particular, $\tau|_{e_s\,\EE}$ is a $\KK$-embedding of $\dfield{e_s\,\EE}{\sigma^n}$ into $\dfield{\seqK}{\Shift^n}$. By Lemma~\ref{Lemma:EvEi} it follows that
$\ev(e_s\,f,k+n\,r)$ with $r\in\NN$ is zero if $k\neq n-1-l$. This implies that $\tau_s=\lambda_l\circ\tau|_{e_s\,\EE}$ is a $\KK$-embedding. The evaluation in~\eqref{Equ:SequCut} is obvious.\\
(2) Take the $\KK$-embedding $\fct{\tau_s}{e_s\,\EE}{\seqK}$. By Theorem~\ref{Thm:Interlacing} it follows that the $\dfield{e_j\,\EE}{\sigma^n}$ are isomorphic to $\dfield{e_s\,\EE}{\sigma^n}$ with $\sigma^{s-j}$ for $0\leq j<n$. Thus we get the $\KK$-embeddings $\tau_j$ as claimed in the statement. In particular, a defining function for $\tau_j$ with $0\leq j<n$ is given by $\ev_j(e_j\,f,k)=\ev_s(e_s\,\sigma^{s-j}(f),k)$ for all $f\in\EE$. Now take the map $\tau$ as defined in~\eqref{Equ:TauInterlacing}. It is easily seen that $\tau$ is a ring isomorphism. 
Moreover, for any $f\in\EE$, for $k\in\NN$ big enough and $0\leq r<n$ it follows that
$\ev_{n-1-r}(e_{n-1-r}\sigma(f),k)=\ev_s(e_s\sigma^{s-n+1+r}(\sigma(f)),k)=\ev_s(e_s\sigma^{s-n+1+(r+1)}(f),k)=\ev_{n-1-(r+1)}(e_{n-1-(r+1)}f,k)$. With~\eqref{Equ:TauInterlacing} we conclude that 
$\tau(\sigma(f))=\Shift(\tau(f))$ and consequently $\tau$ is a difference ring isomorphism. Finally, for any $c\in\KK$ and $0\leq s<n$ we have $\tau_s(c)=\vect{c}$ and therefore $\tau(c)=\vect{c}$. Consequently, $\tau$ is a $\KK$-embedding.
Consider $\tau'=\tau|_{\FF[y]}$ and let $\ev'$ be a defining function. By the construction~\eqref{Equ:TauInterlacing} we have that $\ev'(e_{n-1},n\,k)=\ev_{n-1}(e_{n-1},k)$ for all $k\in\NN$. Since $\tau_{n-1}$ is a $\KK$-embedding and $e_{n-1}$ is the multiplicative identity, $\tau_{n-1}(e_{n-1})=\vect{1}$. Thus there is a $\delta\in\NN$ with $\ev'(e_{n-1},n\,k)=1$ for all $k\geq\delta$. By 
Lemma~\ref{Lemma:EvEi} it follows that $\ev'(y,k)=\alpha^{k+u}$ for some $0\leq u<n$ and $k$ big enough and that 
$\ev'(e_{n-1},n\,k)=1$ if $n\mid n\,k+u+n$. This implies that $u=0$. Thus $\tau(y)=\langle\ev'(y,k)\rangle_{k\geq0}=\langle \alpha^k\rangle_{k\geq0}$.
\end{proof}

\begin{remark}\label{Remark:Interlacing1}
We obtain an alternative proof of part (5) of Theorem~\ref{Thm:Interlacing} if one assumes that there is a $\KK$-embedding $\fct{\tau}{\FF}{\seqK}$ (instead of the weaker assumption that $\dfield{\FF}{\sigma}$ is constant-stable). Namely, together with the assumption that $\dfield{\EE}{\sigma}$ is an \rpisiSE-extension of $\dfield{\FF}{\sigma}$, we can conclude that there is a $\KK$-embedding $\fct{\tau'}{\EE}{\seqK}$ with $\tau'|_{\FF}=\tau$ by using Theorem~\ref{Thm:EquivSimpleConstIso}.
Furthermore, by part~(1) of Lemma~\ref{Lemma:Homoidempotent} we get the $\KK$-embeddings $\fct{\tau_s}{e_s\,\EE}{\seqK}$ for $0\leq s<n$. By part~(3) of Theorem~\ref{Thm:Interlacing} it follows that the $\dfield{e_s\,\EE}{\sigma^n}$ are $PS$-extensions of $\dfield{e_s\,\FF}{\sigma^n}$ for $0\leq s<n$. 
Thus they are \pisiSE-extensions by the equivalence (1)$\Leftrightarrow$(4) of Theorem~\ref{Thm:EquivSimpleConstIso}.
\end{remark}

\noindent We end up at the following final characterizations of \rpisiSE-extensions.

\begin{theorem}[Characterization of \rpisiSE-extensions (IV)]\label{Thm:EquivSimpleConstIso2}
Let $\dfield{\FF}{\sigma}$ be a difference field with $\KK=\const{\FF}{\sigma}$ equipped with a $\KK$-embedding $\fct{\tau}{\FF}{\seqK}$.
Let $\dfield{\EE}{\sigma}$ be a basic $RPS$-extension of $\dfield{\FF}{\sigma}$ given in the form~\eqref{Equ:SingleRPS} where the \rE-monomial $y$ has order $n$ with $\alpha=\frac{\sigma(y)}{y}\in\KK$. Let $e_0,\dots,e_{n-1}$ be the idempotent elements defined in~\eqref{Equ:SpecificEi}. Then the following statements are equivalent.
\begin{enumerate}
 \item[(1)]-- (6) from Theorem~\ref{Thm:EquivSimpleCopy}. 
 \item[(7)] There is a $\KK$-embedding $\fct{\tau'}{\EE}{\seqK}$ with $\tau'|_{\FF}=\tau$.
 \item[(8)] There is a $\KK$-embedding $\fct{\tau'}{\EE}{\seqK}$.
 \item[(9)] For all $s\in\{0,\dots,n-1\}$ there is a $\KK$-embedding $\tau_s$ from $\dfield{e_s\EE}{\sigma^n}$ into $\dfield{\seqK}{\Shift}$ with $\tau_s(e_s\,f)=\lambda_s(\tau(f))$ for all $f\in\FF$.
 \item[(10)] There is an $s\in\{0,\dots,n-1\}$ such that there is a $\KK$-embedding $\tau_s$ from $\dfield{e_s\EE}{\sigma^n}$ into $\dfield{\seqK}{\Shift}$ with $\tau_s(e_s\,f)=\lambda_s(\tau(f))$ for all $f\in\FF$.
 \item[(11)] There are $\KK$-embeddings of $\dfield{e_s\,\EE}{\sigma^n}$ into $\dfield{\seqK}{\Shift}$ for $0\leq s<n$ such that
 $\fct{\tau'}{\EE}{\seqK}$ with~\eqref{Equ:TauInterlacing} is a $\KK$-embedding with $\tau'|_{\FF}=\tau$.
 \item[(12)] There is an $s\in\{0,\dots,n-1\}$ such that there is a $\KK$-embedding $\tau_s$ from $\dfield{e_s\EE}{\sigma^n}$ into $\dfield{\seqK}{\Shift}$.
 \end{enumerate}
\end{theorem}

\begin{proof}
Since $\dfield{\FF}{\sigma}$ is constant-stable  by Lemma~\ref{Lemma:EmbeddingImpliesConstantStable}, the equivalences (1)--(8) follow by Theorems~\ref{Thm:EquivSimpleCopy} and~\ref{Thm:EquivSimpleConstIso}.\\ $(7)\Rightarrow(9)$: Suppose that (7) holds and let $s\in\NN$ with $0\leq s<n$. Then by part~(1) of  Lemma~\ref{Lemma:Homoidempotent} we get the $\KK$-embedding $\tau_s:=\lambda_s\circ\tau'|_{e_s\,\EE}$ from $\dfield{e_s\EE}{\sigma^n}$ into $\dfield{\seqK}{\Shift}$. Since $\tau'|_{\FF}=\tau|_{\FF}$, we have $\lambda_s(\tau'(e_s\,f))=\lambda_s(\tau(f))$ for all $f\in\FF$ which shows part (9).\\ 
$(9)\Rightarrow(10)$ is immediate. $(10)\Rightarrow(11)$: Suppose that (10) holds and let $\ev_s$ and $\ev$ be defining functions for $\tau_s$ and $\tau$, respectively. 
Then by part~(2) of Lemma~\ref{Lemma:Homoidempotent} 
we get the $\KK$-embeddings $\tau_j$ from $\dfield{e_j\,\EE}{\sigma^n}$ to $\dfield{\seqK}{\Shift}$ for $0\leq j<n$ with $\tau_j(e_j\,f)=\tau_s(e_s\,\sigma^{s-j}(f))$ for all $f\in\EE$; for a defining function we take $\ev_j(e_j\,f):=\ev_s(e_s\sigma^{j-s}(f),k)$ for $f\in\EE$. By assumption we have that $\ev_s(e_s\,h,\tilde{k})=\ev(h,n\,\tilde{k}+n-1-s)$ for all $h\in\FF$ and for all $\tilde{k}\in\NN$ and $0\leq s<n$.
Further, we get the $\KK$-embedding $\fct{\tau'}{\EE}{\seqK}$ with~\eqref{Equ:TauInterlacing}. Let $\ev'$ be a defining function for $\tau'$.  Now let $k\in\NN$ be chosen big enough and take $\tilde{k}\in\NN$ and $j$ with $0\leq j<n$ and $k=\tilde{k}\,n+j$. Then for 
$f\in\FF$ we get
\begin{align*}
\ev'(f,k)&=\ev_{n-1-j}(e_{n-1-j}f,\tilde{k})=\ev_s(e_s\sigma^{s-n+1+j}(f),\tilde{k})\\
&=\ev(\sigma^{s-n+1+j}(f),n\,\tilde{k}+n-1-s)=\ev(f,n\,\tilde{k}+j)=\ev(f,k) 
\end{align*}
which proves that $\tau'(f)=\tau(f)$ for all $f\in\FF$. Thus (11) is proven. $(11)\Rightarrow(12)$ is trivial.\\
$(12)\Rightarrow(5)$: By Theorem~\ref{Thm:Interlacing} $\dfield{e_s\EE}{\sigma^n}$ is a $PS$-extension of $\dfield{e_s\FF}{\sigma^n}$. Since $\dfield{\FF}{\sigma}$ is constant-stable, also $\dfield{e_s\,\FF}{\sigma}$ is constant-stable.
Therefore by $(4)\Rightarrow(1)$ of Theorem~\ref{Thm:EquivSimpleConstIso} it follows that $\dfield{e_s\EE}{\sigma^n}$ is a \pisiSE-extension of $\dfield{e_s\FF}{\sigma^n}$. This completes the proof.
\end{proof}

\begin{remark}\label{Remark:Interlacing2}
We emphasize property (11) of Theorem~\ref{Thm:EquivSimpleConstIso2}.  
Consider the representation~\eqref{Equ:DiagramCopies} of an \rpisiSE-extension $\dfield{\EE}{\sigma}$ of $\dfield{\FF}{\sigma}$ together with a $\KK$-embedding $\fct{\tau}{\FF}{\seqK}$. Then one can construct a $\KK$-embedding $\fct{\tau'}{\EE}{\seqK}$ by Theorem~\ref{Thm:EquivSimpleConstIso2}; here we choose $\tau(y)=\langle\alpha^k\rangle_{k\geq0}$. In particular, by part~(1) of Lemma~\ref{Lemma:Homoidempotent}
the $\tau_s:=\lambda_s\circ\tau|_{e_s\,\EE}$ for $0\leq s<n$ are $\KK$-embeddings of the \pisiSE-extensions $\dfield{e_s\,\EE}{\sigma^n}$ of $\dfield{e_s\,\FF}{\sigma^n}$. Thus $\dfield{\EE}{\sigma}$ is isomorphic to $\dfield{\tau(\EE)}{\Shift}$ and by part 2 of  Lemma~\ref{Lemma:Homoidempotent} this ring is nothing else than the interlacing of the sequences $\tau_{n-1}(e_{n-1}\,\EE),\dots,\tau_{0}(e_{0}\,\EE)$. 
And since the $\dfield{e_s\,\EE}{\sigma^n}$ are isomorphic to $\dfield{\tau_s(e_s\,\EE)}{\Shift^n}$, the representation~\eqref{Equ:DiagramCopies}, in particular the statements in Theorem~\ref{Thm:Interlacing}, formulate precisely the interlacing property in the \rpisiSE-language.
\end{remark}

\section{Application I: an alternative algorithm for parameterized telescoping}\label{Sec:Telescoping}

As motivated in Subsection~\ref{Sec:BasicExt}, in particular
in Remark~\ref{Remark:ConstructiveVersion}, one is interested in the telescoping problem: given a difference ring $\dfield{\AA}{\sigma}$ with constant field $\KK$ and given an $f\in\AA$, decide constructively if there exists a $g\in\AA$ with 
\begin{equation}\label{Equ:TeleNewAlg}
\sigma(g)-g=f.
\end{equation}
More generally, the parameterized telescoping problem in Subsec.~\ref{Sec:PLDE} below will
play a prominent role: given $\vect{f}=(f_1,\dots,f_d)\in\AA^d$, find all $c_1,\dots,c_d\in\KK$ and $g\in\AA$ with
\begin{equation}\label{Equ:PT}
\sigma(g)-g=c_1\,f_1+\dots+c_d\,f_d.
\end{equation}
Note that the set (compare~\cite{Karr:81})
$$V(\vect{f},\AA):=\{(c_1,\dots,c_d,g)\in\KK^d\times\AA|\text{ ~\eqref{Equ:PT}~holds}\}$$
is a subspace of $\KK^d\times\AA$ over $\KK$ and its dimension is at most $d+1$; see~\cite[Lemma~2.17]{Schneider:16a}. Thus finding all such solutions can be summarized as follows.

\medskip

\noindent\fbox{\begin{minipage}{12.9cm}
\noindent \textbf{Problem PTDR for $\dfield{\AA}{\sigma}$:} Parameterized Telescoping for a Difference Ring.\\
\textit{Given} a difference ring $\dfield{\AA}{\sigma}$ with constant field $\KK$ and $\vect{f}=(f_1,\dots,f_d)\in\AA^d$.\\
\textit{Find} a basis of the $\KK$-vector space $V(\vect{f},\AA)$.
\end{minipage}}

\medskip

\noindent In~\cite{Schneider:16a} we derived an efficient algorithmic framework that solves Problem~PTDR for a basic \rpisiSE-extension $\dfield{\EE}{\sigma}$ of a difference field $\dfield{\FF}{\sigma}$ under the assumption that the following two problems can be solved in $\dfield{\FF}{\sigma}$. First, we require that there is an algorithm that finds all solutions of a given first-order parameterized linear difference equation.

\medskip

\noindent\fbox{\begin{minipage}{12.9cm}
\noindent \textbf{Problem PFLDE for $\dfield{\FF}{\sigma}$:} Parameterized First-Order Linear Difference Equ.\\
\textit{Given} a difference field $\dfield{\FF}{\sigma}$ with constant field $\KK$, $a\in\FF^*$ and $\vect{f}=(f_1,\dots,f_d)\in\FF^d$.\\
\textit{Find} a basis\footnote{Note that $V$ is $\KK$-subspace of $\KK^d\times\FF$ whose dimension is at most $d+1$; see~\cite[Lemma~2.17]{Schneider:16a}.} of $V:=\{(c_1,\dots,c_d,g)\in\KK^d\times\FF\mid \sigma(g)+a\,g=c_1f_1+\dots+c_d\,f_d\}$.
\end{minipage}}

\medskip

\noindent Second, we need an algorithm for the following problem that can be considered as the multiplicative version of Problem~PTDR; see also~\cite{Karr:81}.

\medskip

\noindent\fbox{\begin{minipage}{12.9cm}
\noindent \textbf{Problem PMT for $\dfield{\FF}{\sigma}$:} Parameterized Multiplicative Telescoping.\\
\textit{Given} a difference field $\dfield{\FF}{\sigma}$ and $\vect{f}=(f_1,\dots,f_d)\in(\FF^*)^d$.\\
\textit{Find} a basis\footnote{Note that $M$ is a $\ZZ$-submodule of $\ZZ^d$ whose rank is at most $d$.} $M:=\{(z_1,\dots,z_d)\in\ZZ^d\mid \frac{\sigma(g)}{g}=f_1^{z_1}\dots f_d^{z_d}\text{ for some }g\in\FF^*\}$.
\end{minipage}}

\medskip

\noindent Then a special case of~\cite[Thm.~2.23]{Schneider:16a} leads to the following algorithmic result.

\begin{theorem}\label{Thm:RPSReduction}
Let $\dfield{\EE}{\sigma}$ be a basic \rpisiSE-extension of a difference field $\dfield{\FF}{\sigma}$. If Problems~PFLDE and~PMT are solvable for $\dfield{\FF}{\sigma}$ then Problem~PTDR is solvable for $\dfield{\EE}{\sigma}$.
\end{theorem}
\noindent For instance, one can solve Problems~PFLDE and~PMT for $\dfield{\FF}{\sigma}$ if $\dfield{\FF}{\sigma}$ is one of the difference fields of Example~\ref{Exp:PSField}. More generally, there is the following result~\cite{Schneider:06d} that one obtains by analyzing carefully Karr's summation algorithm~\cite{Karr:81}.
\begin{theorem}\label{Thm:PSReduction}
Let $\dfield{\FF}{\sigma}$ be a \pisiSE-field extension of $\dfield{\GG}{\sigma}$. If Problem~PFLDE is solvable for $\dfield{\GG}{\sigma}$ and $\dfield{\GG}{\sigma}$ is $\sigma^*$-computable (i.e., certain algorithmic properties hold that are specified in~\cite[Def.~1]{Schneider:06d}), then Problems~PFLDE and~PMT are solvable for $\dfield{\FF}{\sigma}$.
\end{theorem}
\noindent E.g., if $\dfield{\GG}{\sigma}$ is the difference field of unspecified sequences~\cite{Schneider:06d} or of radical expressions~\cite{Schneider:07f}, it is $\sigma^*$-computable and Problem~PFLDE is solvable for $\dfield{\GG}{\sigma}$. Thus one can solve Problems~PFLDE and~PMT for a tower of \pisiSE-field extensions over $\dfield{\GG}{\sigma}$; see Remark~\ref{Remark:ConstructiveVersion}.

In Theorem~\ref{Thm:PTRing} below we will derive an alternative approach that tackles Problem~PTDR for a basic \rpisiSE-extension $\dfield{\EE}{\sigma}$ of $\dfield{\FF}{\sigma}$. This new reduction will require more conditions than the ones formulated in Theorem~\ref{Thm:RPSReduction}, but it provides more flexibility: Problem~PTDR can be solved for the total ring of  fractions $Q(\EE)$ by exploiting the existing summation algorithms for \pisiSE-field extensions~\cite{Karr:81,Schneider:04a,Schneider:05f,Schneider:07d,Schneider:08c,Schneider:15}.

We assume that the \rpisiSE-extension $\dfield{\EE}{\sigma}$ of a constant-stable difference field $\dfield{\FF}{\sigma}$ is given in the form~\eqref{Equ:SingleRPS} where the \rE-monomial $y$ has order $n$ and the idempotent elements are given by~\eqref{Equ:SpecificEi} with $0\leq s<n$. Here we emphasize once more that several basic \rE-monomials can be reduced to this specific situation using Lemma~\ref{Lemma:RSingleIsoMultiple}. Then using Theorem~\ref{Thm:Interlacing} we get the direct sum~\eqref{Equ:decomposition} where the $\dfield{e_s\,\EE}{\sigma^n}$ are \pisiSE-extensions of $\dfield{\FF}{\sigma}$ with $0\leq s<n$. More precisely, taking into account the isomorphisms~\eqref{Equ:FirstEquiv} we get the \pisiSE-extensions $\dfield{\tilde{\EE}}{\sigma_s}$ of $\dfield{\FF}{\sigma^n}$ for $0\leq s<n$ with $\tilde{\EE}=\FF\lr{t_1}\dots\lr{t_e}$ and~\eqref{Equ:Sigmas}. 
In this setting we get the following lemma.

\begin{lemma}\label{Lemma:fIsConstant}
Let $\dfield{\FF}{\sigma}$ be a constant-stable difference field with $\KK=\const{\FF}{\sigma}$.
Let $\dfield{\EE}{\sigma}$ be a basic \rpisiSE-extension of $\dfield{\FF}{\sigma}$ with~\eqref{Equ:SingleRPS} and~\eqref{Equ:SpecificEi} for $0\leq s<n$. Then for any $f\in\EE$ with $f+\sigma(f)+\dots+\sigma^{n-1}(f)=0$ we have that $f\in\KK[y]$. 
\end{lemma}
\begin{proof}
Let $f=e_0\,f_0+\dots+e_{n-1}\,f_{n-1}$ as above. Then $0=\sigma(\sigma^{n-1}(f)+\dots+f)-(\sigma^{n-1}(f)+\dots+f)=\sigma^n(f)-f$. Let $0\leq s<n$ . Then by part~(3) of Lemma~\ref{Lemma:MapSigma} we get $\sigma^n(e_s\,f_s)=e_s\,f_s$. Since $\dfield{e_s\,\EE}{\sigma^n}$ is a \pisiSE-extension of $\dfield{e_s\FF}{\sigma^n}$ by part (5) of Theorem~\ref{Thm:Interlacing}, it follows that $e_s\,f_s\in\const{e_s\,\FF}{\sigma^n}$. Therefore $f_s\in\const{\FF}{\sigma^n}$ and since $\dfield{\FF}{\sigma}$ is constant-stable, $f_s\in\KK$. Thus $f_s\in\KK$ for all $0\leq s<n$, and hence $f\in\KK[y]$.  
\end{proof}

\noindent Now we are ready to present the following reduction strategy.

\begin{proposition}\label{Prop:ReduceToCopies}
Let $\dfield{\EE}{\sigma}$ be a basic \rpisiSE-extension of a constant-stable difference field $\dfield{\FF}{\sigma}$ with~\eqref{Equ:SingleRPS} and~\eqref{Equ:SpecificEi} for $0\leq s<n$ and take the \pisiSE-extensions $\dfield{\tilde{\EE}}{\sigma_s}$ of $\dfield{\FF}{\sigma^n}$ for $0\leq s<n$. Let $f\in\EE$ and define $\tilde{f}_s=\Big(\sum_{i=0}^{n-1}\sigma^i(f)\Big)|_{y\to\alpha^{n-1-s}}\in\tilde{\EE}$. 
\begin{enumerate}
\item Then there exists a $g\in\EE$  with~\eqref{Equ:TeleNewAlg} iff there are $\tilde{g}_s\in\tilde{\EE}$ for $0\leq s<n$ with
\begin{equation}\label{Equ:TeleCopy}
\sigma_s(\tilde{g}_s)-\tilde{g}_s=\tilde{f}_s.
\end{equation}
\item If we are given such $\tilde{g}_s$ of $0\leq s<n$, then we get $f':=f-(\sigma(\tilde{g})-\tilde{g})\in\KK[y]$ for $\tilde{g}=e_0\,\tilde{g}_0+\dots+e_{n-1}\,\tilde{g}_{n-1}$. Furthermore, a solution $g'\in\KK[y]$ of $\sigma(g')-g'=f'$ can be computed, and $g=\tilde{g}+g'$ is a solution of~\eqref{Equ:TeleNewAlg}. 
\end{enumerate}
\end{proposition}

\begin{proof}
``$\Rightarrow$'': Suppose that there is a $g\in\EE$ with~\eqref{Equ:TeleNewAlg} and write $g=e_0\,g_0+\dots+e_{n-1}\,g_{n-1}$ with $g_i\in\tilde{\EE}$ for $0\leq i<n$. For $\tilde{f}:=\sum_{i=0}^{n-1}\sigma^i(f)$ we have that $\tilde{f}=\sum_{i=0}^{n-1}\sigma^i(\sigma(g)-g)=\sigma^n(g)-g$. Note in addition that
$\tilde{f}=e_0\,\tilde{f}_0+\dots+e_{n-1}\,\tilde{f}_{n-1}$.
Thus by Theorem~\ref{Thm:Interlacing} it follows that $\sigma^n(e_s\,g_s)-e_s\,g_s=e_s\,\tilde{f}_s$ and therefore $\sigma_s(g_s)-g_s=\tilde{f}_s$ for all $0\leq s<n$.\\
``$\Leftarrow$'': Suppose that we get $\tilde{g}_s\in\tilde{\EE}$ for $0\leq s<e$ with~\eqref{Equ:TeleCopy}.
Now define $f':=f-(\sigma(\tilde{g})-\tilde{g})\in\EE$. Observe that
$f'+\sigma(f')+\dots+\sigma^{n-1}(f')=\tilde{f}-(\sigma^n(\tilde{g})-\tilde{g})=\sum_{s=0}^{n-1}(\tilde{f}_s-(\sigma(\tilde{g}_s)-\tilde{g}_s))=0.$
Therefore by Lemma~\ref{Lemma:fIsConstant} it follows that $f'\in\KK[y]$. 
Moreover, note that for any $0\leq i<n$ we have that
$\sigma(\gamma)-\gamma=\frac{1}{\alpha^i-1}(\alpha^i\,y^i-y^i)=y^i$ for $\gamma=\frac{y^i}{\alpha^i-1}$. Thus by linearity we can compute a $g'\in\KK[y]$ with $\sigma(g')-g'=f'$. Consequently we can take $g=\tilde{g}+g'$ and get 
$f-(\sigma(g)-g)=f-(\sigma(\tilde{g})-\tilde{g})-(\sigma(g')-g')=f'-f'=0.$
Due to this explicit construction, the second part of the proposition follows.
\end{proof}

\noindent In summary, solving the telescoping problem in $\dfield{\EE}{\sigma}$ can be reduced to solving the telescoping problem in the $n$ \pisiSE-extensions $\dfield{\tilde{\EE}}{\sigma_s}$ of $\dfield{\FF}{\sigma^n}$ with $0\leq s<n$. By Theorem\footnote{Here one does not use the full machinery for \rpisiSE-monomials, but exploits only sub-algorithms that tackle \pisiSE-monomials. More precisely, the resulting algorithm equals a special version given in~\cite{Schneider:15}.}~\ref{Thm:RPSReduction} this is possible if one can solve Problems~PFLDE and~PMT for $\dfield{\FF}{\sigma^n}$.

\begin{example}\label{Exp:ComputeTele}
We continue with Examples~\ref{Exp:Structure1}--\ref{Exp:Structure3}. Given $f\in\EE$ from Example~\ref{Exp:Structure1}, we want to calculate a $g\in\EE$ with~\eqref{Equ:TeleNewAlg}. Therefore we set $\tilde{f}=f+\sigma(f)$ and get
$\tilde{f}=\tfrac{-2 y s_1}{(x+1) (x+2)}
+\tfrac{2(
        x^2+4 x+5)}{(x+1)^2 (x+2)^2 (x+3)}
-\tfrac{y}{(x+1) (x+2)}=e_0\,\tilde{f}_0+e_1\,\tilde{f}_1$
with
$\tilde{f}_0=\tilde{f}|_{y\to-1}=\tfrac{2 s_1}{(x+1) (x+2)}
+\tfrac{x^3+8 x^2+19 x+16}{(x+1)^2 (x+2)^2 (x+3)}$ and 
$\tilde{f}_1=\tilde{f}|_{y\to1}=-\tfrac{2 s_1}{(x+1) (x+2)}
-\tfrac{x^3+4 x^2+3 x-4}{(x+1)^2 (x+2)^2 (x+3)}$. 
We follow Proposition~\ref{Prop:ReduceToCopies} 
and try to compute the $\tilde{g_s}\in\tilde{\EE}$ with $s=0,1$ such that~\eqref{Equ:TeleCopy} holds. 
Using the algorithms from~\cite{Schneider:15} we compute 
$\tilde{g}_0=s_1
+s_1^2
-s'_2
+\frac{x+2}{x+1}$ and $\tilde{g}_1=s_1
+s_1^2
-s'_2
-\frac{x}{x+1}$. Thus we set $\tilde{g}=e_0\,\tilde{g}_0+e_1\,\tilde{g}_1
=s_1
+s_1^2
-s'_2
-\frac{-1+y+x\,y}{x+1}\in\EE$ and define $f':=f-(\sigma(\tilde{g})-\tilde{g})=-2\,y\in\KK[y]$. Finally, we compute $g'=y$ with $\sigma(g')-g'=f'$ and get the solution $g=\tilde{g}+g'=s_1
+s_1^2
-s'_2
+\frac{1}{x+1}$ of~\eqref{Equ:TeleNewAlg}.
\end{example}

Now we turn to the problem to solve the telescoping problem for the total difference ring of fractions $\dfield{Q(\EE)}{\sigma}$. Here we will exploit the representation
\begin{equation}\label{Equ:QuotRep}
Q(\EE)=e_0\,Q(\tilde{\EE})\oplus\dots\oplus e_{n-1}\,Q(\tilde{\EE});
\end{equation}
compare~\cite[Sec.~1.3 ]{Singer:99} and~\cite[Cor.~6.9]{Singer:08}.
Note that the difference field
$\dfield{Q(\tilde{\EE})}{\sigma_s}$ forms a polynomial \pisiSE-field extension of $\dfield{\FF}{\sigma^n}$ by iterative application of Corollary~\ref{Cor:LiftToField}. Then using these properties, one can extract the following reduction strategy.

\begin{proposition}\label{Prop:ReduceToCopiesQuot}
Let $\dfield{\EE}{\sigma}$ be a basic \rpisiSE-extension of a constant-stable difference field $\dfield{\FF}{\sigma}$ with~\eqref{Equ:SingleRPS} and~\eqref{Equ:SpecificEi} for $0\leq s<n$ and take the \pisiSE-field extensions $\dfield{Q(\tilde{\EE})}{\sigma_s}$ of $\dfield{\FF}{\sigma^n}$ for $0\leq s<n$. Let $f\in Q(\EE)$ and set $\tilde{f}_s:=\Big(\sum_{i=0}^{n-1}\sigma^i(f)\Big)|_{y\to\alpha^{n-1-s}}\in Q(\tilde{\EE})$. 
\begin{enumerate}
\item Then there is a $g\in Q(\EE)$  with~\eqref{Equ:TeleNewAlg} iff there are $\tilde{g}_s\in Q(\tilde{\EE})$ with~\eqref{Equ:TeleCopy} for $0\leq s<n$. 
\item If we are given such $\tilde{g}_s$ for $0\leq s<n$, a solution $g\in Q(\EE)$ of~\eqref{Equ:TeleNewAlg} can be computed as given in part~(2) of Proposition~\ref{Prop:ReduceToCopies}. 
\end{enumerate}
\end{proposition}

\noindent The proof is analogous to the one for Proposition~\ref{Prop:ReduceToCopies} and is skipped. 
As a consequence, we obtain a telescoping algorithm for $\dfield{Q(\EE)}{\sigma}$ whenever we are given a telescoping algorithm for the \pisiSE-field extensions $\dfield{Q(\tilde{\EE})}{\sigma_s}$ of $\dfield{\FF}{\sigma^n}$ with $0\leq s<n$. 

\smallskip

\noindent Finally, we will solve Problem~PTDR based on the above telescoping technology.

\begin{theorem}\label{Thm:PTRing}
Let $\dfield{\EE}{\sigma}$ be a basic \rpisiSE-extension of a constant-stable difference field $\dfield{\FF}{\sigma}$ with~\eqref{Equ:SingleRPS} where the \rE-monomial $y$ has order $n$.
\begin{enumerate}
\item If Problems~PFLDE and~PMT are solvable for $\dfield{\FF}{\sigma^n}$ then Problem~PTDR is solvable for $\dfield{\EE}{\sigma}$.
\item If $\dfield{\FF}{\sigma^n}$ is $\sigma^*$-computable and Problem~PFLDE is solvable for $\dfield{\FF}{\sigma^n}$ then Problem~PTDR is solvable for $\dfield{Q(\EE)}{\sigma}$.
\end{enumerate}
\end{theorem}
\begin{proof}
(1) Take the idempotent elements~\eqref{Equ:SpecificEi} with $0\leq s<n$. Let $\vect{f}=(f_1,\dots,f_d)\in\EE^d$.
Define $\tilde{f}_{i,s}=(f_i+\sigma(f_i)+\dots+\sigma(f_i)^{n-1})|_{y\to \alpha^{n-1-s}}\in\tilde{\EE}$ for $0\leq s<n$ and $1\leq i\leq d$, and set $\vect{h}_s=(\tilde{f}_{1,s},\dots,\tilde{f}_{d,s})\in\tilde{\EE}^d$ for $0\leq s<n$. Since we can solve Problems~PFLDE and~PMT for $\dfield{\FF}{\sigma^n}$, we can solve~PTDR for $\dfield{\tilde{\EE}}{\sigma_s}$ with $0\leq s<n$ by Theorem~\ref{Thm:RPSReduction}. Thus we can compute the bases of $V_s:=V(\vect{h}_s,\dfield{\tilde{\EE}}{\sigma_s})$ for $0\leq s<n$. By linear algebra we can calculate a basis, say $\{(c_{i,1},\dots,c_{i,d},\tilde{g}_i)\}_{1\leq i\leq r}$, of 
$$W:=\{(c_1,\dots,c_d,e_0\,h_0+\dots+e_{n-1}\,h_{n-1})|\,(c_1,\dots,c_d,h_s)\in V_s\text{ for all }0\leq s<n\}.$$
Now set $f'_i=c_{i,1}f_1+\dots+c_{i,d}\,f_d-(\sigma(\tilde{g}_i)-\tilde{g}_i)$ for $1\leq i\leq r$. By Proposition~\ref{Prop:ReduceToCopies} it follows that $f'_i\in\KK[y]$.
Thus we can calculate $g'_i\in\KK[y]$ such that $\sigma(g'_i)-g'_i=f'_i$ holds for $1\leq i\leq r$; for details see the proof of Proposition~\ref{Prop:ReduceToCopies}. Hence $\{(c_{i,1},\dots,c_{i,d},\tilde{g}_i+g'_i\}_{1\leq i\leq r}$ forms a basis of $V(\vect{f},\EE)$ by Proposition~\ref{Prop:ReduceToCopies}.\\
(2) The proof is analogous to part (1). We start with $\vect{f}=(f_1,\dots,f_d)\in Q(\EE)^d$ and define the vectors $\vect{h}_s\in Q(\tilde{\EE})^d$ for $0\leq s<n$ as in part (1).
Since $\dfield{\FF}{\sigma^n}$ is $\sigma^*$-computable and Problem~PFLDE is solvable for $\dfield{\FF}{\sigma^n}$, we can compute bases $V_s=V(\vect{h}_s,Q(\tilde{\EE}))$ for $0\leq s<n$ by Theorem~\ref{Thm:PSReduction}. The remaining calculation steps are as in part (1).
\end{proof}

\noindent We remark that any of the difference fields mentioned in Remark~\ref{Remark:ConstructiveVersion} satisfy the required properties of Theorem~\ref{Thm:PTRing}.
We observe further that the obtained algorithm is based on solving Problem~PTDR in $\dfield{\tilde{\EE}}{\sigma_s}$ for $0\leq s<n$. Here the summands and multiplicands blow up (see e.g., Example~\ref{Exp:Structure2}) which results in heavy ring calculations. In contrast to that the algorithm presented in~\cite{Schneider:16a} works just with the automorphism $\sigma$ and keeps the multiplicands and summands as simple as possible.
Nevertheless, the new approach is interesting on its own and yields the first algorithm that solves Problem~PTDR for the total ring of fractions of \rpisiSE-extensions. 
\section{Application 2: Symbolic summation and the transcendence of sequences}\label{Sec:Application}

We will utilize the constructive aspects of the previous sections in order to obtain a fully automatic toolbox for symbolic summation that is implemented within the package \texttt{Sigma}~\cite{Schneider:07a,Schneider:13a}. In Subsection~\ref{Sec:RepresentExprToRPS} (Propositions~\ref{Prop:TacticForEAR} and~\ref{Prop:ComputeRPSForExpr}) we will show how one can compute nested sum and product representations where the sums and products are algebraically independent among each other. In addition, we will solve the zero-recognition problem within such expressions.
Moreover, we will provide an automatic machinery in Subsection~\ref{Sec:PLDE} that enables one to tackle the parameterized telescoping problem, in particular the creative telescoping paradigm, for nested sums over nested products. Furthermore,
we will enhance substantially the techniques of~\cite{Schneider:10c} in Theorem~\ref{Thm:ParaProvedTrans} to show the transcendence of sequences using the parameterized telescoping paradigm. 

In the following we will focus on sequences that can be generated by the class of
nested sums over nested products.
The subclass of \notion{nested sums over hypergeometric products} can be defined recursively as follows. Let $f(k)$ be an expression that evaluates at non-negative 
integers (from a certain point on) to elements of a field $\set K$. 

\begin{itemize}
\item $f(k)$ is called a \notion{hypergeometric product w.r.t.\ $\KK$ and $k$} if it is given in the form $f(k)=\prod_{j=l}^kh(j)$ with $l\in\NN$ and a rational function $h(j)\in\KK(j)$; here $l$ is chosen big enough such that $h(\nu)$ has no pole and is non-zero for all $\nu\in\NN$ with $\nu\geq l$.
\item $f(k)$ is called a \notion{nested hypergeometric product w.r.t.\ $\KK$ and $k$} if it is given in the form $\prod_{j=l}^k h_1(j)h_2(j)\dots h_\nu(j)$ with $l\in\NN$ where for $1\leq i\leq \nu$ one of the following holds:
\begin{itemize}
 \item $h_i(j)\in\KK(j)^*$ where $h_i(\nu)$ has no pole and is non-zero for all $\nu\geq k$;
 \item $h_i(j)$ is a nested hypergeometric product w.r.t.\ $\KK$ and $j$ where $h_i(j)$ is free of $k$.
 \end{itemize}
\item $f(k)$ is called a \notion{nested sum expression over hypergeometric products (or over nested hypergeometric products) w.r.t $\KK$ and $k$}
if it is composed recursively by
\begin{itemize}
\item elements from the rational function field $\set K(k)$;
\item hypergeometric products (or nested hypergeometric products) w.r.t.\ $\KK$ and $k$;
\item the three operations\footnote{We do not allow divisions, i.e., sums and products may not occur in denominators. However, we can write ${(\prod_{j=l}^k h(j)})^{-1}$ as $\prod_{j=l}^k \frac1{h(j)}$ which allows to represent Laurent polynomial expressions.}
($+,-,\cdot$); 
\item sums of the form $\sum_{j=l}^kh(j)$ with $l\in\set N$ and with 
$h(j)$ being a nested sum expression over hypergeometric products (or over nested hypergeometric products) w.r.t.\ $\KK$ and $j$ and being free of $k$; here $l$ is chosen big enough such that $h(j)|_{j\to\nu}$ does not introduce poles for any $\nu\geq l$. $\sum_{j=l}^kh(j)$ is also called a nested sum over hypergeometric products (or over nested hypergeometric products) w.r.t.\ $k$ and $\KK$.
\end{itemize}
\end{itemize}

\begin{example}
$\prod_{i=1}^k(-1)(=(-1)^k)$, $\prod_{i=1}^k2(=2^k)$, $\prod_{i=1}^k\frac{n-i+1}{i}(=\binom{n}{k})$, $\prod_{i=1}^k\frac{i}{n-i+1}(=\binom{n}{k}^{-1})$ are hypergeometric products w.r.t.\ $\QQ(n)$ and $k$. The expressions on the right hand side of~\eqref{Equ:NestedPiEmb} are nested hypergeometric products w.r.t.\ $\QQ$ and $k$. The expressions 
$E_1(k)$, $E_2(k)$, $E_3(k)$ in~\eqref{Equ:EExpr} are nested sums over hypergeometric products w.r.t.\ $\QQ$ and $k$. 
\end{example}

\noindent More generally, we introduce \notion{$q$-hypergeometric products} $\prod_{j=l}^kh(q^j)$ with $h(x)\in\KK(x)$ where $\KK(x)$ and $\KK=\KK(q)$ are rational functions fields. Further, we introduce \notion{$q$-mixed hypergeometric products} $\prod_{j=l}^kh(j,q_1^j,\dots,q_v^j)$ with $h(x,x_1,\dots,x_v)\in\KK(x,x_1,\dots,x_v)$ where $\KK(x,x_1,\dots,x_v)$ and $\KK=\KK'(q_1,\dots,q_v)$ are rational function fields. In what follows, by a \notion{simple product} we mean a hypergeometric, $q$-hypergeometric or $q$-mixed hypergeometric product. If one takes nested versions of these products (the same definition as for nested hypergeometric products), one obtains the so-called \notion{nested products}. Finally, the definitions of nested sum expressions over simple products/nested products carry over immediately. The sums are also called \notion{nested sums over simple/nested products}.

We remark that the class of nested sums over hypergeometric products contains as special cases harmonic sums~\cite{Bluemlein:99,Vermaseren:99},
and more generally, generalized harmonic sums~\cite{Moch:02,ABS:13}, cyclotomic 
harmonic sums~\cite{ABS:11} or nested binomial 
sums~\cite{ABRS:14}, which occur as basic building blocks, e.g., in combinatorics or particle physics~\cite{Schneider:16b}. Further, d'Alembertian solutions~\cite{Abramov:94} (and Liouvillian solutions~\cite{Singer:99} given by the interlacing of d'Alembertian solutions~\cite{Reutenauer:12,Petkov:2013}) can be represented by nested sums over simple products.

Any nested sum or product with upper bound $k$ can be evaluated at $\nu$ for any non-negative integer $\nu$ since poles are excluded by definition. In particular, a nested sum expression over simple/nested products $f(k)$ defined as above can be evaluated at $k=\nu$ if the rational functions outside of sums and products do not introduce any poles. If this is the case, we write $f(\nu)$ or $f|_{k\to\nu}$ or $f(k)|_{k\to\nu}$ to perform the evaluation for $\nu\in\NN$. 

\smallskip

We emphasize that basic \rpisiSE-extensions in combination with an appropriately chosen $\KK$-embedding represent exactly this class of expressions. Namely, suppose that we are given 
one of the difference fields $\dfield{\FF}{\sigma}$ of Example~\ref{Exp:qRat} together with the $\KK$-embedding given in the Examples~\ref{Exp:Embedding1} or~\ref{Exp:Embedding2}. Moreover, suppose that we constructed a basic \rpisiSE-extension $\dfield{\EE}{\sigma}$ of $\dfield{\FF}{\sigma}$ with $\EE=\FF\lr{t_1}\dots\lr{t_e}$ and $\KK=\const{\EE}{\sigma}$ and that we constructed a $\KK$-embedding $\fct{\tau}{\EE}{\seqK}$ by iterative applications of Lemma~\ref{Lemma:LiftEvToPoly} (see also  Theorem~\ref{Thm:RPiSiImpliedEmbedding}). 
Within this process we get in addition a defining function $\ev$ together with an $o$-function $L$ and a $z$-function $Z$ for $\tau$. Note that the $t_i$ are mapped to~\eqref{Equ:SumProdHom} where the summands or multiplicands are given in terms of the variables $t_1,\dots,t_{i-1}$. Expanding this definition leads precisely to nested sums over nested products.\\ 
In this regard, we will introduce the following definition. Take $f\in\EE$ and replace all occurrences of $x$ by the symbolic variable $k$ and all $t_i$ by the nested sums and products where the outermost upper bound is $k$. 
The via $\tau$ (or $\ev$) derived expression will be also called an \notion{$\tau$-induced (or $\ev$-induced) expression} w.r.t.\ $k$ and will be denoted by $\expr_k(f)$. 
By construction,  
$$\ev(f,\nu)=\expr_k(f)|_{k\to\nu}$$
for all $\nu\in\NN$ with $\nu\geq L(f)$. In particular, $\expr_k$ satisfies the evaluation properties~\eqref{Ev:Const}--\eqref{Ev:Shift}; here the bound $\delta$ can be obtained by the given $o$-function $L$ and $z$-function $Z$.

\begin{example}
Consider the \rpisiSE-extension $\dfield{\EE}{\sigma}$ of the difference field $\dfield{\QQ(n)(x)}{\sigma}$ with $\EE=\QQ(n)(x)[y]\lr{p_1}\lr{p_2}[s_1]$ and the $\KK$-embedding $\tau$ defined by~\eqref{Equ:s1Embed} from
Example~\ref{Exp:ConstructEmbedding}. Then for 
$\beta_2=\frac{-y}{(x+1)(x+2)}\in\EE$ we get 
\begin{equation}\label{Equ:Beta2}
\expr_k(\beta_2)=\frac{1}{(k+1)(k+2)}\tprod_{i=1}^k(-1)=\frac{1}{(k+1)(k+2)}(-1)^k
\end{equation}
and for $\beta_3=\tfrac{1}{(x+1)^2 (x+2)} \big(
        3
        +x
        -(x+1) (x+2) y \big(
                2 s_1+1\big)
\big)\in\EE$ we obtain
\begin{equation}\label{Equ:Beta3}
\expr_k(\beta_3)=\frac{1}{(k+1)^2 (k+2)} \big(
        3
        +k
        -(k+1) (k+2) (-1)^k \big(
                2 \tsum_{i=1}^k\tfrac{(-1)^i}{i}+1\big)
\big).
\end{equation}
\end{example}

\subsection{Automatic representations of nested sums and products in basic \rpisiSE-extensions}\label{Sec:RepresentExprToRPS}

\noindent We consider the following key problem of indefinite summation.
\medskip

\noindent\fbox{\begin{minipage}{12.9cm}
\noindent \textbf{Problem EAR:} Elimination of Algebraic Relations.\\
\textit{Given} a nested sum expression $A(k)$ over nested products.\\
\textit{Find} a $\delta\in\NN$ and a nested sum expression $B(k)$ over nested products with
\begin{enumerate}
 \item $A(\nu)=B(\nu)$ for all $\nu\geq\delta$;
 \item the sequences produced by the sums and products occurring in $B(k)$ (except products over roots of unity) are algebraically independent.
\end{enumerate}
\end{minipage}}
\medskip

\noindent In part~(1) of Proposition~\ref{Prop:TacticForEAR} a recipe that solves Problem~EAR will be provided. In particular, if Problem~EAR is tackled in this way, it is shown in part~(2) of Proposition~\ref{Prop:TacticForEAR} that the zero-recognition problem is solved automatically.

\begin{proposition}\label{Prop:TacticForEAR}
Let $A(k)$ be a nested sum expression over nested products w.r.t. $\KK$ and $k$. Let $\dfield{\EE}{\sigma}$ be an \rpisiSE-extension of one of the difference fields $\dfield{\FF}{\sigma}$ from Example~\ref{Exp:qRat} with constant field $\KK$; let $\fct{\tau}{\EE}{\seqK}$ be a $\KK$-embedding given by iterative application of Lemma~\ref{Lemma:LiftEvToPoly} starting with the embedding given in Examples~\ref{Exp:Embedding1} or~\ref{Exp:Embedding1}. Let $a\in\EE$ with $B(k):=\expr_k(a)$ and $\delta\in\NN$ be such that $A(\nu)=B(\nu)$ holds for all $\nu\in\NN$ with $\nu\geq\delta$.
Then:
\begin{enumerate}
\item $B(k)$ and $\delta$ are a solution of Problem~EAR.
\item $B(k)$ is the zero-expression iff $A(\nu)=0$ for all $\nu\geq\lambda$ for some $\lambda\in\NN$.
\end{enumerate}
\end{proposition}
\begin{proof}
(1) To show explicitly that $B(k)$ is a solution of Problem~EAR, reorder $\dfield{\EE}{\sigma}$ to the form~\eqref{Equ:APSOrdered} where $\AA=\FF$. Then using the fact that $\tau$ is a $\KK$-embedding we get 
$$\tau(\EE)=\overbrace{\underbrace{\tau(\FF)[\tau(y_1),\dots,\tau(y_l)]}_{=:R}[\tau(p_1),\tau(p_1^{-1})]\dots[\tau(p_r),\tau(p_r^{-1})]}^{=:\LL}[\tau(s_1)]\dots[\tau(s_v)]$$
where $R$ is the ring given by the root of unity sequences, $\LL$ is the ring extension of Laurent-polynomials with the algebraically independent sequences $\tau(p_1),\tau(p_1^{-1}\!),\dots,\tau(p_r),\tau(p_r^{-1}\!)$, and $\LL[\tau(s_1)]\dots[\tau(s_v)]$ forms the polynomial ring with the algebraically independent sequences $\tau(s_1),\dots,\tau(s_v)$. Hence $B$ and $\delta$ are indeed a solution of Problem~EAR.\\
(2) $A(\nu)=B(\nu)$ holds for all $\nu\geq\delta$. This implies ``$\Rightarrow$''.\\ 
``$\Leftarrow$'': Suppose that $B(k)$ is not the zero-expression. Then $B\neq0$ implies $a\neq0$ and thus $\tau(a)\neq\vect{0}$. Hence for any $\lambda\in\NN$ there exists a $\nu\geq\lambda$ with $0\neq B(\nu)=A(\nu)$. 
\end{proof}

\noindent Suppose we are given a nested sum expression $A(k)$ over nested products w.r.t.\ $\KK$ and $k$. Then we will carry out the strategy introduced in Proposition~\ref{Prop:TacticForEAR} as follows.

\smallskip

\begin{description}
\item[Preparation:] Determine a $\rho\in\NN$ such that 
$A(\nu)$ can be evaluated for all $\nu\in\NN$ with $\nu\geq\rho$. Namely, consider the denominators in $A(k)$ which do not arise inside of sums and products, and determine the finite number of zeros, say $z_1,\dots,z_r\in\NN$, that can arise there (for the $q$--case and $q$--mixed case see~\cite[Sec.~3.7]{Bauer:99}). Then we can take $\rho=\max(z_1,\dots,z_r)+1$. In addition, take the appropriate difference field $\dfield{\FF}{\sigma}$  with the $\KK$-embedding $\fct{\tau}{\FF}{\seqK}$ from Examples~\ref{Exp:Embedding1} or~\ref{Exp:Embedding2}.

\smallskip

\item[Process all products:] Try to compute a basic $R$\piE-extension
 $\dfield{\HH}{\sigma}$ of $\dfield{\FF}{\sigma}$ and, using Lemma~\ref{Lemma:LiftEvToPoly}, extend
 $\tau$ to a $\KK$-embedding $\fct{\tau}{\HH}{\seqK}$ together with an $o$-function $L$ and a $z$-function $Z$ such that the following holds: for any product $P(k)$ arising in $A(k)$ (here also the products inside of sums and products are handled where the upper bounds are replaced with $k$) one obtains a $p\in\HH$ and $\delta_p\in\NN$ such that $\expr_k(p)|_{k\to\nu}=P(\nu)$ holds for all $\nu\geq\delta_{p}$; note that $\tau(p)=\langle P(k)\rangle_{k\geq 0}$. If this fails, STOP.  
\end{description}

\begin{remark}\label{Remark:ProductTranslation} 
For a finite set of hypergeometric products w.r.t.\ $\KK$ and $k$ with a rational function field $\KK=\KK'(n_1,\dots,n_{\lambda})$ and $\KK'=\QQ$ this product representation can be calculated by the algorithms given in~\cite{Schneider:05c,DR2}. In~\cite{Ocansey:16} these ideas were generalized to simple products, i.e., to ($q$--)hyperge\-ometric and $q$--mixed hypergeometric expressions where the subfield $\KK'$ of $\KK$ can even be an algebraic number field. For nested products the automatic construction is more subtle and is currently under investigation.
\end{remark}

 \begin{description}
 \item[Process all sums:] By iterative application of Theorem~\ref{Thm:RPSCharacterization} and Lemma~\ref{Lemma:LiftEvToPoly} we will construct a \sigmaSE-extension $\dfield{\EE}{\sigma}$ of $\dfield{\HH}{\sigma}$, and we will extend $\tau$ to a difference ring embedding $\fct{\tau}{\EE}{\seqK}$ such that for any sum $S(k)$ in $A(k)$ (we take also sums occurring inside of other sums and replace the outermost summation bound with $k$) we can take an element $s\in\EE$ with $\tau(s)=\langle S(k)\rangle_{k\geq0}$. We proceed stepwise.\\ 
 Suppose that we have already treated the nested sums and products $A_1(k),\dots,A_n(k)$ of $A(k)$ where for each $A_j(k)$ with $1\leq j\leq n$ we are given an $a_j\in\EE$ with $\tau(a_j)=\langle A_j(i)\rangle_{i\geq0}$ and we are given a $\delta_j\in\NN$ such that $A_j(\nu)=\expr_k(a_j)|_{k\to \nu}$ holds for all $\nu\geq\delta_j$.
 We process the next sum $A_{n+1}(k)=\sum_{i=\lambda}^kh(i)$ with $\lambda\in\NN$ where all sums and products in $h(k)$ have been treated earlier.
 Hence replacing all occurrences of $A_j$ in $h(k)$ by $a_j$ leads to $\beta'\in\EE$ such that $\expr_k(\beta')|_{k\to\nu}=h(\nu)$ holds for all $\nu\geq\delta:=\max(\lambda,\delta_1,\dots,\delta_n,L(\beta'))$ (it suffices to take only those $\delta_j$ where $t_j$ occurs in $\beta$). Note that $\tau(\beta')=\langle h(\nu)\rangle_{\nu\geq0}$. Now we set $\beta:=\sigma(\beta')$ and consider two cases.
 \begin{description}
 \item[(1)] With the algorithms from~\cite{Schneider:16a} one shows that there is no $g\in\EE$ with $\sigma(g)=g+\beta$. Hence we can construct the \sigmaSE-extension $\dfield{\EE(t_{e+1})}{\sigma}$ of $\dfield{\EE}{\sigma}$ with $\sigma(t_{e+1})=t_{e+1}+\beta$. In addition, we can extend $\tau$ to $\fct{\tau}{\EE[t_{e+1}]}{\seqK}$ with~\eqref{Equ:SumProdHom} where $r:=\delta$ and $c=\sum_{k=\lambda}^{\delta-1} h(i)\in\KK$. With $a_{n+1}:=t_{e+1}$ and $\delta_{n+1}:=\delta$ we have for all $\nu\geq\delta_{n+1}$:
 $\ev(a_{n+1},\nu)=\ev(t_{e+1},\nu)=\sum_{i=\delta}^{\nu}\ev(\beta,i-1)+c=\sum_{i=\lambda}^{\nu} h(k)=A_{n+1}(\nu).$
 \item[(2)] One obtains a $g\in\EE$ with $\sigma(g)=g+\beta$. Take $\delta_{n+1}:=\max(\delta,L(g))$ and set $c:=-\ev(g,\delta)+\sum_{i=\lambda}^{\delta} h(i)\in\KK$. With $a_{n+1}:=g+c$ this yields
 $\ev(a_{n+1},\nu+1)=\ev(a_{n+1},\nu)+\ev(\beta,\nu)$ and $A_{n+1}(\nu+1)-A_{n+1}(\nu)=\ev(\beta,\nu)$ for all $\nu\geq\delta_{n+1}$. Since the initial values agree, i.e.,
 $\ev(a_{n+1},\delta)=\ev(g+c,\delta)=\ev(g,\delta)+c=\sum_{i=\lambda}^{\delta}h(i)=A_{n+1}(\delta),$ 
 we conclude that $A_{n+1}(\nu)=\ev(a,\nu)=\expr_k(a)|_{k\to\nu}$ for all $\nu\geq\delta_{n+1}$.
 \end{description}
  $A_{n+1}(k)$ has been tackled and we continue with the remaining sums in $A(k)$.
 \item[Combine expressions:] Let $A_1(k),\dots,A_{m}(k)$ be all sums and products occurring in $A(k)$ with the corresponding $a_1,\dots,a_m\in\EE$ and $\delta_1,\dots,\delta_m\in\NN$, respectively. Replace the $A_j(k)$ by the $a_j$ in $A(k)$ which yields $b\in\EE$. With $\rho$ from the preprocessing step we define $\delta:=\max(\rho,\delta_1,\dots,\delta_{m},L(b))$ and $B:=\expr_k(b)$. By construction we get $B(\nu)=\ev(b,\nu)=A(\nu)$ for all $\nu\geq\delta$.
\end{description}

\medskip

\noindent If one succeeds in this construction, one has solved Problem~EAR by Proposition~\ref{Prop:TacticForEAR}. In particular, if one restricts to the class of nested sums over simple products, this method turns into a complete algorithm by Remark~\ref{Remark:ProductTranslation} which can be summarized as follows.

\begin{proposition}\label{Prop:ComputeRPSForExpr}
Let $A(k)$ be a nested sum expression over simple products w.r.t. $\KK$ and $k$ where $\KK$ is a rational function field over an algebraic number field. Then $\dfield{\EE}{\sigma}$ and $\fct{\tau}{\EE}{\seqK}$ as assumed in Proposition~\ref{Prop:TacticForEAR} can be computed. Further, Problem~EAR and the zero recognition problem for $A(k)$ are solved.
\end{proposition}

\noindent The described method for nested sums over nested products can be executed within the summation package \texttt{Sigma} with the function call \texttt{SigmaReduce[A,k]}.

\begin{example}\label{Exp:SymbolicSummation}
Consider the expressions in~\eqref{Equ:EExpr} and apply \texttt{SigmaReduce} to them. \texttt{Sigma} starts with the \pisiSE-field $\dfield{\KK(x)}{\sigma}$ with $\sigma(x)=x+1$ over $\KK=\QQ(n)$ together with the $\KK$-embedding $\fct{\tau}{\KK(x)}{\seqK}$ from Example~\ref{Exp:NestedP2}. Then the objects given in~\eqref{Equ:EExpr} (parsed in the order they are written) are represented in an \rpisiSE-extension. To be more precise, as worked out in the Examples~\ref{Exp:MainDRNaiveDef} and~\ref{Exp:ConstructRPiSiExt} \texttt{Sigma} constructs the basic \rpisiSE-extension $\dfield{\AA}{\sigma}$ of $\dfield{\KK(x)}{\sigma}$ with $\AA=\KK(x)[y]\ltr{p_1}\ltr{p_2}[s_1]$ (here we set $s_1=t_1$) where $\sigma(y)=-y$, $\sigma(p_1)=2\,p_1$, $\sigma(p_2)=\frac{n-x}{x+1}p_2$ and $\sigma(s_1)=s_1-\frac{y}{x+1}$. Furthermore, \texttt{Sigma} extends $\tau$ to the $\KK$-embedding $\fct{\tau}{\AA}{\seqK}$ with~\eqref{Equ:s1Embed}.
In short, $x$, $y$, $p_1$, $p_2$ and $s_1$ represent the summation objects
$k,(-1)^k,\,2^k,\,\binom{n}{k}=\prod_{i=1}^k\frac{n-i+1}{i},\,E_1(k)$, respectively.\\
Next, \texttt{Sigma} treats $E_2(k)=\sum_{i=1}^kh_2(i)$ with $h_2(k)=\frac{(-1)^k}{k(k+1)}$ given in~\eqref{Equ:EExpr}. Here we can represent the summand $h_2(k)$ by $\beta'_2=\frac{y}{x(x+1)}$. Set
$\beta_2:=\sigma(\beta'_2)=\frac{-y}{(x+1)(x+2)}$. Then we have that $\tau(\beta_2)=\langle h_2(\nu+1)\rangle_{\nu\geq0}$ where  $\expr_k(\beta_2)|_{k\to\nu}=h_2(\nu+1)$ for all $\nu\geq\delta_0:=0$; see also~\eqref{Equ:Beta2}.
As observed in Examples~\ref{Exp:ConstructRPiSiExt} or~\ref{Exp:FindIsomorphism}, one finds $g=\frac{2 (x+1) s_1-y}{x+1}\in\AA$ such that $\sigma(g)=g+\beta_2$ holds. Thus we get $\Shift(\tau(g))=\tau(g)+\tau(\beta_2)$. In particular, for $\delta'_0=\max(L(g),L(\beta_2),\delta_0)=0$ and 
$G_2(k)=\expr_k(g)=\tfrac1{{k+1}}\big(2 (k+1) E_1(k)-(-1)^k\big)$
we obtain $G_2(k+1)=G_2(k)+h_2(k+1)$ for all $k\geq\delta_0'=0$. Since $E_2(k+1)=E_1(k)+h_2(k+1)$, it follows that $G_2(k)+c'$ and $E_2(k)$ agree for all $k\geq0$ if they agree for $k=\delta'_0=0$. Taking $c'=1$ we get $E_2(k)=G_2(k)+1$ which is the second identity in~\eqref{Equ:IdsForEi}.\\  
Finally, \texttt{Sigma} treats $E_3(k)=\sum_{j=1}^k h_3(j)$. The sums and products in $h_3(k)$ have been represented already in $\dfield{\AA}{\sigma}$. As a consequence $h_3(k)$ can be represented by $\beta_3'=\frac{-1
+y (1+x)(
        1
        +2 s_1)
}{x (1+x)}\in\AA$. In particular, $\beta_3:=\sigma(\beta'_3)$ represents $h_3(k+1)$:
$\tau(\beta_3)=\langle h_3(k+1)\rangle_{k\geq0}$ and by construction we get $\expr_k(\beta_3)=h_3(k+1)$; see also~\eqref{Equ:Beta3}.
We demonstrate the importance of enhanced telescoping algorithms to achieve simplifications.\\
$\bullet$ \textit{Naive telescoping.}  
We do not find a $g\in\AA[s_1]$ with $\sigma(g)=g+\beta_3$; see Example~\ref{Exp:FindIsomorphism}. Hence we take the \sigmaSE-extension $\dfield{\AA[s_2]}{\sigma}$ of $\dfield{\AA}{\sigma}$ with $\sigma(s_2)=s_2+\beta_3$ and extend $\tau$ to $\fct{\tau}{\AA[s_2]}{\seqK}$ with $\tau(s_2)=\langle \sum_{i=1}^k h_3(i)\rangle_{k\geq0}$. No simplification has been accomplished and $\expr_k(s_2)$ equals the right side of the third identity in~\eqref{Equ:IdsForEi}. Still we have solved Problem~EAR and concluded that $\langle E_3(k)\rangle_{k\geq0}$ is algebraically independent over $\tau(\AA)$.\\
$\bullet$ \textit{Refined telescoping.} Alternatively, we can use refined summation algorithms as worked out in~\cite{Schneider:05f,Schneider:08c,Schneider:10b,Schneider:15}. Activating \texttt{Sigma} one computes the \sigmaSE-extension $\dfield{\AA[s'_2]}{\sigma}$ of $\dfield{\AA}{\sigma}$ with $\sigma(s'_2)=s'_2+\frac{1}{(x+1)^2}$ together with the solution\footnote{For the given $\dfield{\AA[s'_2]}{\sigma}$ the solution $g\in\AA[s'_2]$ has been computed in Example~\ref{Exp:ComputeTele}.}
$g=s'_2
+s_1
+s_1^2
+\frac{1}{x+1}
\in\AA[s_2]$ of $\sigma(g)=g+\beta_3$. Extending the $\KK$-embedding to $\fct{\tau}{\AA[s'_2]}{\seqK}$ with $\tau(s'_2)=\langle\sum_{i=1}^k\frac{1}{i^2}\rangle_{k\geq0}$ we obtain $\Shift(\tau(g))=\tau(g)+\tau(\beta_3)$
and
$G_3(k+1)=G_3(k)+h_3(k+1)$ for $k\geq \delta_1:=\max(L(g),L(\beta_3),\delta'_0)=0$ where 
$G_3(k)=\expr_k(g)=\sum_{j=1}^k\frac1{i^2}
+\sum_{j=1}^k\frac{(-1)^j}{j}
+\big(\sum_{j=1}^k\frac{(-1)^j}{j}\big)^2
+\frac{1}{k+1}$. 
Thus $G_3(k)+c''$ with $c''\in\KK$ and $E_3(k)$ agree if they are equal at $k=\delta_1=0$. This is the case for $c''=-1$ yielding
\begin{equation}\label{Equ:E3ClosedForm}
E_3(k)=\expr_k(g-1)=G_3(k)-1=\tsum_{j=1}^k \tfrac{1}{j^2}
+
\tsum_{j=1}^k \tfrac{(-1)^j}{j}
+\Big(
        \tsum_{j=1}^k \tfrac{(-1)^j}{j}\Big)^2
-\frac{k}{k+1}.
\end{equation}
Thus $\langle E_3(k)\rangle_{k\geq0}=\tau(g-1)$. We remark that we get the following (Laurent) polynomial ring over the coefficient ring $R=\tau(\QQ(n)(x))[\langle (-1)^k\rangle_{k\geq0}]$:  
$$R\big[\langle 2^k\rangle_{k\geq0},\langle 2^{-k}\rangle_{k\geq0}\big]\big[\langle\tbinom{n}{k}\rangle_{k\geq0},\langle\tbinom{n}{k}^{-1}\rangle_{k\geq0}\big]\big[\langle \tsum_{i=1}^{k}\tfrac{(-1)^i}{i}\rangle_{k\geq0},\langle\tsum_{i=1}^{k}\tfrac{1}{i^2}\rangle_{k\geq0}\big].$$
\end{example}

\begin{example}
When executing \texttt{SigmaReduce[A,k]} with the input $A(k)=\prod_{i=1}^k \prod_{j=1}^i j!$, the basic \piE-extension $\dfield{\KK(x)\ltr{p_1}\ltr{p_2}\ltr{p_3}}{\sigma}$ of $\dfield{\KK(x)}{\sigma}$ from Example~\ref{Exp:NestedP2} is constructed together with the $\KK$-embedding $\fct{\tau}{\KK(x)\ltr{p_1}\ltr{p_2}\ltr{p_3}}{\seqK}$ defined by~\eqref{Equ:NestedPiEmb}. 
Here $p_3$ represents $A(k)$ and we get $\tau(p_3)=\langle A(k)\rangle_{k\geq0}$. Summarizing, \texttt{Sigma} returns $A(k)$ without any simplification. As a by-product, we gain the insight that  
$$\tau(\KK(x))\big[\langle k!\rangle_{k\geq0},\langle k!^{-1}\rangle_{k\geq0}\big]\big[\langle \tprod_{i=1}^k i!\rangle_{k\geq0},\langle \tprod_{i=1}^k i!^{-1}\rangle_{k\geq0}\big]\big[\langle A(k)\rangle_{k\geq0},\langle A(k)^{-1}\rangle_{k\geq0}\big]$$ 
forms a (Laurent) polynomial ring over the field $\tau(\KK(x))$ of rational sequences. 
\end{example}


\subsection{Parameterized telescoping and algebraic independence of sum sequences}\label{Sec:PLDE}

The summation paradigm of parameterized telescoping in terms of nested sum expressions can be formulated as follows.

\medskip

\noindent\fbox{\begin{minipage}{12.9cm}
\noindent \textbf{Problem PT:} Parameterized Telescoping.\\
\textit{Given} nested sum expressions $F_1(k),\dots,F_d(k)$ over nested products w.r.t.\ $\KK$ and $k$.\\
\textit{Find} an appropriate\footnote{In the simplest version one searches for a $G(k)$ in terms of the objects occurring in the $F_i(k)$.}
nested sum expression $G(k)$ over nested products, find constants $c_1,\dots,c_d\in\KK$, not all zero, and find a $\delta\in\NN$ s.t.\ for all $k\geq\delta$ we have
\begin{equation}\label{equ:creaProbSeq}
G(k+1)-G(k)=c_1\,F_1(k)+\dots+c_d\,F_d(k).
\end{equation}
\end{minipage}}

\medskip

\noindent Suppose that we succeed in computing such a $\delta$, a $G(k)$ and the $c_i$. Then we can sum~\eqref{equ:creaProbSeq} over $k$ from $\delta$ to $a$ and obtain the sum relation
\begin{equation}\label{Equ:SumRelation}
c_1\,\sum_{k=\delta}^a F_1(k)+\dots+c_d\,\sum_{k=\delta}^a F_d(k)=G(a+1)-G(\delta).
\end{equation}
Note that the special case $d=1$ boils down to the telescoping problem for indefinite summation. Further, Problem~PT contains the summation paradigm of creative telescoping~\cite{Zeilberger:91} for definite summation, which we will illustrate in Example~\ref{Exp:CreativeTele} below. In this regard, we refer to~\cite{AequalB,PauleSchorn:95,PauleRiese:97,Bauer:99,CK:12,CSFFL:15} for the \hbox{($q$--mixed)}hypergeometric approach, to~\cite{Zeilberger:90a,Chyzak:00,Koutschan:13} for the holonomic approach, to~\cite{Wilf:92,Wegschaider,AZ:06} for the multi-sum approach, or to~\cite{Schneider:05f,Schneider:08c,Schneider:10b,Schneider:15} for further refinements of the difference field approach.

\medskip

The simplest form of Problem~PT can be solved within the summation package \texttt{Sigma} by executing the function call \texttt{ParameterizedTelescoping[$\{F_1,F_2,\dots,F_d\},k]$}. Here \texttt{Sigma} starts with the \pisiSE-field $\dfield{\FF}{\sigma}$ and $\KK$-embedding from Examples~\ref{Exp:Embedding1} or~\ref{Exp:Embedding2}. Using the machinery of Subsection~\ref{Sec:RepresentExprToRPS} (see Proposition~\ref{Prop:ComputeRPSForExpr})
\texttt{Sigma} computes an \rpisiSE-extension $\dfield{\EE}{\sigma}$ of $\dfield{\FF}{\sigma}$ and a $\KK$-embedding $\fct{\tau}{\EE}{\seqK}$ with a defining function $\ev$ with the following elements: $f_1,\dots,f_d\in\EE$ and $\lambda=\max(L(f_1),\dots,L(f_d))$ such that $\expr_k(f_i)|_{k\to\nu}=\ev(f_i,\nu)=F_i(\nu)$ holds for all $\nu\geq\lambda$ and $1\leq i\leq d$. We emphasize that this construction is fully algorithmic for nested sums over simple products.\\
Using the algorithms from~\cite{Schneider:16a,DR2} or from Section~\ref{Sec:Telescoping} (or enhanced telescoping algorithms from~\cite{Schneider:05f,Schneider:08c,Schneider:10b,Schneider:15}) we can now compute a basis of $V((f_1,\dots,f_d),\EE)$, and can decide if there exists a $g\in\EE$ and $(c_1,\dots,c_d)\in\KK^d\setminus\{\vect{0}\}$ such that~\eqref{Equ:PT} holds. If yes, we get
$$\ev(g,\nu+1)-\ev(g,\nu)=c_1\,\ev(f_1,\nu)+\dots+c_d\,\ev(f_d,\nu)$$
for all $\nu\geq\delta$ with $\delta:=\max(\lambda,L(g))$. In particular, setting $G(k):=\expr_k(g)$ we have that $G(\nu)=\ev(g,\nu)$ and $\ev(f_j,\nu)=\expr_k(f_i)|_{k\to\nu}=F_i(\nu)$ for all $\nu\geq\delta$. Thus $(c_{1},\dots,c_{d},G(k))$ is a solution of~\eqref{equ:creaProbSeq} for all $k\geq\delta$.

\begin{example}\label{Exp:CreativeTele}
A special case of parameterized telescoping is Zeilberger's creative telescoping paradigm~\cite{Zeilberger:91}. We illustrate it with the sum
\begin{equation}\label{Equ:DefiniteSummand}
S(n)=\sum_{k=0}^n F(n,k)=\sum_{k=0}^n\binom{n}{k} \big(
        (-2)^k
        +2^k
\big) \tsum_{i=1}^k\frac{(-1)^i}{i}.
\end{equation}
We set $F_i(k):=F(n+i-1,k)$ for $i\geq1$ and obtain
$F_i(k)=\prod_{j=1}^{i-1}\frac{n+j}{n-k+j}F(n,k)$. Now we try to find a solution of~\eqref{equ:creaProbSeq} for $d=1,2,3\dots$.
Given $F_1(k)$, we start with the \pisiSE-field $\dfield{\KK(x)}{\sigma}$ over $\KK=\QQ(n)$ with $\sigma(x)=x+1$. We parse the summation objects in~\eqref{Equ:DefiniteSummand} and construct the \rpisiSE-extension $\dfield{\EE}{\sigma}$ of $\dfield{\FF}{\sigma}$ with $\EE=\FF[y]\lr{p_1}\lr{p_2}[s_1]$ from Example~\ref{Exp:ConstructEmbedding} together with the $\QQ(n)$-embedding 
$\fct{\tau}{\EE}{\seqP{\QQ(n)}}$ with~\eqref{Equ:s1Embed}. Here we obtain $f_1=p_2\,s_1\big(
        y\, p_1
        +p_1
\big)$ with $\expr_k(f_1)=F_1(k)$. Activating the algorithms in \texttt{Sigma}, we fail to find a $g\in\EE$ such that~\eqref{equ:creaProbSeq} holds for $d=1$. Hence we proceed with $d=2,3,\dots$. Here we represent the $F_i(k)$ with  $f_i=\prod_{j=1}^{i-1}\frac{n+j}{n-x+j} f_0\in\EE$ for $i\geq1$, i.e., we have that $\expr_k(f_i)=F_i(k)$ for all $k\geq0$.
Eventually, at $d=5$ we succeed: we find 
$c_1=9 (n+1) (n+2)$, $c_2=12 (n+2)^2$, $c_3=-2 \big(
        n^2+5 n+9\big)$, $c_4=-4 (n+3)^2$, $c_5=(n+3) (n+4)$ 
and $g\in\EE$ (which is too big to print it here) such that~\eqref{Equ:PT} holds. Taking $G(k)=\expr_k(g)$ leads to the solution~\eqref{equ:creaProbSeq} with $d=5$ for $k\geq0$. Summing this equation over $k$ from $0$ to $a$ gives the sum relation
\begin{equation}\label{Equ:ExpParaSol}
c_1\,\sum_{k=0}^aF(n,k)+c_2\,\sum_{k=0}^aF(n,k+1)+\dots+c_5\,\sum_{k=0}^aF(n+4,k)=G(a+1)-G(0).
\end{equation}
Finally, setting $a=n$ and taking care of the missing terms produces
\begin{align*}
9 (n+1) (n+2) S(n)
+12 (n+2)^2 &S(n+1)
-2 \big(
        n^2+5 n+9\big) S(n+2)\\
-4 (n+3)^2 &S(n+3)
+(n+3) (n+4) S(n+4)
=-8.
\end{align*}
Solving this recurrence relation in terms of d'Alembertian solutions with the algorithms from~\cite{Abramov:89a,Petkov:92,Abramov:94,Abramov:96,Bron:00,Schneider:01,Schneider:05a}, simplifying the solutions by our advanced telescoping algorithms~\cite{Schneider:05f,Schneider:08c,Schneider:10b,Schneider:15} and taking the first 4 initial values of $S(n)$ produces
$$S(n)=(-1)^n 
\sum_{i=1}^n \frac{(-3)^i}{i}
-(-1)^n 
\sum_{i=1}^n \frac{(-1)^i}{i}
+3^n 
\sum_{i=1}^n \frac{\big(
        -\tfrac{1}{3}\big)^i}{i}
-3^n 
\sum_{i=1}^n \frac{\big(\tfrac{1}{3}\big)^i}{i}.
$$
\end{example}

Summarizing, a non-trivial solution of the parameterized telescoping problem provides the linear relation~\eqref{Equ:SumRelation}. Conversely, it is amazing that the non-existence of such a solution implies the algebraic independence of the sums given in~\eqref{Equ:SumRelation}. This aspect has been worked out in~\cite{Schneider:10c} in the setting of \pisiSE-field extensions; see also~\cite{Singer:08}. In the following we will generalize these concepts from the field to the ring setting. In this regard, we utilize~\cite[Proposition~1]{AS:15}; compare~\cite{Karr:85,Schneider:10a} for various field versions.

\begin{proposition}\label{Prop:TeleStructure}
Let $\dfield{\AA[t_1]\dots[t_e]}{\sigma}$ be a \sigmaSE-extension of $\dfield{\AA}{\sigma}$ with $\sigma(t_i)-t_i\in\AA$ for $1\leq i\leq e$ where $\KK:=\const{\AA}{\sigma}$ is a field. Let $g\in\AA[t_1]\dots[t_e]$ with $\sigma(g)-g\in\AA$. Then $g=\sum_{i=1}^e c_i\,t_i+w$ with $c_i\in\KK$ and $w\in\AA$.
\end{proposition}

\noindent With this result we can now generalize~\cite[Thm~3.1]{Schneider:10c} to the ring setting. 

\begin{theorem}\label{Thm:ConnectionDefSumAndProperSum} Let
$\dfield{\AA}{\sigma}$ be a difference ring  with constant field
$\KK$ and let $(f_1,\dots,f_d)\in\AA^d$. Then the following statements
are equivalent.

\begin{enumerate}
\item There are no $g\in\AA$ and $\vect{0}\neq(c_1,\dots,c_d)\in\KK^d$
with~\eqref{Equ:PT}.
\item There is a \sigmaSE-extension
$\dfield{\AA[t_1]\dots[t_d]}{\sigma}$ of $\dfield{\AA}{\sigma}$ with
$\sigma(t_i)=t_i+f_i$ for $1\leq i\leq d$.
\end{enumerate}
\end{theorem}
\begin{proof}
$(2)\Rightarrow(1)$: Suppose that~\eqref{Equ:PT} holds for some
$\vect{0}\neq(c_1,\dots,c_d)\in\KK^d$ and $g\in\AA$. In
addition, assume that there exists a \sigmaSE-extension
$\dfield{\EE}{\sigma}$ of $\dfield{\AA}{\sigma}$ with $\EE=\AA[t_1]\dots[t_d]$ and
$\sigma(t_i)=t_i+f_i$ for $1\leq i\leq d$. Then
$\sigma(g)-g=\sum_{i=1}^dc_i\big(\sigma(t_i)-t_i\big)
=\sigma(\sum_{i=1}^d c_i\,t_i)-\sum_{i=1}^d c_i\,t_i$, and thus
$\sigma(h)=h$ with $h=\sum_{i=1}^d c_i\,t_i-g$. Since not all $c_i$ are zero and $g\in\AA$, $h\notin\const{\EE}{\sigma}\setminus\KK$; a contradiction since $\dfield{\EE}{\sigma}$ is a \sigmaSE-extension of $\dfield{\AA}{\sigma}$.\\
$(1)\Rightarrow(2)$: Let $0\leq i<d$ be maximal such that
$\dfield{\AA[t_1]\dots[t_{i}]}{\sigma}$ is a \sigmaSE-extension of
$\dfield{\AA}{\sigma}$, but $t_{i+1}$ is not a \sigmaSE-monomial. Hence $\sigma(\gamma)-\gamma=f_{i+1}$ for some
$\gamma\in\AA[t_1]\dots[t_{i}]$, and 
thus
$\gamma=h+\sum_{j=1}^{i}c_j\,t_j$ for some $c_j\in\KK$ and  
$h\in\AA$ by Prop.~\ref{Prop:TeleStructure}. Then
$\sigma(h)-h=f_{i+1}-\sum_{j=1}^{i}
c_j(\sigma(t_j)-t_j)=f_{i+1}-\sum_{j=1}^ic_j\,f_j$, i.e., we get a solution of~\eqref{Equ:PT} in $\AA$.
\end{proof}

\noindent To this end, we arrive at the following result; for a special case see~\cite[Thm.~5.2]{Schneider:10c}.

\begin{theorem}\label{Thm:ParaProvedTrans}
Let $\dfield{\EE}{\sigma}$ be a basic \rpisiSE-extension of $\dfield{\FF}{\sigma}$ with constant field
$\KK:=\const{\FF}{\sigma}$, and let
$\fct{\tau}{\EE}{\seqK}$ be a $\KK$-embedding
with~\eqref{Equ:EvDef} together with an o-function $L$; let
$(f_1,\dots,f_d)\in\EE^d$.  Then the following
statements are equivalent:
\begin{enumerate}
\item There are no $g\in\EE$ and $\vect{0}\neq(c_1,\dots,c_d)\in\KK^d$
with~\eqref{Equ:PT}.

\item The sequences $\langle S_1(a)\rangle_{a\geq0},\dots,\langle S_d(a)\rangle_{a\geq0}$ given by
\begin{equation}\label{Equ:SumEvEmbed}
S_1(a):=\sum_{k=l}^a\ev(f_1,k),\dots,S_d(a):=\sum_{k=l}^a\ev(f_d,k)
\end{equation}
with $l\geq\max(L(f_1),\dots,L(f_d))\in\NN$ are algebraically independent over $\tau(\EE)$.
\end{enumerate}
\end{theorem}
\begin{proof}
$(1)\Rightarrow(2)$: Suppose that there does not exist a $g\in\EE$ and $\vect{0}\neq(c_1,\dots,c_d)\in\KK^d$ with~\eqref{Equ:PT}. Then there does not exist a $g\in\EE$ and $\vect{0}\neq(c_1,\dots,c_d)\in\KK^d$ with $\sigma(g)-g=c_1\,\sigma(f_1)+\dots+c_d\,\sigma(f_d)$. By Theorem~\ref{Thm:ConnectionDefSumAndProperSum} we can construct the \sigmaSE-extension $\dfield{\HH}{\sigma}$ of $\dfield{\EE}{\sigma}$ with $\HH=\EE[s_1]\dots[s_d]$ and $\sigma(s_i)=s_i+\sigma(f_i)$ for $1\leq i\leq d$. Note that $\dfield{\HH}{\sigma}$ is a basic \rpisiSE-extension of $\dfield{\FF}{\sigma}$. By iterative application of Lemmas~\ref{Lemma:LiftEvToPoly} and~\ref{Lemma:SimpleImpliesEmbedding} we can extend $\tau$ to a $\KK$-embedding $\fct{\tau}{\HH}{\seqK}$. In particular, we can take $l\in\NN$ with $l\geq\max(L(f_1),\dots,L(f_d))\in\NN$ and $c=0$ in~\eqref{Equ:SumProdHom} and get $\expr_k(s_j)|_{k\to n}=S_i(n)$ for all $n\geq l$. Consequently $\tau(\EE)[\langle S_1(n)\rangle_{n\geq 0},\dots,\langle S_d(n)\rangle_{n\geq 0}]$ forms a polynomial ring.\\
$(2)\Rightarrow(1)$: Assume that there exist a $g\in\EE$ and $\vect{0}\neq(c_1,\dots,c_d)\in\KK^d$ with~\eqref{Equ:PT}. Then we obtain~\eqref{equ:creaProbSeq} with $G(k)=\ev(g,k)$ and $F_i(k)=\ev(f_i,k)$ for all $k\geq\delta$ with $\delta:=\max(L(f_1),\dots,L(f_d),L(g))$. Therefore summing~\eqref{equ:creaProbSeq} over $k$ from $\delta$ to $a$ yields~\eqref{Equ:SumRelation} for all $a\geq\delta$. Note that for $l\in\NN$ with $l\geq\max(L(f_1),\dots,L(f_d))$ we have that $l\leq\delta$. Therefore we can adapt the lower bounds of the sums in~\eqref{Equ:SumRelation} to $l$ by taking care of compensating terms which are elements from $\KK$. Thus we obtain a slightly modified right hand side in~\eqref{Equ:SumRelation}. Thus the sequences produced by $S_1(a),\dots,S_d(a)$ are algebraically dependent over $\tau(\EE)$.
\end{proof}

\begin{example}
Take the \rpisiSE-extension $\dfield{\EE}{\sigma}$ of $\dfield{\FF}{\sigma}$ from Example~\ref{Exp:CreativeTele} (resp.\ from Examples~\ref{Exp:ConstructRPiSiExt} and~\ref{Exp:ConstructEmbedding}). 
As demonstrated in Example~\ref{Exp:CreativeTele}, we can represent the shifted versions $F_i(k)=F(n+i-1,k)$ of the summand $F(n,k)$ in~\eqref{Equ:DefiniteSummand} for all $i\geq1$ by $f_i\in\EE$. Namely, we have that $\ev(f_i,k)=F_i(k)=F(n+i-1,k)$ for all $k\geq0$.
In particular, we checked in Example~\ref{Exp:CreativeTele} that there is no $g\in\EE$ with~\eqref{Equ:PT} for $d=4$.
Hence the sequences~\eqref{Equ:SumEvEmbed} with $l=0$ and $d=4$ are algebraically independent over $\tau(\EE)$. We remark that for $d=5$ we obtain the relation~\eqref{Equ:ExpParaSol} for explicitly given $c_i$ and $G(a+1)-G(0)$.
\end{example}

\section{Conclusion}\label{Sec:Conclusion}

Starting from the results of~\cite{Karr:81,Schneider:16a} we derived new insight for basic \rpisiSE-extensions and 
found results similar to those known from the Galois theory of difference equations~\cite{Singer:99,Singer:08}. As a consequence we obtained new intrinsic characterizations of basic \rpisiSE-extensions based on the notion of simple difference rings, on the decomposition of interlaced difference rings, and on the embedding of difference rings into the ring of sequences. In particular, these results yield a new method to solve the parameterized telescoping problem within such a ring or its total ring of fractions. Moreover, we provided a constructive machinery to embed basic \rpisiSE-extensions into the ring of sequences. As a consequence we can justify in full generality that the summation package \texttt{Sigma} produces simplifications where the arising sums are algebraically independent. In this regard, we could generalize the difference field results of~\cite{Schneider:10c} to show that the non-existence of a parameterized telescoping solution (in particular, of a creative telescoping solution) in an \rpisiSE-extension provides a proof that certain indefinite nested sums are algebraically independent. 

Further investigations will be necessary in order to obtain similar results for simple \rpisiSE-extensions~\cite{Schneider:16a} which enable one to represent also expressions in terms of nested products that depend on $(-1)^{\binom{k}{l}}$ for some $l\in\NN$. Besides this, it will be very interesting to see if the obtained results (like, e.g., Theorem~\ref{Thm:ConstructIso}) can contribute to new algorithmic aspects of the Galois theory of difference equations~\cite{Singer:97,Singer:08}.

\section*{Acknowledgement}

\noindent I would like to thank the anonymous referees for their detailed comments and suggestions to improve the presentation of this article.


\end{document}